\documentclass[smallcondensed,a4paper,twoside,final]{svjour3}
\usepackage[english]{babel}
\usepackage[utf8]{inputenc}
\usepackage[hang,centerlast]{subfigure}
\usepackage{amsmath,amssymb,multirow,ifthen,enumitem,xspace,mathpartir,datetime,manfnt,stmaryrd}
\usepackage[clock,geometry]{ifsym}
\newcommand{\timespent}{\raisebox{-2pt}{\showclock{7}{25}}}
\usepackage{tikz}
\usetikzlibrary{automata,patterns,topaths,shapes,calc,arrows,decorations.pathreplacing}

\tikzstyle{every picture}+=[>=stealth,initial text=]
\tikzstyle{accepting}=[accepting by arrow]
\colorlet{irrelevant}{black!60}
\tikzstyle{interval duration}=[draw,dotted,thick,<->]
\tikzstyle{interval length}=[draw,>=angle 90,>-<]

\title{Interrupt Timed Automata: Verification and Expressiveness\thanks{Parts of this paper have been published in the proceedings of FoSSaCS'09~\cite{berard09} and Time'10~\cite{berard10}}}
\author{B\'eatrice B\'erard \and Serge Haddad \and Mathieu Sassolas}
\institute{B\'eatrice B\'erard \and Mathieu Sassolas \at Universit\'e Pierre \& Marie Curie, LIP6/MoVe, CNRS UMR 7606, Paris, France\\ \email{\{beatrice.berard $\mid$ mathieu.sassolas\}@lip6.fr} \and Serge Haddad \at \'Ecole Normale Sup\'erieure de Cachan, LSV, CNRS UMR 8643, INRIA, Cachan, France\\ \email{haddad@lsv.ens-cachan.fr}}
\date{\today}
\journalname{FMSD}
\titlerunning{Interrupt Timed Automata: Verification and Expressiveness}
\authorrunning{B. B\'erard, S. Haddad, M. Sassolas}

\newcommand{\rel}{\bowtie}
\newcommand{\N}{\ensuremath{\mathbb{N}}}

\newcommand{\Q}{\ensuremath{\mathbb{Q}}}
\newcommand{\R}{\ensuremath{\mathbb{R}}}

\newcommand{\A}{\mathcal{A}}
\newcommand{\Ba}{\mathcal{B}}
\newcommand{\C}{\mathcal{C}}
\newcommand{\M}{\mathcal{M}}
\newcommand{\I}{\mathcal{I}}
\newcommand{\T}{\mathcal{T}}

\newcommand{\G}{\mathcal{G}}

\newcommand{\La}{\mathcal{L}}
\newcommand{\U}{\mathcal{U}}
\newcommand{\tr}{\xrightarrow}
\newcommand{\fee}{\varphi}
\newcommand{\eps}{\varepsilon}
\newcommand{\vect}[1]{\mathbf{#1}}

\newcommand{\sem}[1]{\left\llbracket#1\right\rrbracket}

\newcommand{\etat}[2]{%
  \begin{tabular}{c}
    \(#1\) \\ 
    \(#2\)    
  \end{tabular}
}
\newcommand{\timedtrans}[3]{%
  \begin{tabular}{c}
    \(#1\) \\ 
    \(#2\) \\ 
    \(#3\) 
  \end{tabular}
}
\newcommand{\timedtransnoreset}[2]{%
  \begin{tabular}{c}
    \(#1\) \\ 
    \(#2\)    
  \end{tabular}
}
\newcommand{\until}{\ensuremath{{\sf \,U\,}}}
\newcommand{\since}{\ensuremath{{\sf \,S\,}}}
\newcommand{\untilsub}[1]{\ensuremath{{\sf \,U_{#1}\,}}}
\newcommand{\always}{\ensuremath{{\sf A\,}}}
\newcommand{\expath}{\ensuremath{{\sf E\,}}}
\newcommand{\lastoc}[1]{\ensuremath{{\mbox{\TriangleLeft}_{#1}}}}    
\newcommand{\nextoc}[1]{\ensuremath{{\mbox{\TriangleRight}_{#1}}}} 
\newcommand{\globally}{\ensuremath{{\sf G}}}
\newcommand{\eventually}{\ensuremath{{\sf F}}}

\newcommand{\ctl}{\textsf{CTL}\xspace}
\newcommand{\ctlstar}{\textsf{CTL}$^*$\xspace}
\newcommand{\tctl}{\textsf{TCTL}\xspace}

\newcommand{\tctlcint}{\textsf{TCTL}$_c^{\textrm{int}}$\xspace}
\newcommand{\tctlp}{\textsf{TCTL}$_p$\xspace}
\newcommand{\ltl}{\textsf{LTL}\xspace}
\newcommand{\mtl}{\textsf{MTL}\xspace}
\newcommand{\mitl}{\textsf{MITL}\xspace}
\newcommand{\scl}{\textsf{SCL}\xspace}

\newcommand{\propp}{\ensuremath{\mathsf{p}}}
\newcommand{\propq}{\ensuremath{\mathsf{q}}}
\newcommand{\proph}{\ensuremath{\mathsf{h}}}

\newcommand{\ruler}{\raisebox{-1pt}{%
\rule{0.1ex}{0.75em}\rule{1.5em}{0.1ex}\hspace{-1.5em}\rule[0.7em]{1.5em}{0.1ex}\rule{0.1ex}{0.75em}\hspace{-1.5em}%
\hspace{0.125em}{\rule{0.1ex}{0.9ex}}%
\hspace{0.25em}{\rule{0.1ex}{0.9ex}}%
\hspace{0.25em}{\rule{0.1ex}{1.25ex}}%
\hspace{0.25em}{\rule{0.1ex}{0.9ex}}%
\hspace{0.25em}{\rule{0.1ex}{0.9ex}}%
\hspace{0.5em}%
}}

\makeatletter
\def\subparagraph{\@startsection{subparagraph}{5}{\z@}%
    {-13pt plus-8pt minus-4pt}{\z@}{\normalsize\bfseries\itshape}}
\makeatother

\newcommand{\MS}[1]{\relax}
\newcommand{\BB}[1]{\relax}
\newcommand{\SH}[1]{\relax}

\begin{document}
\maketitle

\begin{abstract}
  We introduce the class of Interrupt Timed Automata (ITA), a subclass
  of hybrid automata well suited to the description of timed
  multi-task systems with interruptions in a single processor
  environment.

  While the reachability problem is undecidable for hybrid automata we
  show that it is decidable for ITA. More precisely we prove that the
  untimed language of an ITA is regular, by building a finite
  automaton as a generalized class graph.  We then establish that the
  reachability problem for ITA is in NEXPTIME and in PTIME when the
  number of clocks is fixed.  To prove the first result, we define a subclass
  ITA$_-$ of ITA, and show that (1) any ITA can be reduced to a
  language-equivalent automaton in ITA$_-$ and (2) the reachability
  problem in this subclass is in NEXPTIME (without any class graph).

  In the next step, we investigate the verification of real time
  properties over ITA. We prove that model checking \scl, a fragment
  of a timed linear time logic, is undecidable. On the other hand, we
  give model checking procedures for two fragments of timed branching
  time logic.
  
  We also compare the expressive power of classical timed automata and
  ITA and prove that the corresponding families of accepted languages
  are incomparable.  The result also holds for languages accepted by
  controlled real-time automata (CRTA), that extend timed automata. We
  finally combine ITA with CRTA, in a model which encompasses both
  classes and show that the reachability problem is still
  decidable. Additionally we show that the languages of ITA are
  neither closed under complementation nor under intersection.


\end{abstract}

\section{Introduction}

\subsection{Context}

The model of timed automata (TA), introduced in~\cite{alur94a}, has
proved very successful due to the decidability of several important
verification problems including reachability and model checking. A
timed automaton consists of a finite automaton equipped with real
valued variables, called clocks, which evolve synchronously with time,
during the sojourn in states.  When a discrete transition occurs,
clocks can be tested by guards, which compare their values with
constants, and reset. The decidability results were obtained through the
construction of a finite partition of the state space into regions,
leading to a finite graph which is time-abstract bisimilar to the
original transition system, thus preserving reachability.

Consider several tasks executing on a single processor (possibly
scheduled beforehand, although this step is beyond the scope of this
paper).  As a result, tasks are intertwined and may interrupt one
another~\cite{silberschatz08}. Since the behaviour of such systems may
depend on the current execution times of the tasks, a timed model
should measure these execution times, which involves clock suspension
in case of interruptions. Unfortunately, timed automata lack this
feature of clock suspension, hence more expressive models should be
considered.

Hybrid automata (HA) have subsequently been proposed as an extension
of timed automata~\cite{maler92}, with the aim to increase the
expressive power of the model. In this model, clocks are replaced by
variables which evolve according to a differential
equation. Furthermore, guards consist of more general constraints on
the variables and resets are extended into (possibly non
deterministic) updates. This model is very expressive, but
reachability is undecidable in HA.  The simpler model obtained by
allowing clocks to be stopped and resumed, stopwatch automata (SWA),
would be sufficient to model task interruptions in a processor.
However, reachability is also undecidable for SWA~\cite{cassez00}.
Many classes have been defined, between timed and hybrid automata, to
obtain the decidability of this problem.

Task automata~\cite{fersman07} and suspension
automata~\cite{mcmanis94} model explicitly the scheduling of
processes.  Some classes restrict the use of variation of clock rate
in hybrid automata to achieve decidability.  Examples of such classes
are systems with piece-wise constant derivatives~\cite{asarin95},
controlled real-time automata~\cite{Zielonka}. Guards may also be
restricted, as in multi-rate or rectangular automata~\cite{alur95},
some integration graphs~\cite{kesten99}, or polygonal hybrid
systems~\cite{asarin07}.  Restricting reset may also lead to
decidability as in the hybrid automata with strong
resets~\cite{bouyer08} or initialized stopwatch automata~\cite{hen98}.
O-minimal hybrid systems~\cite{lafferriere99,lafferriere01} provide
algebraic constraints on hybrid systems to yield decidability.
Extensions of timed automata to release some constraints were also
considered, as in some updatable timed automata~\cite{bouyer04}.

While untimed properties like reachability and
\ltl~\cite{pnueli77,sistla85} or \ctl model
checking~\cite{emerson82,queille82,Cim02}, are useful for such models,
real time verification consider more precise requirements, for
instance quantitative response time properties. Therefore, timed
extensions of these logics have been defined.  In the case of linear
time logics, verification of the most natural extension
\mtl~\cite{koymans90} is undecidable on TA.  However, several
decidable fragments such as \mitl~\cite{alur96} and
\scl~\cite{raskin97} have subsequently been defined.  In the case of
timed variants of branching time logics, different versions of Timed
\ctl (\tctl)~\cite{alur93,HNSY94} have been defined.  Model checking
procedures on TA for both versions of \tctl have been developed and
implemented in several tools~\cite{uppaal04,kronos98}.

\subsection{Contributions}

In this paper, we define a subclass of hybrid automata, called
Interrupt Timed Automata (ITA), well suited to the description of
multi-task systems with interruptions in a single processor
environment.

\paragraph{The ITA model.}
In an ITA, the finite set of control states is organized according to
\emph{interrupt levels}, ranging from $1$ to $n$, with exactly one
active clock for a given level. The clocks from lower levels are
suspended and those from higher levels are not yet defined (thus have
arbitrary value $0$). On the transitions, guards are linear
constraints using only clocks from the current level or the levels
below and the relevant clocks can be updated by linear expressions,
using clocks from lower levels.
Finally, each state has a policy (lazy, urgent or delayed) that rules
the sojourn time. This model is rather expressive since it combines
variables with rate $1$ or $0$ (usually called stopwatches) and linear
expressions for guards or updates.  The ITA model is formally defined
in Section~\ref{sec:background}.

\paragraph{Reachability problem.}
As said before, the reachability problem is undecidable for automata
with stopwatches~\cite{hen98,cassez00,brihaye06}.  However, we prove
that it is decidable for ITA.

More precisely, we first show that the untimed language of an ITA is
effectively regular (Section~\ref{sec:regular}). The corresponding
procedure significantly extends the classical region construction
of~\cite{alur94a} by associating with each state a family of orderings
over linear expressions. This construction yields a decision algorithm
for reachability in $2$-EXPTIME, and PTIME when the number of clocks
is fixed. This should be compared to TA with 3 clocks for which
reachability is PSPACE complete~\cite{courcoubetis92}.

We define a slight restriction of the model, namely ITA$_-$, which
forbids updates of clocks other than the one of the current level.  We
prove that for any ITA one can build an equivalent ITA$_-$ w.r.t.\
language equivalence, whose size is at most exponential w.r.t.\
the size of the ITA and polynomial when the number of clocks is fixed.
Based on the existence of a bound for the length of the minimal
reachability path, we then show that reachability on ITA$_-$ can be
decided in NEXPTIME without any class graph construction.
This yields a NEXPTIME procedure for reachability in ITA (Section~\ref{sec:reachcomplexity}).

\paragraph{Model checking over ITA.}
We then focus on the verification of real time properties for ITA
(Section~\ref{sec:modelchecking}), expressed in timed extensions of
\ltl and \ctl.

First we show that the model checking of timed (linear time) logic
\mitl~\cite{alur96} is undecidable.  Actually, even the fragment
\scl~\cite{raskin97} cannot be verified on ITA, while the
corresponding verification problem over TA is PSPACE-complete.

We then consider two fragments of the timed (branching time) logic
\tctl, introduced in~\cite{HNSY94} and also studied later from the
expressiveness point of view~\cite{bouyer05}.  The first one,
\tctlcint, contains formulas involving comparisons of model clocks as
atomic propositions. In this logic, it is possible to express
properties like: \emph{(P1) a safe state is reached before spending 3
  t.u. in handling some interruption}.  Decidability is obtained by a
generalized class graph construction in 2-EXPTIME (PTIME if the
number of clocks is fixed).  Since the corresponding fragment cannot
refer to global time, we consider a second fragment, \tctlp, in which
we can reason on minimal or maximal delays.  Properties like
\emph{(P2) the system is error free for at least 50 t.u.} or
\emph{(P3) the system will reach a safe state within 7 t.u.}  can be
expressed. In this case, the decidability procedure has a complexity
in NEXPTIME for the existential fragment and 2-EXPTIME for the
universal fragment (respectively NP and co-NP if the number of clocks
is fixed).

\paragraph{Expressiveness.}
We also study the expressive power of the class ITA
(Section~\ref{sec:exp}), in comparison with the original model of
timed automata and the more general controlled real-time automata
(CRTA) proposed in~\cite{Zielonka}. In CRTA, clocks and states are
colored and a time rate is associated with every state.  During the
visit of a state, all clocks colored by the color of the state evolve
with the state rate while the others do not evolve.  We prove that the
corresponding families of languages ITL and TL, as well as ITL and
CRTL, are incomparable.  Additionally we show that ITL is neither
closed under complementation nor under intersection.

\paragraph{Extensions.}
We finally investigate compositions of ITA and other timed models
(Section~\ref{sec:combination}).  In the first composition, a
synchronous product of an ITA and a TA, we prove that the reachability
problem becomes undecidable.  We then define a more appropriate
product of ITA and CRTA. The CRTA part describes a basic task at an
implicit additional level $0$.  For this extended model denoted by
ITA$^+$, we show that reachability is still decidable with the same
complexity and in PSPACE when the number of clocks is fixed.

\section{Interrupt Timed Automata}\label{sec:background}

\subsection{Notations}
The sets of natural, rational and real numbers are denoted
respectively by $\N$, $\Q$ and $\R$.  A \emph{timed word} over an
alphabet $\Sigma$ is a finite sequence $w=(a_1,\tau_1) \ldots
(a_n,\tau_n)$ where $a_i$ is in $\Sigma$ and $(\tau_i)_{1 \leq i \leq
  n}$ is a non-decreasing sequence of real numbers.  The \emph{length}
of $w$ is $n$ and the \emph{duration} of $w$ is $\tau_n$.

For a finite set $X$ of clocks, a linear expression over $X$ is a term
of the form $\sum_{x \in X} a_x \cdot x + b$ where $b$ and $(a_x)_{x
  \in X}$ are in $\Q$. We denote by $\C(X)$ the set of constraints
obtained by conjunctions of atomic propositions of the form $C \rel
0$, where $C$ is a linear expression over $X$ and $\rel \,\in
\{>,\geq,=,\leq,<\}$. The subset $\C_0(X)$ of $\C(X)$ contains
constraints of the form $x +b \rel 0$. An \emph{update} over $X$ is a
conjunction (over $X$) of assignments of the form $x := C_x$, where
$x$ is a clock and $C_x$ is a linear expression over $X$.  The set of
all updates over $X$ is written $\U(X)$, with $\U_0(X)$ for the subset
containing only assignments of the form $x := 0$ (reset) or of the
form $x := x$ (no update). For a linear expression $C$ and an update
$u$, the expression $C[u]$ is obtained by ``applying'' $u$ to $C$,
\textit{i.e.} substituting each $x$ by $C_x$ in $C$, if $x := C_x$ is
the update for $x$ in $u$. For instance, for the set of two clocks $X
= \{x_1, x_2\}$, expression $C= x_2 -2x_1 + 3$ and update $u$ defined
by $x_1 := 1 \wedge x_2 := 2x_1 +1$, applying $u$ to $C$ yields the
expression $C[u] = 2x_1 + 2$.

A clock valuation is a mapping $v : X \mapsto \R$, with $\vect{0}$ the
valuation where all clocks have value $0$.  The set of all clock
valuations is $\R^X$ and we write $v \models \fee$ when valuation $v$
satisfies the clock constraint $\fee \in \C(X)$. For a valuation $v$,
a linear expression $C$ and an update $u$, the value $v(C)$ is
obtained by replacing each $x$ in $C$ by $v(x)$ and the valuation
$v[u]$ is defined by $v[u](x) = v(C_x)$ for $x$ in $X$ if $x := C_x$
is the update for $x$ in $u$. Observe that an update is performed
simultaneously on all clocks. For instance, let $X = \{x_1, x_2,
x_3\}$ be a set of three clocks. For valuation $v = (2, 1.5, 3)$ and
update $u$ defined by $x_1 := 1 \wedge x_2 := x_2 \wedge x_3 := 3x_2 -
x_1$, applying $u$ to $v$ yields the valuation $v[u] = (1, 1.5, 2.5)$.
\subsection{Models of timed systems}

The model of ITA is based on the principle of multi-task systems with
interruptions, in a single processor environment.  We consider a set
of tasks with different priority levels, where a higher level task
represents an interruption for a lower level task. At a given level,
exactly one clock is active (rate $1$), while the clocks for tasks of
lower levels are suspended (rate $0$), and the clocks for tasks of
higher levels are not yet activated and thus contain value $0$. The
mechanism is illustrated in \figurename~\ref{fig:levels}, where
irrelevant clock values are greyed.  An example of such behavior can
be produced by the ITA depicted in \figurename~\ref{fig:italevels},
which describes a system that answer requests according to their
priority. It starts by receiving a request for a \emph{main} task of
priority $1$. The treatment of this task can be interrupted by tasks
of priority $2$ or $3$, depending on how far the system is in the
execution of the main task.  Tasks of priority $2$ and $3$ may
generate errors (modeled by an interruption of higher level), after
which the system recovers.  On this system, deciding if it is possible
-- or always the case -- that the main task is executed in less than a
certain amount of time would give an insight on the quality of service
of the system.  \MS{Reviewer 2, Long commentaire 4. \`A relire.}

\begin{figure}
\centering
\begin{tikzpicture}[scale=0.8,node distance=1cm,auto]
\node[anchor=east] (n1) at (0,0) {level 1};
\node (r1) at ($(n1) + (2,0)$) {$\forall i, x_i:=0$};
\node[anchor=east] (n2) [above of=n1] {level 2};
\node[anchor=east] (n3) [above of=n2] {level 3};
\node[anchor=east] (n4) [above of=n3] {level 4};
\node[inner sep=0pt] (p1) at ($(r1.east) + (1.5,0)$) {};

\path[->,draw=black,thick] (r1.east) -- (p1);
\path[->,draw=black,thick,dashed] (p1) -- (n3 -| p1);
\path[->,draw=black,thick] (n3 -| p1) -- ++(2.1,0) node[inner sep=0pt](p2) {};
\path[->,draw=black,thick,dashed] (p2) -- (n4 -| p2);
\path[->,draw=black,thick] (n4 -| p2) -- ++(1.7,0) node[inner sep=0pt](p3) {};
\path[->,draw=black,thick,dashed] (p3) -- (n1 -| p3);
\path[->,draw=black,thick] (n1 -| p3) -- ++(2.2,0) node[inner sep=2pt,anchor=west](p4) {\dots};
\node[text width=1.1cm,anchor=south east] at ($(n1 -| p3) + (0,-0.2)$) {\small $x_4 := 0$ $x_3 := 0$ $x_2 := 0$};

\tikzstyle{leg}=[text width=1cm, inner ysep=-10pt, inner xsep=-3pt]
\node[leg,anchor=east] (vleg) at (0,-1.5) {\parbox{\textwidth}{\[\left[\begin{array}{c} x_1 \\ x_2 \\ x_3 \\ x_4 \end{array}\right]\]}};
\node[leg,text width=0.75cm] (leg0) at (vleg -| r1.east)  {\parbox{\textwidth}{\[\left[\begin{array}{c} 0\\ \color{irrelevant}{0}\\ \color{irrelevant}{0}\\ \color{irrelevant}{0}\end{array}\right]\]}};
\node[leg] (leg1) at (vleg -| p1) {\parbox{\textwidth}{\[\left[\begin{array}{c} 1.5\\ \color{irrelevant}{0}\\\color{irrelevant}{0}\\ \color{irrelevant}{0}\end{array}\right]\]}};
\node[leg] (leg2) at (vleg -| p2) {\parbox{\textwidth}{\[\left[\begin{array}{c} 1.5\\ 0\\2.1\\\color{irrelevant}{0}\end{array}\right]\]}};
\node[leg] (leg3) at (vleg -| p3) {\parbox{\textwidth}{\[\left[\begin{array}{c} 1.5\\ 0\\ 2.1 \\ 1.7\end{array}\right]\]}};
\node[leg] (leg4) at (vleg -| p4.west) {\parbox{\textwidth}{\[\left[\begin{array}{c} 3.7\\ \color{irrelevant}{0}\\ \color{irrelevant}{0}\\ \color{irrelevant}{0}\end{array}\right]\]}};

\path[->] (leg0) edge node {$1.5$} (leg1);
\path[->] (leg1) edge node {$2.1$} (leg2);
\path[->] (leg2) edge node {$1.7$} (leg3);
\path[->] (leg3) edge node {$2.2$} (leg4);
\end{tikzpicture}
\caption{Interrupt levels and clocks in an ITA.}
\label{fig:levels}
\end{figure}
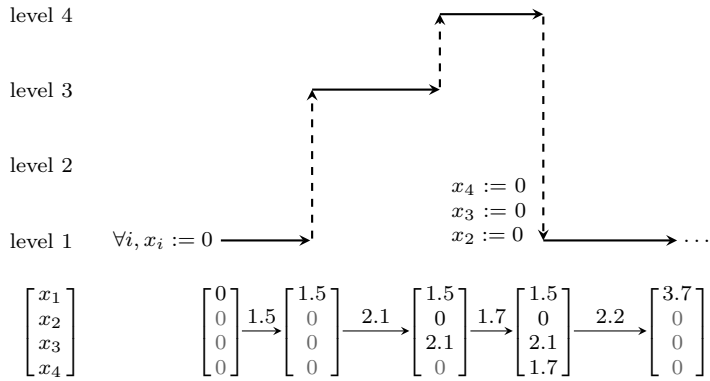

\begin{figure}
\centering
\begin{tikzpicture}[auto,node distance=4cm]
\node[state] (q0) at (0,0) {$q_1,1$};
\node[state,accepting,node distance=5cm,right of=q0] (q1) {$q_2,1$};
\node[state,initial,node distance=5cm,left of=q0] (qi) {$q_0,1$};
\node[state] (q2) at (0,2.5) {$q_3,2$};
\node[state] (q3) at (0,-2.5) {$q_4,3$};
\node[state,right of=q2] (q2err) {$q_5,4$};
\node[state,right of=q3] (q3err) {$q_6,4$};

\path[->] (q0) edge node {$3 \leq x_1 \leq 5$, \textit{answer\_prio1}} (q1);
\path[->] (qi) edge node {\textit{request\_prio1}, $x_1:=0$} (q0);
\path[->] (q0) edge[bend left] node {$x_1 \leq 1$, \textit{request\_prio2}} (q2);
\path[->] (q2) edge[bend left] node {$1 \leq x_2 \leq 2$, \textit{answer\_prio2}} (q0);
\path[->] (q0) edge[bend right,swap] node {$x_1 \leq 2$, \textit{request\_prio3}} (q3);
\path[->] (q3) edge[bend right,swap] node {$2 \leq x_3 \leq 3$, \textit{answer\_prio3}} (q0);
\path[->] (q2) edge[bend left=10] node {\textit{error}} (q2err);
\path[->] (q2err) edge[bend left=10] node {$x_4 \leq 2$, \textit{recover}} (q2);
\path[->] (q3) edge[bend left=10] node {\textit{error}} (q3err);
\path[->] (q3err) edge[bend left=10] node {$x_4 \leq 2$, \textit{recover}} (q3);

\end{tikzpicture}
%
%
%
\caption{An ITA that produces -- among others -- the behavior represented in \figurename~\ref{fig:levels}.}
\label{fig:italevels}
\end{figure}
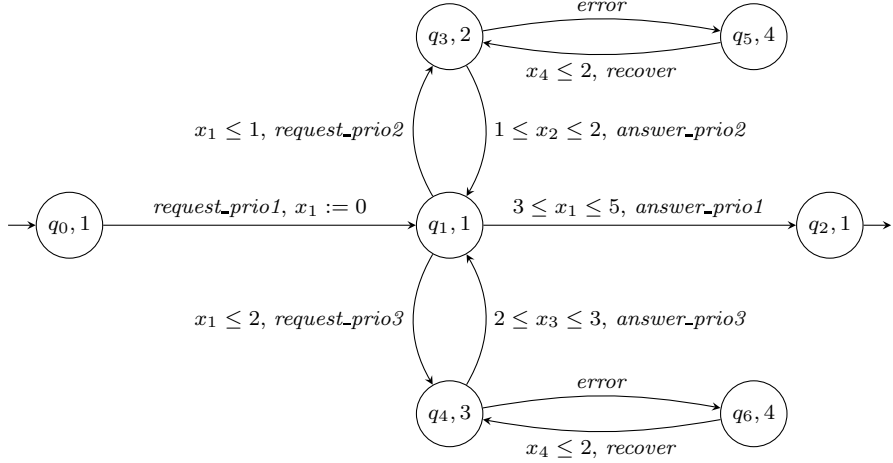

Enabling of a transition depends on the clocks valuation. The enabling
conditions, called \emph{guards}, are linear constraints on the clock
values of levels lower than or equal to the current level: the ones
that are relevant before the firing of the transition.  Additionally,
a transition can update the values of the clocks.  If the transition
decreases (resp. increases) the level, then each clock which is
relevant after (resp. before) the transition can either be left
unchanged or take a linear expression of clocks of strictly lower
level.

Along with its level, each state has a timing policy which indicates
whether time may (Lazy, default), may not (Urgent) or must (Delayed)
elapse in a state. Note that in TA, this kind of policy can be
enforced by an additional clock while this is not possible here
because there is a single clock per level.  This additional feature is
needed for the definition and further use of the model of ITA$_-$ (see
Section~\ref{sec:reachcomplexity}).  Note that the class graph
construction of Section~\ref{sec:regular} is still valid without them.
\MS{Reviewer 1, Rq 2. Done}

We also add a labeling of
states with atomic propositions, in view of interpreting logic
formulas on these automata.  In the sequel, the level of a transition
is the level of its source state. We also say that a transition is
lazy (resp. urgent, delayed) if the policy of its source state is lazy
(resp. urgent, delayed).

\begin{definition}
  An \emph{interrupt timed automaton} is a tuple $\A=\langle\Sigma,
  AP, Q, q_0, F,pol, X, \lambda,$ $lab, \Delta\rangle$, where:
\begin{itemize}
\item $\Sigma$ is a finite alphabet, $AP$ is a set of atomic
  propositions
	\item $Q$ is a finite set of states, $q_0$ is the initial
        state, $F \subseteq Q$ is the set of final states,
      \item $pol: Q \rightarrow \{Lazy,Urgent,Delayed\}$ is the timing
        policy of states,
	\item $X=\{x_1, \ldots, x_n\}$ consists of $n$ interrupt clocks,
	\item the mapping $\lambda : Q \rightarrow \{1, \ldots, n\}$
          associates with each state its level and we call
          $x_{\lambda(q)}$ the \emph{active clock} in state $q$. The
          mapping $lab : Q \rightarrow 2^{AP}$ labels each state with
          a subset of $AP$ of atomic propositions,
	\item $\Delta \subseteq Q \times \C(X) \times (\Sigma \cup
          \{\eps\}) \times \U(X) \times Q$ is the set of transitions.
          Let $q \tr{\fee, a, u} q'$ in $\Delta$ be a transition with
          $k=\lambda(q)$ and $k'=\lambda(q')$. The guard $\fee$ is a
          conjunction of constraints $\sum_{j=1}^k a_jx_j +b \rel 0$
          (involving only clocks from levels less than or equal to
          $k$). The update $u$ is of the form $\wedge_{i=1}^{n} x_i :=
          C_i$ with:
	\begin{itemize}
        \item if $k > k'$, \textit{i.e.} the transition decreases the
          level, then for $1 \leq i \leq k'$, $C_i$ is either of the
          form $\sum_{j=1}^{i-1} a_jx_j+b$ or $C_i=x_i$ (unchanged
          clock value) and for $i > k'$, $C_i=0$;
\MS{Reviewer 1, Rq 1. Done.}
        \item if $k \leq k'$ then for $1 \leq i \leq k$, $C_i$ is of
          the form $\sum_{j=1}^{i-1} a_jx_j +b$ or $C_i=x_i$, and for
          $i > k$, $C_i=0$.
	\end{itemize}
\end{itemize}
\end{definition}
A configuration $(q,v,\beta)$ of the associated transition system
consists of a state $q$ of the ITA, a clock valuation $v$ and a
boolean value $\beta$ expressing whether time has elapsed since the
last discrete transition.  This third component is needed to define
the semantics according to the policies.  \MS{Reviewer 1, Rq 5. Done.}
\begin{definition} 
\label{def:semantics}
The semantics of an ITA $\A$ is defined by the (timed) transition
system $\T_{\A}= (S, s_0, \rightarrow)$.  The set $S$ of
configurations is $\left\{\!(q,v,\beta) \mid q \in Q, \ v \in \R^X, \
  \beta \in \{\top,\bot\} \!\right\}\!$, with initial configuration
$s_0=(q_0, \vect{0},\bot)$. The relation $\rightarrow$ on $S$ consists
of two types of steps:
\begin{description}[font=\em]
\item[Time steps:] Only the active clock in a state can evolve, all
  other clocks are suspended.  For a state $q$ with active clock
  $x_{\lambda(q)}$, a time step of duration $d>0$ is defined by
  $(q,v,\beta) \tr{d} (q, v',\top)$ with
  $v'(x_{\lambda(q)})=v(x_{\lambda(q)})+ d$ and $v'(x)=v(x)$ for any
  other clock $x$. A time step of duration $0$ leaves the system
  $\T_{\A}$ in the same configuration.  When $pol(q)=Urgent$, only
  time steps of duration $0$ are allowed from $q$.
\item[Discrete steps:] A discrete step $(q, v,\beta) \tr{a} (q', v',\bot)$
  can occur if there exists a transition $q \tr {\fee, a, u} q'$ in
  $\Delta$ such that $v \models \fee$ and $v' = v[u]$.  When
  $pol(q)=Delayed$ and $\beta=\bot$, discrete steps are forbidden.
\end{description}
\end{definition}
The labeling function $lab$ is naturally extended to configurations by
$lab(q,v,\beta)= lab(q)$.

\smallskip An ITA $\A_1$ is depicted in \figurename~\ref{fig:exita1},
with two interrupt levels (and two interrupt clocks).  A geometric
view is given in figure~\ref{fig:traj}, with a possible trajectory:
first the value of $x_1$ increases from $0$ in state $q_0$ (horizontal
line) and, after transition $a$ occurs, its value is frozen in state
$q_1$ while $x_2$ increases (vertical line) until reaching the line
$x_2 = -\frac{1}{2}x_1 + \frac{1}{2}$. The light grey zone defined by $\left(0
  < x_1 < 1, \ 0 < x_2 < -\frac{1}{2}x_1 + \frac{1}{2}\right)$ corresponds to the set of valuations reachable in state $q_1$ and from which state $q_2$ is reachable.
\MS{Reviewer 1, Rq 6. Done.}

\begin{figure}[ht]
\centering
\subfigure[An ITA $\A_1$ with two interrupt levels]{\label{fig:exita1}
\begin{tikzpicture}[node distance=2.5cm,auto]
\node[state,initial] (q0) at (0,0) {$q_0,1$};
\node[state] (q1) [node distance=2.75cm,above right of=q0] {$q_1,2$};
\node[state,accepting] (q2) [right of=q1] {$q_2,2$};

\path[->] (q0) edge node {\timedtrans{x_1 < 1}{a}{(x_2:=0)}} (q1);
\path[->] (q1) edge node (tr) {\timedtransnoreset{x_1 + 2 x_2 = 2}{b}} (q2);
\end{tikzpicture}
}\hfill{~}
\subfigure[A possible trajectory in $\A_1$]{\label{fig:traj}
\begin{tikzpicture}[scale=1.5]
\path[draw=black,->] (0,0) -- (2.25,0) node[anchor=west] {$x_1$};
\path[draw=black,->] (0,0) -- (0,1.75) node[anchor=south east] {$x_2$};
\node[anchor=north east] at (0,0) {$0$};

\path[draw=black] (0,1) -- (2,0);
\path[draw=black] (1,0) -- (1,0.5);
\node[anchor=north] at (1,0) {$1$};
\node[anchor=north] at (2,0) {$2$};
\node[anchor=east] at (0,1) {$1$};

\node[anchor=north] (a) at (0.7,0) {\large$a$};
\node[anchor=south west] (b) at (0.7,0.65) {\large$b$};
\path[draw=black,very thick] (0,0) -- (a.north);
\path[fill=black!25] (0.05,0.05) -- (0.95,0.05) -- (0.95,0.475) -- (0.05,0.925) -- cycle;
\path[draw=black,very thick] (a.north) -- (b.south west);
\fill(a.north) circle (0.05);
\fill (b.south west) circle (0.05);
\end{tikzpicture}
}

\caption{An example of ITA and a possible execution.}
\label{fig:exita1traj}
\end{figure}
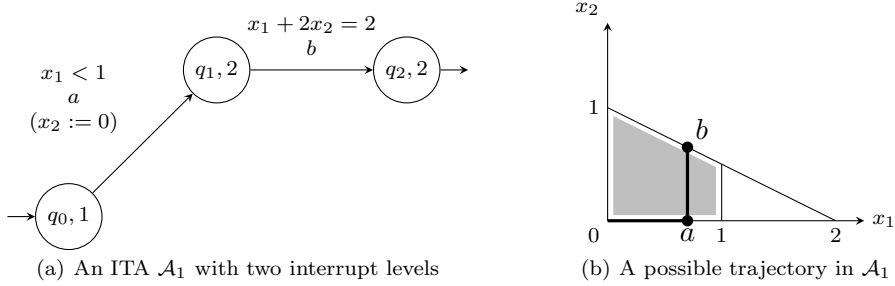

\bigskip
We now briefly recall the classical model of Timed Automata
(TA)~\cite{alur94a} as well as the model of Controlled Real-Time
Automata (CRTA)~\cite{Zielonka}. Note that in both models, timing
policies can be enforced by clock constraints.

\begin{definition}
  A \emph{timed automaton} is a tuple $\A=\langle\Sigma, Q, q_0, F, X,
  \Delta\rangle$, where $\Sigma$, $Q$, $q_0$, $F$ are defined as in an
  ITA, $X$ is a set of clocks and the set of transitions is $\Delta
  \subseteq Q \times \C_0(X) \times (\Sigma \cup \{\eps\}) \times
  \U_0(X) \times Q$, with guards in $\C_0(X)$ and updates in
  $\U_0(X)$.
\end{definition}

The semantics of a timed automaton is also defined as a timed
transition system, with the set $Q \times \R^X$ of configurations (no
additional boolean value). Discrete steps are similar to those of ITA
but in time steps, all clocks evolve with same rate $1$: $(q,v) \tr{d}
(q, v')$ iff for each clock $x$ in  $X$,  $v'(x) = v(x) + d$.

\bigskip Controlled Real-Time Automata extend TA with the
following features: the clocks and the states are partitioned
according to colors belonging to a set $\Omega$ and with every state
is associated a rational velocity.  When time elapses in a state, the
set of active clocks (i.e. with the color of the state) evolve with
rate equal to the velocity of the state while other clocks remain
unchanged.  For sake of clarity, we now propose a slightly
simplified version of CRTA.

\begin{definition}
\label{def:crta}
A CRTA $\A=(\Sigma, Q, q_0, F, X, up, low, vel, \lambda, \Delta)$ on a
finite set $\Omega$ of colors is defined by:
\begin{itemize}
\item $\Sigma$, the alphabet of actions,
\item $Q$, the set of states, with $q_0 \in Q$ the initial state and
  $F \subseteq Q$ the set of final states,
\item $X$ the set of clocks,
\item mappings $up$ and $low$ associate with each clock respectively 
an upper and a lower bound,
\item $vel: Q \mapsto \Q$ the velocity mapping,
\MS{Reviewer 2, Rq 1. Done}
\item $\lambda : X \uplus Q \mapsto \Omega$ the coloring mapping and
\item $\Delta \subseteq Q \times \C_0(X) \times (\Sigma \cup \{\eps\})
  \times \U_0(X) \times Q$ the set of transitions, with guards in
  $\C_0(X)$ and updates in $\U_0(X)$.
\end{itemize}
Moreover, the lower and upper bound mappings satisfy $low(x) \leq 0
\leq up(x)$ for each clock $x \in X$, and $low(x) \leq b \leq up(x)$
for each constant $b$ such that $x \rel b$ is a constraint in $\A$.
\end{definition}

The original semantics of CRTA is rather involved in order to obtain
decidability of the reachability problem.  It ensures that entering a
state $q$ in which clock $x$ is active, the following conditions on
the clock bounds hold : if $vel(q) > 0$ then $x \geq low(x)$ and if
$vel(q) < 0$ then $x \leq up(x)$. Instead (and equivalently) we add a
syntactical restriction which ensures this behavior.  For instance, if
a transition with guard $\fee$ and reset $u$ enters state $q$ with
$vel(q)<0$ and if $x$ is the only clock such that
$\lambda(x)=\lambda(q)$, then we replace this transition by two other
transitions: the first one has guard $\fee \wedge x > up(x)$ and adds
$x:=0$ to the reset condition $u$, the other has guard $\fee \wedge x
\leq up(x)$ and reset $u$. In the general case where $k$ clocks have
color $\lambda(q)$, this leads to $2^k$ transitions.  With this
syntactical condition, again the only difference from ITA concerns a
time step of duration $d$, defined by $(q,v) \tr{d} (q, v')$, with
$v'(x)=v(x)+ vel(q)d$ if $\lambda(x)=\lambda(q)$ and $v'(x)=v(x)$
otherwise.

\bigskip A run of an automaton $\A$ in ITA, TA or CRTA is a finite or
infinite path in the associated timed transition system $\T_{\A}$,
where (possibly null) time steps and discrete steps alternate. An
\emph{accepting run} is a finite run starting in $s_0$ and ending in a
configuration associated with a state of $F$. For such a run with
label $d_1 a_1 d_2 \ldots d_n a_n$, we say that the word $(a_1,d_1)
(a_2,d_1+d_2)\ldots (a_n,d_1+\cdots+d_n)$ (where $\eps$ actions are
removed) is accepted by $\A$. The set $\La(\A)$ contains the timed
words accepted by $\A$ and $Untimed(\La(\A))$, the untimed language of
$\A$, contains the projections onto $\Sigma^*$ of the timed words in
$\La(\A)$. Interrupt Timed Languages or
ITL (resp. Timed Languages or TL and Controlled Real-Time Languages or
CRTL) denote the family of timed languages accepted by an ITA (resp. a
TA and a CRTA).

For instance, the language $L_1$ accepted by the ITA $\A_1$ in
\figurename~\ref{fig:exita1} is \[L_1 = \La(\A_1) = \{ (a,\tau)(b, 1 +
\frac{\tau}{2}) \mid 0 \leq \tau <1 \}.\]

Languages of infinite timed words accepted by B\"uchi or Muller conditions could be studied 
but this analysis should address technical issues such as Zeno runs and infinite sequences of $\eps$-transitions.
\MS{Reviewer 1, Rq 7. Done.}

In the context of model-checking, we also consider \emph{maximal runs}
which are either infinite or such that no discrete step is possible
from the last configuration. The set of maximal runs starting from
configuration $s$ is denoted by $Exec(s)$. Since maximal runs can be
finite or infinite, we do not exclude Zeno behaviors. We use the
notion of (totally ordered) positions (which allow to consider several
discrete actions simultaneously) along a maximal run~\cite{HNSY94}:
for a run $\rho$, we denote by $<_{\rho}$ the strict order over
positions. For position $\pi$ along $\rho$, the corresponding
configuration is denoted by $s_{\pi}$, the prefix of $\rho$ up to
$\pi$ is written $\rho^{\leq \pi}$ and its duration,
$Dur\left(\rho^{\leq\pi}\right)$, is the sum of all delays along the
finite run $\rho^{\leq \pi}$.  Similarly, the suffix of $\rho$
starting from $\pi$ is denoted by $\rho^{\geq \pi}$.  For two
positions $\pi \leq_{\rho} \pi'$, the subrun of $\rho$ between these
positions is written $\rho_{[\pi,\pi']}$, its duration is
$Dur\left(\rho^{\leq \pi'}\right)-Dur\left(\rho^{\leq \pi}\right)$.
The length of $\rho$, denoted by $|\rho|$, is the number of discrete
transitions occurring in $\rho$.

\section{Regularity of untimed ITL}
\label{sec:regular}
We prove in this section that the untimed language of an ITA is
regular.  Similarly to TA (and to CRTA), the proof is based on the
construction of a (finite) class graph which is time abstract
bisimilar to the transition system $\T_{\A}$.
This result also holds for infinite words with standard B\"uchi conditions.
As a consequence, we obtain decidability of the reachability problem, as well as decidability for plain \ctlstar model-checking.
\MS{Reviewer 2, Rq 4 and 7. Done}

The construction
of classes is much more involved than in the case of TA. More
precisely, it depends on the expressions occurring in the guards and
updates of the automaton (while in TA it depends only on the maximal
constant occurring in the guards).  We associate with each state $q$ a
set of expressions $Exp(q)$ with the following meaning. The values of
clocks giving the same ordering of these expressions correspond to a
class.  In order to define $Exp(q)$, we first build a family of sets
\MS{Reviewer 2, Rq 2. Done}
$\{E_k\}_{1\leq k\leq n}$.  Then $Exp(q)= \bigcup_{k\leq \lambda(q)}
E_k$ (recall that $\lambda(q)$ is the index of the active clock in
state $q$).  Finally in Theorem~\ref{prop:reach} we show how to build
the class graph which proves the regularity of the untimed language.
This immediately yields a reachability procedure given in Proposition~\ref{prop:reach}.

\subsection{Construction of $\{E_k\}_{k\leq n}$}
\label{subsec:contructionexpressions}

\medskip We first introduce an operation, called \emph{normalization},
on expressions relative to some level. As explained in the
construction below, this operation will be used to order expression
values at a given level.

\begin{definition}[Normalization]
Let $C=\sum_{i\leq k}a_ix_i+b$ be an expression over $X_k= \{x_i
\mid i \leq k\}$, the \emph{$k$-normalization} of $C$, ${\tt
  norm}(C,k)$, is defined by:
\begin{itemize}
	\item if $a_k\neq 0$ 
then ${\tt norm}(C,k)=x_k+(1/a_k)(\sum_{i<k}a_ix_i+b)$;
	\item else ${\tt norm}(C,k)=C$. 
\end{itemize}
\end{definition}

Since guards are linear expressions with rational constants, we can
assume that in a guard $C \rel 0$ occurring in a transition outgoing
from a state $q$ with level $k$, the expression $C$ is either
$x_k+\sum_{i<k}a_ix_i+b$ (by $k$-normalizing the expression and if
necessary changing the comparison operator) or $\sum_{i<k}a_ix_i+b$.
It is thus written as $\alpha x_k+\sum_{i<k}a_ix_i+b$, with $\alpha
\in \{0,1\}$.

\bigskip The construction of $\{E_k\}_{k\leq n}$ proceeds top down
from level $n$ to level $1$ after initializing $E_k=\{x_k,0\}$ for all
$k$.  As we shall see below, when handling the level $k$, we add new
terms to $E_i$ for $1\leq i\leq k$.
These expressions are the ones needed to compute a (pre)order on the expressions in $E_k$.
\MS{Reviewer 1, Rq 10. Done}
\begin{itemize}
\item At level $k$, first for every expression $\alpha
x_k+\sum_{i<k}a_ix_i+b$ (with $\alpha \in \{0,1\}$) occurring in a
guard of an edge leaving a state of level $k$, we add
$-\sum_{i<k}a_ix_i-b$ to $E_k$.
\item Then we iterate the following procedure until no new term is
added to any $E_i$ for $1\leq i\leq k$.
	\begin{enumerate}
        \item Let $q \tr{\fee,a,u} q'$ with $\lambda(q)\geq k$ and
          $\lambda(q')\geq k$. Let $C \in E_{k}$, then we add $C[u]$
          to $E_{k}$ (recall that $C[u]$ is the expression obtained by
          applying update $u$ to $C$).
      \item Let $q \tr{\fee,a,u} q'$ with $\lambda(q) < k$ and
        $\lambda(q') \geq k$.  Let $C$ and $C'$ be two different
        expressions in $E_{k}$. We compute $C''={\tt
          norm}(C[u]-C'[u],\lambda(q))$, choosing an arbitrary order
        between $C$ and $C'$ in order to avoid redundancy. Let us
        write $C''$ as $\alpha
        x_{\lambda(q)}+\sum_{i<\lambda(q)}a_ix_i+b$ with $\alpha \in
        \{0,1\}$. Then we add $-\sum_{i<\lambda(q)}a_ix_i-b$ to
        $E_{\lambda(q)}$.
	\end{enumerate}
\end{itemize}

We illustrate this construction of expressions for the automaton
$\A_1$ of \figurename~\ref{fig:exita1}.
Initially, we have $E_1 = \{0,x_1\}$ and $E_2 = \{0,x_2\}$.
When treating level $2$, first, expression $-\frac12 x_1 + 1$ is added to $E_2$ as normalization of the guard $x_1 + 2x_2=2$.
Then transition labeled by $a$ updates $x_2$ (by reseting it to $0$).
As a result, we have to add to $E_1$ all differences of expressions of $E_2$ updated by $x_2 := 0$.
This only produces expression $-\frac12 x_1 +1 - 0$ which is normalized into $x_1 - 2$; thus expression $2$ is added to $E_1$.
When treating level $1$, expression $1$ from the guard of transition $a$ is added to $E_1$.
As a result, we obtain $E_1 = \{x_1, 0, 1, 2\}$ and $E_2 = \{x_2, 0, -\frac12 x_1 +1\}$.
\MS{Reviewer 1, Rq 9. Done}

\begin{lemma}
\label{prop:terminate}
The construction procedure of $\{E_k\}_{k\leq n}$ terminates and the
size of every $E_k$ is bounded by $(E+2)^{2^{n(n-k+1)}+1}$ where
$E$ is the size of the edges of the ITA.
\end{lemma}
\begin{proof}
  Given some $k$, we prove the termination of the stage relative to
  $k$. Observe that the second step only adds new expressions to
  $E_{k'}$ for $k'<k$. Thus the two steps can be ordered.  Let us
  prove the termination of the first step of the saturation
  procedure. We define $E_{k}^{0}$ as the set $E_k$ at the beginning
  of this stage and $E_{k}^{i}$ as this set after insertion of the
  $i^{th}$ item in it. With each added item $C[u]$ can be associated
  its \emph{father} $C$. Thus we can view $E_{k}$ as an increasing
  forest with finite degree (due to the finiteness of the edges) and
  finitely many roots. Assume that this step does not terminate. Then
  we have an infinite forest and by K\"onig lemma, it has an infinite
  branch $C_0,C_1,\ldots$ where $C_{i+1}=C_i[u_i]$ for some update
  $u_i$ such that $C_{i+1}\neq C_i$. Observe that the number of
  updates that change the variable $x_k$ is either 0 or 1 since once
  $x_k$ disappears it cannot appear again.  We split the branch into
  two parts before and after this update or we still consider the
  whole branch if there is no such update. In these (sub)branches, we
  conclude with the same reasoning that there is at most one update
  that change the variable $x_{k-1}$. Iterating this process, we
  conclude that the number of updates is at most $2^k-1$ and the
  length of the branch is at most $2^k$.

  For the sake of readability, we set $B=E+2$. The final size of
  $E_{k}$ is thus at most $E_{k}^{0}\times B^{2^k}$ since the width of
  the forest is bounded by $B$.

  In the second step, we add at most $B \times
  (|E_{k}|\times(|E_{k}|-1))/2$ to $E_{i}$ for every $i<k$. This
  concludes the proof of termination.

\bigskip
We now prove by a painful backward induction that as soon as $n\geq 2$,
$|E_{k}|\leq B^{2^{n(n-k+1)}+1}$.
The doubly exponential size of $E_n$ (proved above) is propagated downwards by the saturation procedure.
\MS{Reviewer 1, Rq 11. Done}
We define $p_k = |E_{k}|$.

\paragraph{Basis case $k=n$.}
We have $p_n \leq p_n^{0}\times B^{2^n}$ where $p_n^{0}$ 
is the number of guards of the outgoing edges from states of level $n$.
Thus $p_n \leq B\times B^{2^n}= B^{2^n+1}= B^{2^{n(n-n+1)}+1}$
which is the claimed bound.

\paragraph{Inductive case.}
Assume that the bound holds for $k < j \leq n$.  Due to all executions
of the second step of the procedure at strictly higher levels, $p_k^0$
expressions were added to $E_k$, with:
\begin{eqnarray*}
p_k^{0} &\leq& B + B \times ((p_{k+1}\times (p_{k+1}-1))/2 + \cdots + (p_{n}\times (p_{n}-1))/2)\\
p_k^{0} &\leq& B + B \times (B^{2^{n(n-k)+1}+2} + \cdots + B^{2^{n+1}+2})\\
p_k^{0} &\leq& B \times (n-k+1) \times B^{2^{n(n-k)+1}+2}\\
p_k^{0} &\leq& B \times B^n \times B^{2^{n(n-k)+1}+2} \quad\textrm{(here we use } B \geq 2\textrm)\\
p_k^{0} &\leq& B^{2^{n(n-k)+1}+n+3}
\end{eqnarray*}
Taking into account the first step of the procedure for level $k$, we
have: \[p_k \leq B^{2^{n(n-k)+1}+2^k+n+3}.\] Let us consider the term
$\delta = 2^{n(n-k+1)}+1-(2^{n(n-k)+1}+2^k+n+3)$. Since $k < n$,
\begin{eqnarray*}
\delta &\geq& (2^{n-1}-1)2^{n(n-k)+1}-(2^k+n+2)\\
\delta &\geq& (2^{n-1}-1)2^{n(n-k)+1}-(2^{n-1}+2^n)\\
\delta &\geq& (2^{n-1}-1)2^{n(n-k)+1}-2^{n+1} \geq 0\\
\end{eqnarray*}
Thus $p_k \leq B^{2^{n(n-k)+1}+2^k+n+3} \leq B^{2^{n(n-k+1)}+1} =
(E+2)^{2^{n(n-k+1)}+1}$ which is the claimed bound.  \qed
\end{proof}

\subsection{Construction of the class automaton}
\label{subsec:contructiongraph}

In order to analyze the size of the class automaton
defined below, we recall and adapt a classical
result about partitions of $n$-dimensional Euclidian
spaces.

\begin{definition}
Let $\{H_k\}_{1\leq k \leq m}$ be a family
of hyperplanes of $\R^n$. A \emph{region} defined by this
family is a connected component of $\R^n \setminus \bigcup_{1\leq k \leq m} H_k$.
An \emph{extended region} defined by this family is
a connected component of $\bigcap_{k \in I} H_k \setminus \bigcup_{k \notin I} H_k$
where $I \subseteq \{1,\ldots, m\}$.  
\end{definition}

\begin{proposition}[\cite{zas75}]
The number of regions defined by the family $\{H_k\}_{1\leq k \leq m}$
is at most $\sum_{i=0}^n \binom{m}{i}$.
\end{proposition}

We derive from this proposition:
\begin{corollary}
\label{cor:zas}
The number of extended regions defined by the family $\{H_k\}_{1\leq k \leq m}$
is at most $\sum_{p=0}^n\binom{m}{p}\sum_{i=0}^{n-p} \binom{m-p}{i}\leq e^2m^n$.
\end{corollary}
\begin{proof}
Observe that an extended region is a region belonging to an intersection of at most $n$ hyperplanes
(by removing redundant hyperplanes). Thus counting the number of such intersections and
applying the previous proposition yields the following formula:
\[ \sum_{p=0}^n\binom{m}{p}\sum_{i=0}^{n-p} \binom{m-p}{i}\leq \sum_{p=0}^n\frac{m^p}{p!}\sum_{i=0}^{n-p}\frac{m^{n-p}}{i!}=
m^n  \sum_{p=0}^n\frac{1}{p!}\sum_{i=0}^{n-p}\frac{1}{i!}\leq e^2m^n\]
\qed
\end{proof}

\begin{theorem}
\label{prop:reach}
The untimed language of an ITA is regular.
\end{theorem}

\begin{proof}
First, we assume that the policy of every state
is lazy. At the end of the proof, we explain how to adapt
the construction for states with urgent or delayed policies.

\paragraph{Class definition.}
Let $\A$ be an ITA with $E$ transitions and $n$ clocks,
the decision algorithm is based on the construction of a (finite)
class graph which is time abstract bisimilar to the transition
system $\T_{\A}$.  A class is a syntactical representation of a
subset of reachable configurations.  More precisely, it is defined
as a pair $R=(q,\{\preceq_k\}_{1 \leq k \leq \lambda(q)})$ where $q$
is a state and $\preceq_k$ is a total preorder over $E_k$, for $1 \leq k \leq \lambda(q)$.

The class $R$ describes the set of configurations:
\[\sem{R}=\{(q,v,\beta) \mid \beta \in \{\top,\bot\},\; \forall k \leq \lambda(q)\ 
\forall (g,h) \in E_k,\ g[v] \leq h[v] \mbox{~iff~} g \preceq_k h\}\]

The initial state of this graph is defined by the
class $R_0$ with $\sem{R_0}$ containing $(q_0,{\bf 0},\bot)$ which can be
straightforwardly determined.
For example, for ITA $\A_1$ of \figurename~\ref{fig:exita1}, the initial class is $R_0=(q_0, Z_0)$ with $Z_0: x_1=0 < 1 < 2$.
The final states are all 
$R=\left(q,\{\preceq_k\}_{1 \leq  k \leq \lambda(q)}\right)$ with $q \in F$.


Observe that fixing a state, the set of configurations $\sem{R}$ of a non empty class $R$
is exactly an extended region associated with the hyperplanes defined
by the comparison of two expressions of some $E_k$.
Since $(E+2)^{2^{n^2}+1}$ is an upper bound of the number of
expressions of any level, $m=(E+2)^{2^{n^2+1}+2}$
is an upper bound of the number of hyperplanes.
So using Corollary~\ref{cor:zas},
the number of semantically different classes for a given state is bounded by:
$$e^2m^n=e^2(E+2)^{2^{n^2+1}n+2n}$$
Since one can test semantical equality between classes
in polynomial time w.r.t. their size~\cite{RoTeVi97},
we implicitely consider in the sequel of the proof classes modulo
the semantical equivalence.

As usual, there are two kinds of transitions in the graph,
corresponding to discrete steps and time steps.

\paragraph{Discrete step.}
Let $R=(q,\{\preceq_k\}_{1 \leq k \leq \lambda(q)})$ and
$R'=(q',\{\preceq'_k\}_{1 \leq k \leq \lambda(q')})$ be two classes. There is a
transition $R \tr{e} R'$ for a transition $e : q \tr{\fee,a,u} q'$ if
there is some $(q,v) \in\, \sem{R}$ and $(q',v') \in\, \sem{R'}$ such that
$(q,v) \tr{e} (q',v')$. In this case, for all $(q,v) \in\, \sem{R}$
there is a $(q',v') \in\, \sem{R'}$ such that $(q,v) \tr{e}
(q',v')$. This can be decided as follows.

\subparagraph{Firability condition.}  
Write $\fee=\bigwedge_{j \in J}
C_j \bowtie_j 0$.  Since we assumed normalized guards, for every $j$, $C_j=\alpha
x_k+\sum_{i<k}a_ix_i+b$ (with $\alpha \in
\{0,1\}$ and $k = \lambda(q)$). By construction $C'_j=-\sum_{i<\lambda(q)}a_ix_i-b \in
E_k$. For each $j \in J$, we define a condition depending
on $\bowtie_j$. For instance, if $C_j \leq 0$, we require that
\MS{Reviewer 2, Rq 3. Done}
$\alpha x_k \preceq_k C'_j$, or if $C_j > 0$ we require that $\alpha x_k \npreceq_k C_j' \wedge C_j' \preceq \alpha x_k$.

\subparagraph{Successor definition.}
$R'$ is defined as follows. Let $k \leq \lambda(q')$ and $g',h' \in E_k$.
\begin{enumerate}
\item Either $k \leq \lambda(q)$, by construction, $g'[u],h'[u] \in
E_k$ then $g' \preceq'_k h'$ iff $g'[u] \preceq_k h'[u]$.
\item Or $k > \lambda(q)$, let
$D=g'[u]-h'[u]=\sum_{i\leq\lambda(q)}c_ix_i+d$, and
$C={\tt norm}(D,\lambda(q))$, and write $C=\alpha
x_{\lambda(q)}+\sum_{i<\lambda(q)}a_ix_i+b$ (with $\alpha \in
\{0,1\}$).
By construction $C'=-\sum_{i<\lambda(q)}a_ix_i-b \in E_{\lambda(q)}$.\\
When $c_{\lambda(q)}\geq 0$ then $g' \preceq'_k h'$ iff
$\alpha
x_{\lambda(q)} \preceq_{\lambda(q)} C' $.\\
When $c_{\lambda(q)}< 0$ then $g' \preceq'_k h'$ iff $C' \preceq_{\lambda(q)} \alpha x_{\lambda(q)}$.

\end{enumerate}
By definition of $\sem{\,\cdot\,}$,
\begin{itemize}
\item For any $(q,v) \in \sem{R}$, if there exists $(q,v) \tr{e}
(q',v')$ then the firability condition is fulfilled and $(q',v')$
belongs to $\sem{R'}$.
\item If the firability condition is fulfilled then for each $(q,v) \in
\sem{R}$ there exists $(q',v') \in \;\sem{R'}$ such that $(q,v)
\tr{e} (q',v')$.
\end{itemize}

\paragraph{Time step.}
Let $R=(q,\{\preceq_k\}_{1 \leq  k \leq \lambda(q)})$.
There is a transition $R \tr{succ} Post(R)$ for
$Post(R)=(q,\{\preceq'_k\}_{1 \leq k \leq \lambda(q)})$, the time
successor of $R$, which is defined as follows.

For every $i < \lambda(q)\ \preceq'_i=\preceq_i$.  Let $\sim$ be the
equivalence relation $\preceq_{\lambda(q)} \cap
\preceq^{-1}_{\lambda(q)}$ induced by the preorder. On equivalence
classes, this (total) preorder becomes a (total) order.  Let $V$ be
the equivalence class containing $x_{\lambda(q)}$.
\begin{enumerate}
\item Either $V=\left\{x_{\lambda(q)}\right\}$ and it is the greatest
equivalence class. Then $\preceq'_{\lambda(q)}=\preceq_{\lambda(q)}$
(thus $Post(R)=R$).
\item Either $V=\left\{x_{\lambda(q)}\right\}$ and it is not the greatest
equivalence class.  Let $V'$ be the next equivalence class. Then
$\preceq'_{\lambda(q)}$ is obtained by merging $V$ and $V'$, and
preserving $\preceq_{\lambda(q)}$ elsewhere.
\item Either $V$ is not a singleton. Then we split $V$ into
$V\setminus \left\{x_{\lambda(q)}\right\}$ and $\left\{x_{\lambda(q)}\right\}$ and
``extend'' $\preceq_{\lambda(q)}$ by $V \setminus \left\{x_{\lambda(q)}\right\}
\preceq'_{\lambda(q)} \left\{x_{\lambda(q)}\right\}$.
\end{enumerate}
By definition of $\sem{\,\cdot\,}$, for each $(q,v) \in \sem{R}$,
there exists $d>0$ such that $(q,v+d) \in \sem{Post(R)}$ and for each $d$
with $0 \leq d' \leq d$, then $(q,v+d') \in \sem{R} \cup \sem{Post(R)}$.

\bigskip We now explain how the policy is handled. Given a state $q$
such that $pol(q)=U$, for every class $R=(q,\{\preceq_k\}_{1 \leq k
  \leq \lambda(q)})$ we delete the time steps outgoing from $R$. The
case of a state $q$ such that $pol(q)=D$, is a little bit more
involved.  First we partition classes between \emph{time open}
classes, where for every every configuration of the class there exists
a small amount of time elapse that let the new configuration in the
same class, and \emph{time closed} classes.  The partition is
performed w.r.t. the equivalence class $V$ of $x_{\lambda(q)}$ for the
relation $\sim$ (see above in the proof). The class $R$ is time open
iff $V=\{x_{\lambda(q)}\}$. Then we successively replace every time
closed class $R$ by two copies $R^-$ and $R^+$, which capture wether
time has elapsed since the last last discrete step. Thus, a time edge
entering $R$ is redirected towards $R^+$ while a discrete edge
entering $R$ is redirected towards $R^-$. A time step $R \tr{succ} R'$
is replaced by two transitions $R^- \tr{succ} R'$ and $R^+ \tr{succ}
R'$, while a discrete step $R \tr{e} R'$ is replaced by the transition
$R^+ \tr{e} R'$. Time open classes allow time elapsing, hence no
splitting is required for these classes.

Since there is at most one time edge outgoing from a class, the number
of edges of the new graph is at most twice the number of edges in the
original graph.  \qed
\end{proof}

\begin{proposition}
\label{prop:reachita}
The reachability problem for Interrupt Timed Automata is decidable and belongs to
\emph{2-EXPTIME} and \emph{PTIME} when the number of clocks is fixed.
\end{proposition}
\begin{proof}
The reachability problem is solved
by building the class graph and applying standard reachability algorithm.
Since the number of semantically different classes is at most doubly
exponential in the size of the model and the semantical equivalence 
can be checked in polynomial time w.r.t. the size of the class (also doubly exponential)
this leads to a 2-EXPTIME complexity.
When the number of clocks is fixed the size of the graph
is at most polynomial w.r.t. the size of the problem leading
to a PTIME procedure. No complexity gain can be obtained
by a non deterministic search without building the graph
since the size of the graph is only polynomial w.r.t. the size
of a class.
\qed
\end{proof}

\noindent \textbf{Remarks.} This result should be contrasted with the
similar one for TA. The reachability problem for TA is PSPACE-complete
and thus less costly to solve than for ITA. However, fixing the
number of clocks does not reduce the complexity for TA (when this
number is greater than or equal to $3$) while this problem belongs now
to PTIME for ITA. Summarizing, the main source of complexity for ITA is
the number of clocks, while in TA it is the binary encoding of the
constants~\cite{courcoubetis92}.
\MS{Reviewer 1, rq 14. Done.}

Since the construction of the graph depends on a set of expressions, there is no notion of \emph{granularity} as in Timed Automata.
\MS{Reviewer 2, Long commentaire 1. Done.}
When the only guards are comparisons to constants and the only updates resets of clocks (as in Timed Automata), 
the abstraction obtained is coarser than the region abstraction of~\cite{alur94a}: it consists only in products of intervals.
\MS{Reviewer 1, Rqs 4 and 8. Done.}

\subsection{Example}
We illustrate this construction of a class automaton for the automaton
$\A_1$ of \figurename~\ref{fig:exita1}. The resulting class automaton is depicted on  \figurename~\ref{fig:regaut}, where
dashed lines indicate time steps.

\begin{figure}[ht]
\centering
\begin{tikzpicture}[node distance=3.75cm,auto]
\tikzstyle{every state}=[inner sep=2pt,draw=black,shape=rectangle,rounded corners=5pt]
\tikzstyle{time step}=[draw,->,dashed]
\tikzstyle{time succ}=[node distance=1.3cm]

\node[state, initial] (r0) at (0,0) {$R_0$};
\node[state] (r01) [node distance=1.5cm,below of=r0] {$R_0^1$};
\node[state] (r02) [time succ,below of=r01] {$R_0^2$};
\node (r04) [time succ, below of=r02] {\raisebox{4pt}{$\vdots$}};
\node[state] (r05) [time succ,below of=r04] {$R_0^5$};

\node[state] (r1) [node distance=2cm,right of=r0] {$R_1$};
\node[state] (r2) [node distance=3cm,right of=r1] {$\begin{array}{cc} q_1, Z_0 \\ 0 < x_2 < 1 \end{array}$};
\node[state] (r3) [right of=r2] {$\begin{array}{cc} q_1, Z_0 \\ 0 < x_2 = 1 \end{array}$};
\node[state,accepting] (r4) [node distance=2cm,below of=r3] {$\begin{array}{cc} q_2, Z_0 \\ 0 < x_2 = 1 \end{array}$};
\node[state,accepting] (r5) [node distance=2cm,below of=r4] {$\begin{array}{cc} q_2, Z_0 \\ 0 < 1 < x_2 \end{array}$};

\node[state] (r11) [node distance=2cm,right of=r01] {$R_1^1$};
\node[state] (r12) [node distance=3cm,right of=r11] {$\begin{array}{cc} q_1, Z_0^1 \\ 0 < x_2 <  -\frac{1}{2}x_1 +1 \end{array}$};
\node[state] (r13) [node distance=2cm,below of=r12] {$\begin{array}{cc} q_1, Z_0^1 \\ 0 < x_2 =  -\frac{1}{2}x_1 +1 \end{array}$};
\node[state,accepting,accepting where=left] (r14) [node distance=2cm,below of=r13] {$\begin{array}{cc} q_2, Z_0^1 \\ 0 < x_2 =  -\frac{1}{2}x_1 +1 \end{array}$};
\node[state,accepting] (r15) [right of=r14] {$\begin{array}{cc} q_2, Z_0^1 \\ 0 < -\frac{1}{2}x_1 +1 < x_2 \end{array}$};

\path[->] (r01) edge node {$a$} (r11);
\path[time step] (r11) -- (r12);
\path[time step] (r12) -- (r13);
\path[->] (r13) edge node {$b$} (r14);
\path[time step] (r14) -- (r15);

\path[time step] (r0) -- (r01);
\path[time step] (r01) -- (r02);
\path[time step] (r02) -- (r04);
\path[time step] (r04) -- (r05);

\path[->] (r0) edge node {$a$} (r1);
\path[time step] (r1) -- (r2);
\path[time step] (r2) -- (r3);
\path[->] (r3) edge node {$b$} (r4);
\path[time step] (r4) -- (r5);
\end{tikzpicture}
\caption{The class automaton for $\A_1$}
\label{fig:regaut}
\end{figure}
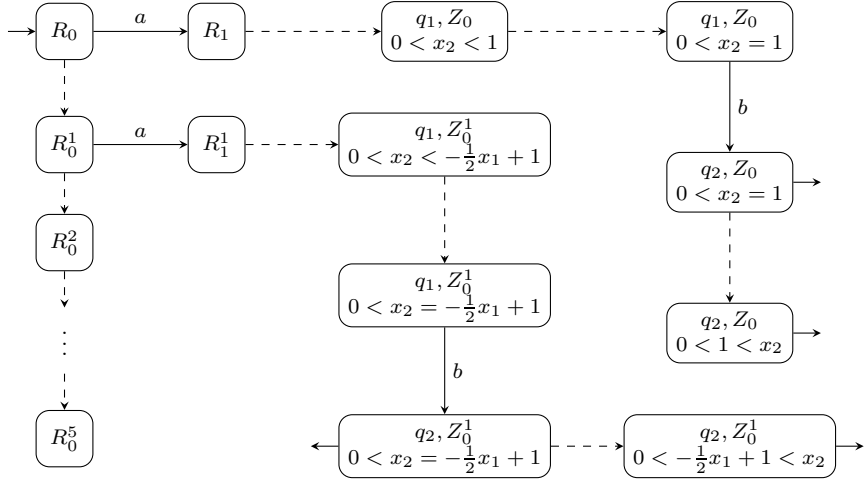

Recall that we obtained $E_1= \{ x_1, 0, 1, 2 \}$ and $E_2= \left\{
x_2, 0, -\frac{1}{2}x_1 + 1\right\}$.  In state $q_0$, the only
relevant clock is $x_1$ and the initial class is $R_0=(q_0, Z_0)$ with
$Z_0: x_1=0 < 1 < 2$. Its time successor is $R_0^1=(q_0, Z_0^1)$ with
$Z_0^1: 0 < x_1 < 1 < 2$.  Transition $a$ leading to $q_1$ can be taken
from both classes, but not from the next time successors $R^2_0=(q_0,
0 < x_1= 1 < 2)$, $R^3_0=(q_0, 0 < 1 < x_1 < 2)$, $R^4_0=(q_0, 0 < 1 < x_1 = 2)$, or $R^5_0=(q_0, 0 < 1 < 2 < x_1)$.

Transition $a$ switches from $R_0$ to $R_1=\left(q_1,Z_0,x_2 = 0 <
1\right)$, because $x_1=0$, and from $R_0^1$ to $R^1_1=\left(q_1,
Z_0^1, x_2=0 < -\frac{1}{2}x_1 + 1\right)$. Transition $b$ is
fired from those time successors for which $x_2= -\frac{1}{2}x_1 + 1$.

On the geometric view of figure~\ref{fig:traj}, the displayed
trajectory corresponds to the following path in the class automaton:
\begin{eqnarray*}
&&R_0 \rightarrow R_0^1 \xrightarrow{a} R_1^1 \rightarrow \left(q_1, Z_0^1,0<x_2<-\frac12 x_1 + 1\right) \rightarrow \left(q_1, Z_0^1,0<x_2=-\frac12 x_1 + 1\right)\\
&&\xrightarrow{b} \left(q_2, Z_0^1,0<x_2=-\frac12 x_1 + 1\right)
\end{eqnarray*}


\section{A simpler model}\label{sec:reachcomplexity}

\subsection{Definition of ITA$_-$}

We introduce a restricted version of ITA, called ITA$_-$, which is
interesting both from a theoretical and a practical point of
view. When modeling interruptions in real-time systems, the clock
associated with some level measures the time spent in this level or
more generally the time spent by some tasks at this level. Thus when
going to a higher level, this clock is not updated until returning to
this level. The ITA$_-$ model takes this feature into
account. Moreover, it turns out that the reachability problem for
ITA$_-$ can be solved more efficiently.
This also provides a better complexity upper-bound for the reachability problem on ITA (in the general case).

\begin{definition}
The subclass ITA$_-$ of ITA is defined by the following restriction
on updates.  For a transition $q \tr{\fee, a, u} q'$ of an automaton
$\A$ in ITA$_-$ (with $k=\lambda(q)$ and $k'=\lambda(q')$), the update
$u$ is of the form $\wedge_{i=1}^{n} x_i := C_i$ with:
\begin{itemize}
\item if $k > k'$, then for $1 \leq i \leq k'$, $C_i := x_i$ and for $k'+1 \leq i \leq n$, $C_i = 0$,
\MS{Reviewer 1, Rq 14. Done.}
\textit{i.e.} the only updates are the resets of now irrelevant clocks;
\item if $k \leq k'$ then 
$C_k$ is of the form $\sum_{j=1}^{k-1} a_jx_j +b$ or $C_k=x_k$.
For $k < i \leq k'$, $C_i=0$ and $C_i=x_i$ otherwise.
\end{itemize}
\end{definition}

Thus, complex updates appear only in transitions increasing the level,
and only for the active clock of the transition level. 

 The proof of the following result is based on
 Propositions~\ref{proposition:itaitamoins}
 and~\ref{proposition:reachitamoins} proved in the next two sections.
 
 \begin{theorem}\label{thm:optreach}
 The reachability problem for $ITA$ belongs to \emph{NEXPTIME}.
 \end{theorem}
 \begin{proof}
   Given an ITA $\mathcal A$ with transitions of size $E$ and constants coded over $b$ bits, we build the ITA$_-$
   $\mathcal A'$ of Proposition~\ref{proposition:itaitamoins}. Then we
   apply on $\mathcal A'$ the reachability procedure of
   Proposition~\ref{proposition:reachitamoins}.  In this procedure, we
   consider paths of length bounded by $(E'+n)^{3n}$, where $E'$ is
   the number of transitions of $\mathcal A'$.  Since $E' \leq
   2^{4b \cdot E \cdot n^2}$ (as shown in the proof of
   Proposition~\ref{proposition:itaitamoins}), the length of the paths
   considered is bounded by
 \[(E'+n)^{3n} \leq \left(2^{4b \cdot E \cdot n^2}+n\right)^{3n} \leq (n+2)^{12 b \cdot E \cdot n^3}\]
 which establishes the claimed upper  bounds.  \qed
 \end{proof}

\subsection{From ITA to ITA$_-$}
\label{sec:relexp}

In this subsection we prove that ITA and ITA$_-$ are equivalent
w.r.t.\ the associated (timed) languages.

\begin{proposition}
\label{proposition:itaitamoins}
Given an ITA $\mathcal A$, we build an automaton $\mathcal A'$ in
ITA$_-$ accepting the same timed language and with the same clocks
such that its number of edges (resp. states) is exponential w.r.t. the
number of edges (resp. states) in $\mathcal A$ and polynomial when the
number of clocks is fixed.
\end{proposition}
\begin{proof}
  Starting from ITA $\A=\langle\Sigma, AP, Q, q_0, F,pol, X, \lambda,$
  $lab, \Delta\rangle$, the construction of automaton $\mathcal A'$
  relies on memorizing at a given level $i$, for every clock $x_{j}$
  at a lower level, an expression depending on $x_1,\ldots,x_{j-1}$,
  corresponding to the delayed update of $x_{j}$. This expression is
  used later to replace the value of $x_j$ in guards and to restore
  its correct value by update after decreasing to level $j$.

To this aim we associate with every pair of levels $i \geq j$, a set
of expressions $F_{i,j}$ inductively defined by:
\begin{itemize}
	\item $F_{i,i}=\{x_i\}$
	\item $\forall i>j\ F_{i,j}=F_{i-1,j} \cup \{e[\{x_k
        \leftarrow e_k\}_{k<j}] \mid e$ is the expression of an
        update of $x_j$ by an edge of level $i$ and $\forall k,\ e_k \in
        F_{i,k}\}$
\end{itemize}
We write $F_j = F_{n,j} = \bigcup_{i=j}^n F_{i,j}$.  The set $F_j$
thus contains all expressions of updates of $x_j$ that appear at
higher levels.

Although the number of expressions is syntactically doubly exponential
w.r.t.\ the number of clocks, one can show that the number of
\emph{distinct} expressions is only singly exponential.

First we assume that ITA $\A$ has only integral constants, the case of
rational constants is handled at the end of the proof.  It can be
shown that every expression $e_k$ of $F_k$ can be written
\[e_k = \sum_{i_0,\dots,i_p \in \textrm{sub}(k)} \alpha_{k,i_p} \cdot
\alpha_{i_p,i_{p-1}} \cdots \alpha_{i_1,i_0} \cdot x_{i_0}\] with the
convention that $x_0$ is the constant $1$, and where $\textrm{sub}(k)$
is the set of all (ordered) subsequences of $0,\dots,k-1$ and
$\alpha_{j,i}$ is the coefficient of $x_i$ in some update of $x_j$.

For the family $\alpha$ of all integers $\alpha_{j,i}$, assume that
these constants are coded over $b_\alpha$ bits each (including the
sign of the coefficient).  The expression $x_{i_0}$ can also be coded
into an integer of $\log_2(n)$ bits (with a special symbol to indicate
that it is the expression of a clock rather than a constant).  Let $b
= \max(b_\alpha,\log_2(n)+1)$ be the (maximal) number of bits used to
code a coefficient.  Then each term of the sum is a product of at most
$k$ such coefficients, therefore can be coded with $kb$ bits.  Summing
at most $2^k$ such products yields an integer that can be coded over
$kb+k$ bits.  Thus there can be at most $2^{k(b+1)}$ different
expressions in $E_k$.

%

Automaton $\mathcal A'$ is then defined as follows.
\begin{itemize}
\item The set of states is
\[\begin{array}{rcl}
  Q' &=& \{(q^+,e_1,\ldots,e_{i-1})\mid q \in Q, \ 
  \lambda(q)=i \mbox{ and } \forall j,\ e_j \in F_j \} \\
  & & \cup \{(q^-,e_1,\ldots,e_{i}) \mid q \in Q, \  \lambda(q)=i 
\mbox{ and } \forall j,\ e_j \in F_j \},
\end{array}\]
with $pol(q^+,e_1,\ldots,e_{i-1})=pol(q)$ and $pol(q^-,e_1,\ldots,e_{i})=U$.\\
Note that the sequence is empty if $i=1$. Moreover:
\[\lambda(q^+,e_1,\ldots,e_{i-1})=\lambda(q^-,e_1,\ldots,e_{i})=\lambda(q).\]
\item The initial state of $\mathcal A'$ is $(q_0^+,x_1,\ldots, x_{i-1})$
if $\lambda(q_0)=i$.  The final states of $\mathcal A'$ are the
states with first component $q^+$ for $q \in F$. 
\item Let $q \tr{\fee,a,u} q'$ be a transition in $\A$ such that
  $\lambda(q)=i$, $\lambda(q')=i'$ and $u$ is defined by
  $\bigwedge_{j=1}^{i} x_j := C_j$. 
\begin{itemize}
\item[$\bullet$] If $i \leq i'$, then for every
  $(q^+,e_1,\ldots,e_{i-1})$ there is a transition
  $$(q^+,e_1,\ldots,e_{i-1}) \tr{\fee',a,u'}
  (q'^+,e'_1,\ldots,e'_{i'-1})$$ in $\A'$ with $\fee'=\fee(\{x_j
  \leftarrow e_j\}_{j<i})$, update $u'$ is defined by $x_i:=C_i[\{x_j
  \leftarrow e_j\}_{j<i}]$; for all $j<i$, $e'_j=C_j[\{x_k \leftarrow
  e_k\}_{k<i}]$ and for all $j$ such that $i\leq j<i'$, $e'_j=x_j$.
\item[$\bullet$] If $i > i'$ then for every $(q^+,e_1,\ldots,e_{i-1})$
  there is a transition $$(q^+,e_1,\ldots,e_{i-1}) \tr{\fee',a,u'}
  (q'^-,e'_1,\ldots,e'_{i'})$$ in $\A'$ with $\fee'=\fee(\{x_j
  \leftarrow e_j\}_{j<i})$, update $u'$ contains only the trivial
  updates $x_j := x_j$ for all clocks and for all $j\leq i'$,
  $e'_j=C_j[\{x_k \leftarrow e_k\}_{k<i}]$.

\item[$\bullet$] For every $(q^-,e_1,\ldots,e_{i})$ there is in $\A'$
  a transition
\[(q^-,e_1,\ldots,e_{i}) \tr{true,\varepsilon,x_{i}:=e_{i}}
(q^+,e_1,\ldots,e_{i-1}).\]
\end{itemize}
\end{itemize}

In words, given a transition, the guard is modified according to these
expressions. The modification of the update consists only in applying
the update at the current level and taking into account the other
updates in the expressions labeling the destination state.  When the
transition increases the level, the expression associated with a new
``frozen'' clock ($x_j$ for $i\leq j <i'$) is the clock itself. The
urgent states $(q^-,-)$ are introduced for handling the case of a
transition that decreases the level.  In this case, one reaches such a
state that memorizes also the expression of the clock at the current
level.  Note that the memorized expressions can correspond to an
update proceeded at any (higher) level.  From this state a single
transition must be (immediately) taken whose effect is to perform the
update corresponding to the memorized expression.

It is routine to check that the languages of the two automata are
identical.  Each transition in $\A$ is replaced by several transitions
in $\A'$, which number is bounded by the number of expressions that
can be attached to the source of the original transition.  In
addition, transitions decreasing level are further ``split'' through
states $(q^-,-)$.  Thus the number $E'$ of transitions in $\A'$ is
bounded by
\begin{eqnarray*}
E' &\leq& 2\cdot E \cdot |F_n|^n \\
&\leq& 2 \cdot E \cdot \left(2^{n(b(E+1)+1)}\right)^n \\
&\leq& 2 \cdot E \cdot 2^{n^2(b(E+1)+1)} \\
&\leq& 2^{n^2(b(E+1)+1)+1+\log_2(E)} \\
&\leq& 2^{n^2((b+1)(E+1)+1)} \\
E' &\leq& 2^{4b \cdot E \cdot n^2}
\end{eqnarray*}
(provided $E\geq 2$).
This yields the exponential complexity for the number of transitions.
The case of the number of states is similar.

In the case when there are rational constants, assume each constant is
coded with a pair $(r,d)$ of numerator and denominator.  Assume each
$r$ and $d$ can be coded over $b$ bits.  We compute the \emph{lcm}
$\delta$ of all denominators: since there are at most $E$ constants
($E$, the size of $\Delta$ contains the number of guards and updates),
$\delta$ can be coded over $Eb$ bits.  We consider ITA $\A_\delta$
which is $\A$ where all constants are multiplied by $\delta$.  Thus a
constant of $\A_\delta$ is an integer that can be coded over $b' = Eb
+b = b(E+1)$ bits.  The above bound on the number of expressions
applies on $\A_\delta$.  Note that after the construction of
$\A_\delta'$, $\A'$ can be obtained by dividing each constant in
$\A_\delta'$ by $\delta$.  \qed
\end{proof}

\paragraph{Example.}
We illustrate this construction on ITA $\A_2$ of
\figurename~\ref{fig:exitatoitamoins}. The sets of expressions are computed as on \tablename~\ref{tab:itatoitamoins} and the resulting ITA$_-$ $\A'_2$ is depicted on \figurename~\ref{fig:exitamoinsfromita}.

\begin{figure}[h]
\centering
\begin{tikzpicture}[auto]
\node[state,initial] (q0) at (-0.5,-0.25) {$q_0,2$};
\node[state,initial] (q1) at (-0.5,-1.75) {$q_1,2$};
\node[state] (q2) at (2.5,-1) {$q_2,3$};
\node[state] (q3) at (6.5,-1) {$q_3,3$};
\node[state,accepting] (q4) at (9.5,-0.25) {$q_4,3$};
\node[state,accepting] (q5) at (9.5,-1.75) {$q_5,2$};

\path[->] (q0) edge node [pos=0.25] {$x_1:=2$} (q2);
\path[->] (q1) edge (q2);
\path[->] (q2) edge node {\timedtransnoreset{2x_2+x_1 > 3 \wedge x_3 < 2}{x_2:=2x_1+1}} (q3);
\path[->] (q3) edge node [pos=0.75] {\timedtransnoreset{x_2:=x_1+1}{x_3:= 2x_2}} (q4);
\path[->] (q3) edge node [swap,pos=0.75] {$x_1:=1$} (q5);

\end{tikzpicture}
\caption{ITA $\A_2$ containing updates of frozen clocks}
\label{fig:exitatoitamoins}
\end{figure}
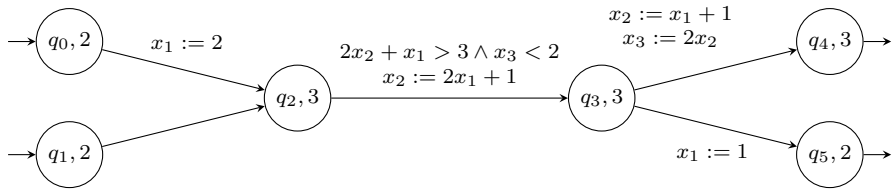

\begin{figure}[h]
\centering
\begin{tikzpicture}[node distance=3.125cm,auto]
\tikzstyle{big state}=[state,shape=rounded rectangle,inner xsep=0pt]

\node[state,initial] (q0) at (0,0) {$q_0^+,2$};
\node[big state,node distance=2cm,right of=q0] (q2cond) {$\begin{array}{c}q_2^+,3\\x_1:=2\end{array}$};
\node[big state,right of=q2cond] (q3cond1) {$\begin{array}{c}q_3^+,3\\x_1:=2\\x_2:=5\end{array}$};
\node[big state,accepting,node distance=4.25cm,right of=q3cond1] (q4cond1) {$\begin{array}{c}q_4^+,3\\x_1:=2\\x_2:=3\end{array}$};

\node[big state,accepting,below of=q4cond1] (q5p) {$\begin{array}{c}q_5^+,2\\x_1:=1\end{array}$};
\node[big state] (q5m1) at ($(q3cond1 |- q5p) + (0,1)$) {$\begin{array}{c}q_5^-,2,U\\x_1:=1\\x_2:=5\end{array}$};
\node[big state] (q5m2) at ($(q3cond1 |- q5p) - (0,1)$) {$\begin{array}{c}q_5^-,2,U\\x_1:=1\\x_2:=2x_1+1\end{array}$};

\node[big state] (q3cond2) at ($(q5m2) + (0,-2)$) {$\begin{array}{c}q_3^+,3\\x_2:=2x_1+1\end{array}$};
\node[state,left of=q3cond2] (q2) {$q_2^+,3$};
\node[state,initial,node distance=2cm,left of=q2] (q1) {$q_1^+,2$};
\node[big state,accepting,node distance=4.25cm,right of=q3cond2] (q4cond2) {$\begin{array}{c}q_4^+,3\\x_2:=x_1+1\end{array}$};

\path[->] (q0) edge (q2cond);
\path[->] (q2cond) edge node {\timedtransnoreset{2x_2+2>3}{\wedge\, x_3<2}} (q3cond1);
\path[->] (q3cond1) edge node {$x_3:=10$} (q4cond1);
\path[->] (q3cond1) edge (q5m1);
\path[->] (q5m1) edge node [pos=0.375] {$\varepsilon,x_2:=5$} (q5p);

\path[->] (q1) edge (q2);
\path[->] (q2) edgenode {\timedtransnoreset{2x_2+x_1>3}{\wedge\, x_3<2}} (q3cond2);
\path[->] (q3cond2) edge node {$x_3:=4x_1+2$} (q4cond2);
\path[->] (q3cond2) edge (q5m2);
\path[->] (q5m2) edge node [swap,pos=0.375] {$\varepsilon,x_2:=x_1+1$} (q5p);

\end{tikzpicture}
\caption{ITA$_-$ $\A'_2$ equivalent to $\A_2$}
\label{fig:exitamoinsfromita}
\end{figure}
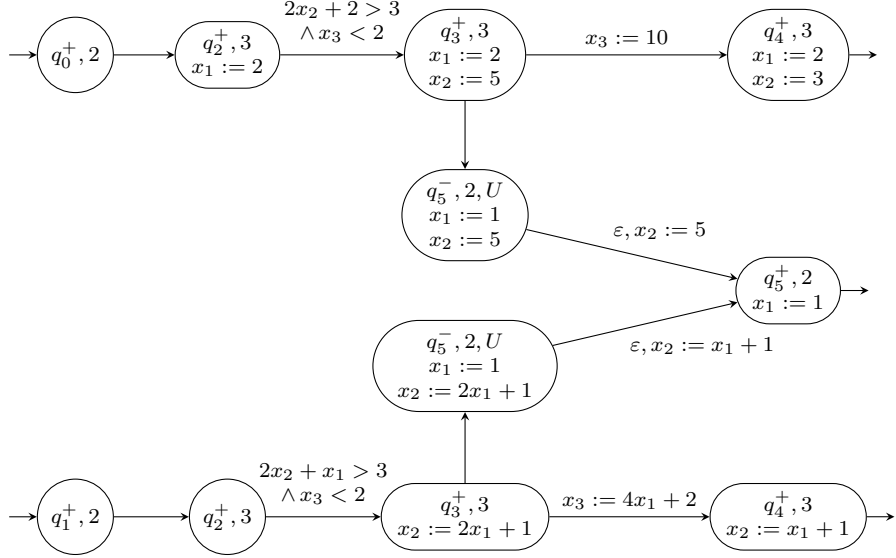

\begin{table}[h]
\centering
\[\begin{array}{|l|c|c|c|}
\hline
i \ \backslash\  j & 1 & 2 & 3 \\
\hline
1 & \{x_1\} & & \\
\hline
2 & \{x_1,2\} & \{x_2\} & \\
\hline
3 & \{x_1,2,1\} & \{x_2, \underbrace{2x_1+1,\  5,\  3}_{q_2 \xrightarrow{x_2:=2x_1+1} q_3}, \underbrace{x_1+1,\  3,\  2}_{q_3 \xrightarrow{x_2:=x_1+1} q_4}\!\!\} & \{x_3\} \\
\hline
\end{array}\]
\caption{Sets of expressions $F_{i,j}$ for $\A_2$.}
\label{tab:itatoitamoins}
\end{table}

The translation above of an ITA into an equivalent ITA$_-$ induces an
exponential blowup.  The proposition below shows that the bound is reached.
\begin{proposition}
  There exist a family $\{\A_n\}_{n\in \N}$ of ITA with two states,
  $n$ clocks and constants coded over $b$ bits, where $b$ is
  polynomial in $n$, such that the equivalent ITA$_-$ built by the
  procedure above has a number of states greater than or equal to
  $2^n$.
\end{proposition}

\begin{proof}
  For $n \in \N$, let $\A_n$ be the ITA with $n$ clocks and two states
  $q_{\textrm{init}}$ (initial) and $q$ (final) both of level $n$ (and
  lazy policy) built as follows.
There is a transition from $q_{\textrm{init}}$ to $q$ with update $\bigwedge_{k=1}^n x_k:=1$ that sets all clocks to $1$.
For $1 \leq k \leq n$ there are two loops on $q$ with updates $x_k := x_{k-1}$ and $x_k := \alpha_{k} x_{k-1}$ respectively, where $\alpha_k$ is the $k$th prime number (and with the convention that $x_0$ is the constant $1$).

When building the sets of expressions, no expressions are added until level $n$, since all updates occur at this level.
At level $k$, $F_{n,k}$ contains (at least) $2^k$ expressions: all possible products of the first $k$ prime numbers, namely
\[F_{n,k} \supset \left\{\prod_{i \in I} \alpha_i \:\middle|\: I \subseteq \{1,\dots,k\}\right\}.\]
Indeed, at level $1$, $F_{n,1} = \{x_1,1,2\}$.
Now assume that $F_{n,k-1}$ contains all products $\prod_{i \in I} \alpha_i$ where $I \subseteq \{1,\dots,k-1\}$.
By update $x_k := x_{k-1}$, $F_{n,k} \supset F_{n,k-1}$.
By update $x_k := \alpha_k x_{k-1}$, $F_{n,k}$ contains all products $\alpha_k \prod_{i \in I} \alpha_i = \prod_{i \in I\uplus\{k\}} \alpha_i$.
Therefore \[F_{n,k} \supset \left\{\prod_{i \in I} \alpha_i \:\middle|\: I \subseteq \{1,\dots,k-1\}\right\} \cup \left\{\prod_{i \in I\uplus\{k\}} \alpha_i \:\middle|\: I \subseteq \{1,\dots,k-1\}\right\} \]
\[F_{n,k} \supset \left\{\prod_{i \in I} \alpha_i \:\middle|\: I \subseteq \{1,\dots,k\}\right\}.\]
The expressions thus built are distinct, since they are products of distinct prime numbers.
Remark that the set of expression for level $k$ is in bijection with a sequence of updates $x_1 := \dots, x_2:= \dots, \dots, x_k:= \dots$, the choice of the update depending on the choice of the set $I$.

Therefore all expressions of $F_{n,n}$ are reached (in association
with state $q$) and the set of states in $\A_n'$ is at least of size
$2^n$.  In addition, it should be noted that the $n$th prime number is
in $O(n\log_2(n))$, therefore can be coded over
$O(\log_2(n)^2)$ bits.  So the size of the constants appearing in the
updates (and the size of the representation of $\A_n$) is polynomial
in $n$ while the representation of $\A'_n$ is exponential in $n$.
\end{proof}
\subsection{Reachability on ITA$_-$}

In this section we use counting arguments to obtain an upper bound for
the reachability problem on ITA$_-$.

The following counting lemma does not depend on the effect of the
updates but only on the timing constraints induced by the policies.
\begin{lemma}[Counting Lemma]
\label{lemma:counting}
Let $\mathcal A$ be an ITA$_-$ with $E$ transitions and $n$ clocks,
then in a sequence $(e_1,\ldots,e_l)$ of transitions of $\mathcal A$
where $l> (E+n)^{3n}$, there exist $i<j$ with $e_i=e_j$ such that the
level of any transition $e_k$ with $i\leq k\leq j$ is greater than or
equal to the level of $e_i$, say $p$, and:
\begin{itemize}
	\item either $e_i$ updates $x_{p}$,
	\item either no $e_k$ with $i\leq k\leq j$ updates $x_{p}$
	and $e_i$ is delayed or lazy.
	\item or no $e_k$ with $i\leq k\leq j$ updates $x_{p}$
	and no time elapses for clock $x_p$ between $e_i$ and $e_j$.
\end{itemize}

\end{lemma}
\begin{proof}
Assume that the conclusions of the lemma are not satisfied,
we claim that $l\leq (E+2n)^{3n}$. 

First we prove that the number of transitions of
level $m$ that occur between two occurrences of transitions of
strictly lower level is less than or equal to $(E+2)^3$. Indeed there
can be no more than $E$ occurrences of transitions that update $x_m$.
Then between two such transitions (or before the first or after the
last) there can be no more than $E$ lazy or delayed transitions of
level $m$ that do not update $x_m$.  Finally between any kind of
previous transitions (or before the first or after the last), there
can be no more than $E$ urgent transitions that do not update $x_m$,
since they prevent time from elapsing at level $m$.

Summing up, there can be no more than $E+E(E+1)+E(E(E+1)+1)\leq
(E+1)^3$ transitions of level $m$ that occur between two occurrence of
transitions of strictly lower level.

Now we prove by induction that the number of transitions at level less
than or equal to $m$ is at most $(E+m)^{3m}$. This is true for $m=1$ by
the previous proof. Assume the formula valid for $m$, then grouping
the transitions of level $m+1$ between the occurrences of transition
of lower level (or before the first or after the last), we obtain that
the number of transitions at levels less than or equal to $m+1$
is at most:
\begin{eqnarray*}(E+m)^{3m} + ((E+m)^{3m}+1)(E+1)^3 &\leq& (E+m)^{3m+3}+2(E+m)^{3m} \\&\leq& (E+m+1)^{3(m+1)}\qquad\qed\end{eqnarray*}
\end{proof}

\begin{proposition}
\label{proposition:reachitamoins}
The reachability problem for ITA$_-$ belongs to \emph{NEXPTIME}. More precisely,
reachability can be checked over paths with length less than or equal
to $(E+n)^{3n}$, where $E$ is the number of transitions and $n$
is the number of clocks.
\end{proposition}
\begin{proof}
  Let $\A=(\Sigma, Q, q_0,F,pol, X, \lambda, \Delta)$ be an ITA$_-$
  with $n$ clocks.  Let $E=|\Delta|$ be the number of transitions of
  $\A$.  Assume that there is a run of minimal length $\rho$ from
  $(q_0,v_0)$ to some configuration $(q_f,v_f)$.  Suppose now that
  $|\rho| > B = (E+n)^{3n}$.  We will build a run $\rho'$ from
  $(q_0,v_0)$ to $(q_f,v_f)$ that is strictly smaller, hence
  contradicting the minimality hypothesis.

  Since $|\rho|>B$, then one of the three cases of
  Lemma~\ref{lemma:counting} applies.  Therefore there is a transition
  $e$ at level $k$ repeated twice, from positions
  $\pi$ and $\pi'$ and separated by a subrun $\sigma$ containing only
  transitions of level higher than or equal to $k$.  Moreover:
\begin{itemize}
\item Either $e$ updates $x_k$.  In this case, all clocks have the
  same value after the first and the second occurrence of $e$.  Hence
  removing $e \sigma = \rho_{[\pi,\pi'[}$ from $\rho$ yields a valid
  run $\rho'$ of $\A$ reaching $(q_f,v_f)$.  Run $\rho'$ is strictly
  smaller than $\rho$, since $e\sigma$ which is of length at least $1$
  was removed.
\item Either no update occurred for $x_k$ and $e$ is delayed or lazy.
  In this case, upon reaching $\pi'$, the clocks of level $i<k$ have
  retained the same value, while $x_k$ has increased by
  $Dur\left(\rho_{[\pi,\pi']}\right)$.  Hence when replacing $e \sigma
  = \rho_{[\pi,\pi'[}$ by a time step of duration
  $Dur\left(\rho_{[\pi,\pi']}\right)$, the configuration in $\pi'$ is
  unchanged.  In addition, since $e$ was delayed or lazy, this time
  step is allowed in $\A$, and this yields a shorter run of $\A$.
\item Or no update occurred and $\pi$ and $\pi'$ are at the same
  instant (separated by instantaneous actions).  In this case, all
  clocks of level smaller than or equal to $k$ again have the same value after the first and the second
  occurrence of $e$.  Again removing $\rho_{[\pi,\pi'[}$ yields a
  smaller run.
\end{itemize}

The decision procedure is as follows. It non deterministically guesses
a path in the ITA$_-$ whose length is less than or equal to the
bound. In order to check that this path yields a run, it builds a
linear program whose variables are $\left\{x_i^j\right\}$, where $x_i^j$ is the
value of clock $x_i$ after the $j$th step, and $\{d_j\}$ where $d_j$
is the amount of time elapsed during the $j$th step, when $j$
corresponds to a time step.  The equations and inequations are deduced
from the guards and updates of discrete transitions in the path and
the delay of the time steps. The size of this linear program is
exponential w.r.t. the size of the ITA$_-$. As a linear program can
be solved in polynomial time~\cite{RoTeVi97}, we obtain a procedure in
NEXPTIME.
\qed
\end{proof}

One could wonder whether the class graph construction would lead to a better
complexity when applied on ITA$_-$. Unfortunately, the number
of expressions occurring in the class graph 
while being smaller than for ITA is still doubly exponential
w.r.t. the size of the model.

\section{Timed model-checking}\label{sec:modelchecking}

First observe that model-checking \ctlstar formulas on ITA can be done
with classical procedures on the class graph previously built. We now
consider verification of real time formulas.

\bigskip In the case of linear time, the logic \ltl has been extended
into the Metric Temporal Logic (\mtl)~\cite{koymans90}, by adding time
intervals as constraints to the $\until$ modality. However, \mtl
suffers from undecidability of the model-checking problem on TA.
Hence decidable fragments have been proposed, such as Metric Interval
Temporal Logic (\mitl)~\cite{alur96}, which prohibits the use of point
intervals (of the form $[a,a]$). Later, \mitl was restricted into
State Clock Logic (\scl)~\cite{raskin97}, in order to obtain more
efficient verification procedures.  Model-checking \mitl (thus \scl)
on TA is decidable.  Unfortunately, we show here that model-checking
\scl (thus \mitl) on ITA is undecidable. For this, we reduce the
halting problem on a two counter machine into model-checking an \scl
formula on an ITA.

\bigskip Concerning branching time logics, at least two different
timed extensions of \ctl have been proposed. The first
one~\cite{alur93} also adds time intervals to the $\until$ modality
while the (more expressive) second one considers formula
clocks~\cite{HNSY94}.  Model-checking timed automata was proved
decidable in both cases and compared expressiveness was revisited
later on~\cite{bouyer05}.

We conjecture that model-checking of \tctl is undecidable when using
two (or more) formula clocks.  Indeed, as shown in
Section~\ref{subsec:mctctlundec}, the reachability problem in a
product of an ITA and a TA with two clocks is undecidable, thus
prohibiting model-checking techniques through automaton product and
reachability testing as in~\cite{aceto98}. However, contrary to what
is claimed in~\cite{berard10}, this is not enough to yield an
undecidability proof.

Two fragments for which model-checking is decidable on ITA have
nonetheless been identified.  The first one, \tctlcint, accepts only
internal clocks (from the automaton on which the formulas will be
evaluated) as formula clocks.  The second one, \tctlp, restricts the
nesting of $\until$ modalities. We provide verification procedures in
both cases.

\subsection{Undecidability of State Clock Logic}

We first consider the timed extension of linear temporal logic,
and more particularly the \scl fragment~\cite{raskin97}.

\begin{definition}
Formulas of the timed logic \scl are defined by the following grammar:
\[\psi = p \mid \psi \wedge \psi \mid \neg \psi \mid \psi \until \psi \mid \psi \since \psi \mid \nextoc{\bowtie a} \psi \mid \lastoc{\bowtie a} \psi\]
where $p \in AP$ is an atomic proposition, $\rel \,\in \{>,\geq,=,\leq,<\}$, and $a$ is a rational number.
\end{definition}
We use the usual shorthands $\mathbf{t}$ for $\neg(p \wedge \neg p)$,
$\eventually \psi$ for $\mathbf{t} \until \psi$, $\globally \psi$ for
$\neg (\eventually \neg\psi)$ and $\varphi \Rightarrow \psi$ for
$\neg(\varphi \wedge \neg \psi)$.

The semantics are defined in the usual manner for boolean operators
and $\until$.  The $\since$ modality is the past version of $\until$.
Modality $\nextoc{\bowtie a} \psi$ is true if the \emph{next} time
$\psi$ is true will occur in a delay that respects the condition
$\bowtie a$.  Similarly, $\lastoc{\bowtie a} \psi$ is true if the
\emph{last} time $\psi$ was true occurred in a (past) delay that
respects the condition $\bowtie a$.
More formally, for an execution $\rho$, we inductively define $(\rho,\pi) \models \varphi$ by:
\[\begin{array}{lcl}
(\rho,\pi) \models p &\quad\textrm{iff}\ & p \in lab(s_\pi) \\
(\rho,\pi) \models \varphi \wedge \psi &\quad\textrm{iff}\ & (\rho,\pi) \models \varphi \ \textrm{and}\  (\rho,\pi) \models \psi\\
(\rho,\pi) \models \neg \varphi &\quad\textrm{iff}\ & (\rho,\pi) \not\models \varphi\\

(\rho,\pi) \models \varphi \until \psi  &\quad\textrm{iff}\ &\textrm{there is a position } \pi' \geq_\rho \pi \textrm{ such that } (\rho,\pi') \models \psi\\
& &\textrm{and forall } \pi'' \textrm{ s.t. } \pi \leq_\rho \pi'' <_\rho \pi', \ (\rho,\pi'') \models \varphi \vee \psi\\

(\rho,\pi) \models \varphi \since \psi  &\quad\textrm{iff}\ &\textrm{there is a position } \pi' \leq_\rho \pi \textrm{ such that } (\rho,\pi') \models \psi \\
& &\textrm{and forall } \pi'' \textrm{ s.t. } \pi \geq_\rho \pi'' >_\rho \pi', \ (\rho,\pi'') \models \varphi \vee \psi\\

(\rho,\pi) \models \nextoc{\bowtie a} \varphi &\quad\textrm{iff}\ & \textrm{either } (\rho,\pi) \models \varphi \textrm{ and } 0 \bowtie a\\
& &\textrm{or, there is a position } \pi' >_\rho \pi \textrm{ such that } (\rho,\pi') \models \varphi,\\
& & \ Dur\!\left(\rho_{[\pi,\pi']}\right) \bowtie a \textrm{ and forall } \pi''\textrm{ s.t. } \pi \leq_\rho \pi'' <_\rho \pi', (\rho,\pi'') \not\models \varphi\\

(\rho,\pi) \models \lastoc{\bowtie a} \varphi &\quad\textrm{iff}\ & \textrm{either } (\rho,\pi) \models \varphi \textrm{ and } 0 \bowtie a\\
& &\textrm{or, there is a position } \pi' <_\rho \pi \textrm{ such that } (\rho,\pi') \models \varphi,\\
& & \ Dur\!\left(\rho_{[\pi',\pi]}\right) \bowtie a \textrm{ and forall } \pi'' \textrm{ s.t. } \pi \geq_\rho \pi'' >_\rho \pi', (\rho,\pi'') \not\models \varphi
\end{array}\]
Given an ITA $\A$ and an \scl formula $\varphi$, $\A \models \varphi$ if 
for all executions $\rho$ of $\A$, $(\rho,\pi_0) \models \varphi$, where 
$\pi_0=0$ is the initial position of $\rho$.

\begin{theorem}\label{thm:mcsclundec}
  Model checking \scl over ITA is undecidable.  Specifically, there
  exists a fixed formula using only modalities $\until$ and $\lastoc{=
    a}$ such that checking its truth over ITA with $3$ levels is
  undecidable.
\end{theorem}

\begin{proof}
  We build an ITA and an \scl formula that together simulate a
  deterministic two counter machine.  More specifically, we define a
  formula $\varphi_{2cm}$ such that given a two counter machine $\M$,
  we can build an ITA $\A_\M$ with three clocks such that $\A_\M
  \models \varphi_{2cm}$ if and only if $\M$ does not halt.
  
  Recall that such a machine $\cal M$ consists of a finite sequence of
  labeled instructions, which handle two counters $c$ and $d$, and
  ends at a special instruction with label \emph{Halt}.  The other
  instructions have one of the two forms below, where \(e \in
  \{c,d\}\) represents one of the two counters:
\begin{itemize}
  \item \(e := e + 1\); goto $\ell'$
  \item if \(e > 0\) then ($e := e -1$; goto $\ell'$) else goto $\ell''$
\end{itemize}
Without loss of generality, we may assume that the counters have
initial value zero.  The behavior of the machine is described by a
(possibly infinite) sequence of configurations:
\(\langle\ell_0,0,0\rangle \langle\ell_1,n_1,p_1\rangle\dots
\langle\ell_i,n_i,p_i\rangle\dots\), where $n_i$ and $p_i$ are the
respective counter values and $\ell_i$ is the label, after the
$i^{th}$ instruction.  The problem of termination for such a machine
(``is the \emph{Halt} label reached?'') is known to be
undecidable~\cite{minsky67}.

The idea of the encoding is that, provided the execution satisfies the
formula, clocks of level $1$ and $2$ keep the values of $c$ and $d$
indifferently, by $x_i = \frac1{2^n}$ if $n$ is the value of a counter
$e$.  Level $3$ will be used as the working level.  Transmitting the
value of clocks to lower levels, prohibited in the ITA model, will be
enforced by \scl formulas.  In the sequel, we will define:
\begin{itemize}
\item a module $\A_{\leftrightarrow}$ and a formula
  $\varphi_{\leftrightarrow}$ such that the values contained in clocks
  $x_1$ and $x_2$ at the beginning of an execution $\rho$ are swapped
  if and only if $(\rho,0) \models \varphi_{\leftrightarrow}$,
\item a module $\A_{+}$ and a formula $\varphi_{+}$ such that if the
  value of $x_2$ is $\frac1{2^n}$ at the beginning of an execution
  $\rho$, then $x_2$ has value $\frac1{2^{n+1}}$ if and only if
  $(\rho,0) \models \varphi_{+}$,
\item a module $\A_{-}$ and a formula $\varphi_{-}$ such that if the
  value of $x_2$ is $\frac1{2^n}$ with $n>0$ at the beginning of an
  execution $\rho$, then $x_2$ has value $\frac1{2^{n-1}}$ if and only
  if $(\rho,0) \models \varphi_{-}$.
\end{itemize}
Joining these modules according to $\M$ yields an ITA.  Combining the
formulas (independently of $\M$), we obtain an \scl formula that is
satisfied if some execution, while complying to the formulas of the
modules, reaches the final state. Both constructions are explained in
details after the definitions below.

\bigskip Let us define formulas $Span_1 = \propq' \Rightarrow
\lastoc{=1} \propq$ and $Span_2 = \propp' \Rightarrow \lastoc{=2}
\propp$ where $\propp$, $\propp'$, $\propq$, $\propq'$ are
propositional variables.  Let $x_1^0$ and $x_2^0$ denote the
respective values of $x_1$ and $x_2$ upon entering a given module.
\begin{description}[font=\bfseries]
\item[Swapping module.] The module $\A_{\leftrightarrow}$ that swaps
  the values of $x_1$ and $x_2$ is depicted in
  \figurename~\ref{fig:modswap}.  Note that this module does not
  actually swap the values of $x_1$ and $x_2$ for every execution.
  However, by imposing that state $q_{end}$ is reached exactly $2$
  time units after $q_{0}$ (or $q_0'$) was left, and that $q_4$
  (resp. $q_4'$) is reached exactly $1$ t.u. after $q_1$
  (resp. $q_1'$) was left, the values of $x_1$ and $x_2$ will be
  swapped.  This requirement can be expressed in \scl by
  $\varphi_{\leftrightarrow} = \globally\, \left(Span_1 \wedge
    Span_2\right)$.  Let $w_i$ be the time elapsed in state $q_i$, for
  an execution $\rho$ of $\A_{\leftrightarrow}$ that satisfies
  $\varphi_{\leftrightarrow}$.  Note that $q_{start}$ and $q_{end}^{\neq}$
  are all urgent, hence no time can elapse in these states.  We shall
  therefore consider only what happens in the swapping submodules.  We
   detail only the case when $x_2>x_1$, the case when $x_2<x_1$ is
  analogous.  The ITA constraints provide:
\[\begin{array}{rclcl}
w_0 &=& 0 &\qquad& (q_0 \textrm{ is urgent})
\\
w_1 &=& x_2^0 - x_1^0 && (\textrm{update } x_3:=x_1 \textrm{ and guard } x_3=x_2)
\\
w_2 &=& 1 - x_2^0 && (\textrm{guard } x_3=1)
\\
w_4 &=& 0 && (q_4 \textrm{ is urgent})
\end{array}\]
\MS{Reviewer 1, Rq 17. Done.}
The time spent between the last instant $\propq$ was satisfied (upon leaving $q_1$) and the only instant when $\propq'$ is true (upon entering $q_4$) is exactly the time spent in states $q_2$ and $q_3$.
Similarly, the time between the last instant $\propp$ was satisfied (leaving $q_0$) and the instant $\propp'$ is true (when reaching $q_{end}^{\neq}$) is the total amount of time spent in $q_1$, $q_2$, $q_3$, $q_4$, and $q_5$.
Hence, if $\varphi_{\leftrightarrow}$ is satisfied then:
\[\begin{array}{rclcl}
w_2 + w_3 &=& 1 &\qquad& \left(\propq' \Rightarrow \lastoc{=1} \propq\right)
\\
w_1 + w_2 + w_3 + w_4 + w_5 &=& 2 && \left(\propp' \Rightarrow \lastoc{=2} \propp\right)
\end{array}\]
Hence $w_3 = x_2^0$ and $w_5=1-w_1=x_1^0-\left(x_2^0-1\right)$.
Since upon entering $q_3$, clock $x_1$ has value $0$, when leaving, $x_1$ has value $x_2^0$.
Similarly, when entering $q_5$, $x_2$ has value $x_1 - 1 = x_2^0 -1$, therefore $x_2$ has value $x_1^0$ when reaching $q_{end}^{\neq}$.
Note that this module swaps $x_1$ and $x_2$ regardless of their coding a counter value.
\smallskip
\item[Incrementation module.] The same idea applies for the
  incrementation module $\A_{+}$ of \figurename~\ref{fig:modincrem}.
  We force the time spent in total in $r_1$ and $r_2$ is one,
  expressed in \scl by $\varphi_{+}=\globally\, Span_1$.  The guards
  and updates in $\A_{+}$ ensure that, with the same notation as
  above, time spent in $r_1$ will be $1-\frac12 x_2^0$.  Hence, when
  reaching $r_3$, clock $x_2$ will have value $\frac12 x_2^0$.
  Therefore, if $x_2^0 = \frac1{2^n}$, coding a counter of value $n$,
  at the end of $\A_{+}$, $x_2$ has value $\frac1{2^{n+1}}$, thus
  coding a value $n+1$ for the same counter.
\item[Decrementation module.] Decrementation, for which the
  corresponding module is depicted on \figurename~\ref{fig:moddecrem},
  is handled in a similar manner (with
  $\varphi_{-}=\varphi_{+}=\globally\, Span_1$).  The only
  difference is that $x_2$ has to be compared to $1$ in order to test
  if the value of the counter encoded by $x_2$ is $0$.
\end{description}
Since the constraints in $Span_1$ (and $Span_2$) are equalities, they can be satisfied only if $\propq'$ (and $\propp'$) are true at a single point in time.

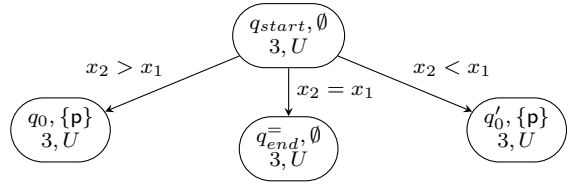
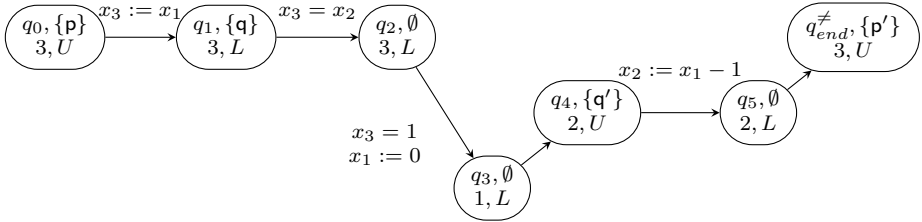
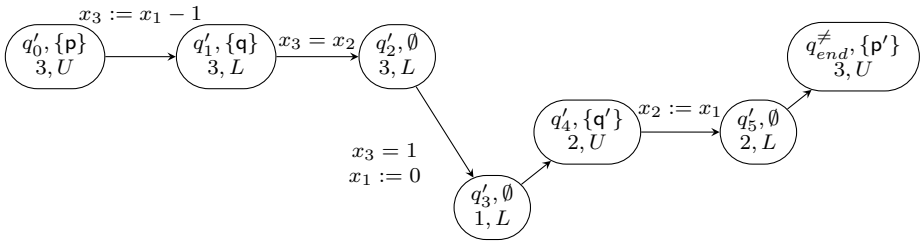
\begin{figure}
\centering
\subfigure[Choice submodule.]{\label{fig:modswapchoice}
\begin{tikzpicture}[node distance=2.375cm,initial text=,auto,scale=0.5]
\tikzstyle{every state}=[draw=black,fill=white,inner xsep=-4pt,font=\small,shape=rounded rectangle]
\node[state] (q0) at (0,0) {\etat{q_{start},\emptyset}{3,U}};
\node[state] (qgeq) at (-6,-2.5) {\etat{q_0,\{\propp\}}{3,U}};
\node[state] (qleq) at (6,-2.5) {\etat{q_0',\{\propp\}}{3,U}};
\node[state] (qeq) at (0,-3) {\etat{q_{end}^=,\emptyset}{3,U}};

\path[->] (q0) edge node[swap] {$x_2 > x_1$} (qgeq);
\path[->] (q0) edge node {$x_2 < x_1$} (qleq);
\path[->] (q0) edge node {$x_2 = x_1$} (qeq);
\end{tikzpicture}
}

\subfigure[Swapping submodule ($x_2 > x_1$).]{\label{fig:modswapgeq}
\begin{tikzpicture}[node distance=2.25cm,initial text=,auto,scale=0.5]
\tikzstyle{every state}=[draw=black,fill=white,inner xsep=-4pt,font=\small,shape=rounded rectangle]

\node (lev3) at (0,0) {};
\node (lev2) at (0,-2) {};
\node (lev1) at (0,-4) {};

\node[state] (q0) at (lev3) {\etat{q_0,\{\propp\}}{3,U}};
\node[state] (q1) [right of=q0] {\etat{q_1,\{\propq\}}{3,L}};
\node[state] (q2) [right of=q1] {\etat{q_2,\emptyset}{3,L}};
\node[state] (q3) at ($(lev1 -| q2) +(2.5,0)$) {\etat{q_3,\emptyset}{1,L}};
\node[state] (q4) at ($(lev2 -| q3) +(2.5,0)$) {\etat{q_4,\{\propq'\}}{2,U}};
\node[state] (q5) [right of=q4] {\etat{q_5,\emptyset}{2,L}};
\node[state] (q6) at ($(lev3 -| q5) +(2.5,0)$) {\etat{q_{end}^{\neq},\{\propp'\}}{3,U}};

\path[->] (q0) edge node [outer sep=4pt] {$x_3:=x_1$} (q1);
\path[->] (q1) edge node [outer sep=4pt] {$x_3=x_2$} (q2);
\path[->] (q2) edge node [swap] {\timedtransnoreset{x_3=1}{x_1:=0}} (q3);
\path[->] (q3) edge (q4);
\path[->] (q4) edge node [outer sep=10pt] {$x_2:=x_1-1$} (q5);
\path[->] (q5) edge (q6);
\end{tikzpicture}
}

\subfigure[Swapping submodule ($x_2 < x_1$).]{\label{fig:modswapleq}
\begin{tikzpicture}[node distance=2.25cm,initial text=,auto,scale=0.5]
\tikzstyle{every state}=[draw=black,fill=white,inner xsep=-4pt,font=\small,shape=rounded rectangle]

\node (lev3) at (0,0) {};
\node (lev2) at (0,-2) {};
\node (lev1) at (0,-4) {};

\node[state] (q0) at (lev3) {\etat{q_0',\{\propp\}}{3,U}};
\node[state] (q1) [right of=q0] {\etat{q_1',\{\propq\}}{3,L}};
\node[state] (q2) [right of=q1] {\etat{q_2',\emptyset}{3,L}};
\node[state] (q3) at ($(lev1 -| q2) +(2.5,0)$) {\etat{q_3',\emptyset}{1,L}};
\node[state] (q4) at ($(lev2 -| q3) +(2.5,0)$) {\etat{q_4',\{\propq'\}}{2,U}};
\node[state] (q5) [right of=q4] {\etat{q_5',\emptyset}{2,L}};
\node[state] (q6) at ($(lev3 -| q5) +(2.5,0)$) {\etat{q_{end}^{\neq},\{\propp'\}}{3,U}};

\path[->] (q0) edge node [outer sep=10pt] {$x_3:=x_1-1$} (q1);
\path[->] (q1) edge node {$x_3=x_2$} (q2);
\path[->] (q2) edge node [swap] {\timedtransnoreset{x_3=1}{x_1:=0}} (q3);
\path[->] (q3) edge (q4);
\path[->] (q4) edge node [outer sep=3pt] {$x_2:=x_1$} (q5);
\path[->] (q5) edge (q6);
\end{tikzpicture}
}
\caption[Swapping module $\A_{\leftrightarrow}$.]{Swapping module $\A_{\leftrightarrow}$. Submodules are connected through identical states ($q_0$, $q_0'$, $q_{end}^{\neq}$).}
\label{fig:modswap}
\end{figure}

\begin{figure}
\centering
\begin{tikzpicture}[node distance=5cm,initial text=,auto,scale=0.5]
\tikzstyle{every state}=[draw=black,fill=white,inner xsep=-4pt,font=\small,shape=rounded rectangle]

\node (lev3) at (0,0) {};
\node (lev2) at (0,-2) {};

\node[state] (q0) at (lev3) {\etat{r_0,\{\propq\}}{3,U}};
\node[state] (q1) [right of=q0] {\etat{r_1,\emptyset}{3,L}};
\node[state] (q2) at ($(lev2 -| q1) +(5,0)$) {\etat{r_2,\emptyset}{2,L}};
\node[state] (q3) at ($(lev3 -| q2) +(5,0)$) {\etat{r_3,\{\propq'\}}{3,U}};

\path[->] (q0) edge node {$x_3:=\frac12 x_2$} (q1);
\path[->] (q1) edge node[swap] {\timedtransnoreset{x_3=1}{x_2:=0}} (q2);
\path[->] (q2) edge (q3);
\end{tikzpicture}
\caption{Incrementation module.}
\label{fig:modincrem}
\end{figure}
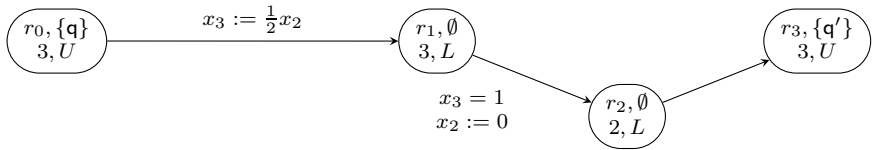

\begin{figure}
\centering
\begin{tikzpicture}[node distance=3.5cm,initial text=,auto,scale=0.5]
\tikzstyle{every state}=[draw=black,fill=white,inner xsep=-4pt,font=\small,shape=rounded rectangle]

\node (lev3) at (0,0) {};
\node (lev2) at (0,-2) {};

\node[state] (q0) at (lev3) {\etat{s_0,\{\propq\}}{3,U}};
\node[state] (q1) [right of=q0] {\etat{s_1,\emptyset}{3,L}};
\node[state] (q2) at ($(lev2 -| q1) +(4,0)$) {\etat{s_2,\emptyset}{2,L}};
\node[state] (q3) at ($(lev3 -| q2) +(4,0)$) {\etat{s_3,\{\propq'\}}{3,U}};
\node[state] (q4) [left of=q0] {\etat{s_4,\emptyset}{3,U}};

\path[->] (q0) edge node {\timedtransnoreset{x_2<1}{x_3:=2\, x_2}} (q1);
\path[->] (q1) edge node[swap] {\timedtransnoreset{x_3=1}{x_2:=0}} (q2);
\path[->] (q2) edge (q3);
\path[->] (q0) edge node[swap] {$x_2=1$} (q4);
\end{tikzpicture}
\caption{Decrementation module.}
\label{fig:moddecrem}
\end{figure}
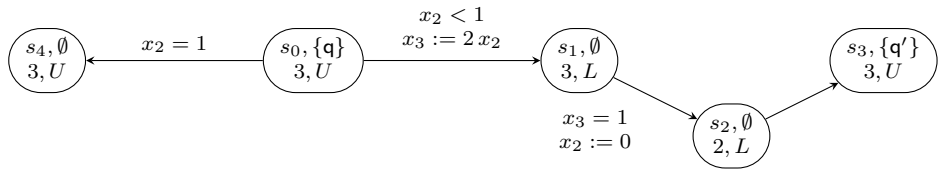

Automaton $\A_\M$ is then defined as the concatenation of modules according to $\M$.
For clarity, a state $(q,\ell)$ denotes state $q$ in a module corresponding to instruction $\ell$.

Namely, an instruction $\ell$ incrementing $c$ and going to $\ell'$ is an incrementation module with a transition from $(r_3,\ell)$ to the first state of the module corresponding to $\ell'$ (either $(q_{start},\ell')$, $(r_0,\ell')$ or $(s_0,\ell')$).
In the case of an incrementation of $d$, the corresponding module will be the concatenation of $\A_{\leftrightarrow}^{in}$, $\A_{+}$, and $\A_{\leftrightarrow}^{out}$.
Modules $\A_{\leftrightarrow}^{in}$ and $\A_{\leftrightarrow}^{out}$ are two copies of a swapping module $\A_{\leftrightarrow}$.
The states of $\A_{\leftrightarrow}^{in}$ and  $\A_{\leftrightarrow}^{out}$ will be respectively denoted $(q,\ell,in)$ and $(q,\ell,out)$) to avoid confusion.
The last swap is performed in order to restore that $x_2$ contains the value of $c$ and $x_1$ the value of $d$.
The concatenation is done by transitions from $(q_{end}^{\neq},\ell,in)$ and $(q_{end}^=,\ell,in)$ to $(r_0,\ell)$, from $(r_3,\ell)$ to $(q_{start},\ell,out)$.
States $(q_{end}^{\neq},\ell,out)$ and $(q_{end}^=,\ell,out)$ are then linked to the first state of the module for $\ell'$.

Decrementation is handled in a similar way.
The main difference resides in the fact that $(s_4,\ell)$ is linked to the first state of $\ell''$.
In the decrementation of $d$, $(s_4,\ell)$ is linked to a swapping module $\A_{\leftrightarrow}^{out'}$ (disjoint from $\A_{\leftrightarrow}^{in}$ and $\A_{\leftrightarrow}^{out}$), in turn linked to the first state of $\ell''$.

The \emph{Halt} instruction is encoded in a single state $h$ labeled with $\{\proph\}$.
The initial state of the automaton is a new state \textit{Init} of level $3$.
It has \emph{urgent} policy and satisfies no atomic proposition.
State \textit{Init} is linked to the first state of the module corresponding to $\ell_0$, the initial instruction of $\M$, by a transition that updates both $x_1$ and $x_2$ to $1$, simulating the initialization of both counters to $0$.
\bigskip

Let us define formula $\varphi_{2cm} = \eventually (\neg Span_1 \vee
\neg Span_2) \vee \globally \neg \proph$.  An execution $\rho$ of
$\A_\M$ satisfies $\varphi_{2cm}$ if either it violates at some point
a constraint $Span_i$, which means $\rho$ does not correspond to an
execution of $\M$, or $\rho$ never reaches state $h$, which means the
execution of $\M$ is not halting.

If $\M$ has a halting execution, then it can be converted into an
execution $\rho$ that complies to the $Span_i$ constraints and reaches
the final state $h$.  Hence $\rho \not\models \varphi_{2cm}$ and
$\A_\M \not\models \varphi_{2cm}$.

Conversely, if $\A_\M \not\models \varphi_{2cm}$, then consider an
execution $\rho$ that does not verify $\varphi_{2cm}$.  Execution
$\rho$ both reaches $h$ and complies to the $Span_i$ constraints,
hence encodes a halting execution of $\M$.

As a result, $\M$ has no halting execution if and only if \[\A_\M
\models \eventually \left(\left(\neg \propq' \wedge \neg
    \lastoc{=1}\propq\right) \vee \left(\neg \propp' \wedge \neg
    \lastoc{=2}\propp\right)\right) \vee \globally \neg\proph.\]
Remark that this formula does not have nested history or prediction
modalities ($\lastoc{\bowtie a}$ and $\nextoc{\bowtie a}$).  Hence
$\scl$ with a discrete semantics (evaluating the subformulas only upon
entering a state) would also be undecidable. \qed
\end{proof}

\subsection{Model-checking branching time properties with internal clocks}

In this section we consider the extension of \ctl with model clocks,
the corresponding fragment being denoted by \tctlcint.  Such a logic
allows to reason about the sojourn times in different levels which is
quite useful when designing real-time operating systems.
For example, formula $\always (x_2 \leq 3) \until \textit{safe}$ expresses that all executions reach a safe state while spending less than 3 time units in level 2 (assuming $x_2$ is not updated during the execution).
\MS{Reviewer 2, Rq 5. Done.}
\MS{Reviewer 1, Rq 19. Correction comment\'ee car je pense qu'elle fait doublon avec le texte pr\'eexistant.}
Model-checking is achieved
by adapting a class graph construction for untiming ITA
(Section~\ref{sec:regular}) and adding information relevant to the
formula. The problem is thus reduced to a \ctl model checking problem
on this graph.

\begin{definition} Formulas of the timed logic \tctlcint are defined by the
  following grammar:
\[
\psi ::= p \mid \psi \wedge \psi \mid \neg \psi \mid \sum_{i \geq 1} a_i \cdot x_i + b
\rel 0 \mid \always \psi \until \psi \mid \expath \psi \until \psi
\]
where $p \in AP$ is an atomic proposition, $x_i$ are model clocks, $a_i$ and $b$ are rational numbers such
that $(a_i)_{i \geq 1}$ has finite domain, and $\rel \,\in \{>,\geq,=,\leq,<\}$.
\end{definition}
As before we use the classical shorthands $\eventually$, $\globally$,
and boolean operators.

Let $\A=\langle\Sigma, AP, Q, q_0, F,pol, X, \lambda, lab,
\Delta\rangle$ be an interrupt timed automaton and $S= \{(q,v,\beta)
\mid q \in Q, \ v \in \R^X, \ \beta \in \{\top,\bot\} \}$, the set of
configurations. The formulas of \tctlcint are interpreted over
configurations\footnote{The boolean value in the configuration is not
  actually used. The logic could be enriched to take advantage of this
  boolean, to express for example that a run lets some time elapse in
  a given state.} $s=(q,v,\beta)$.

The semantics of \tctlcint is defined as follows on the transition
system $\T_{\A}$ associated with $\A$. For atomic propositions and a
configuration $s=(q,v,\beta)$, with $lab(s)=lab(q)$:
\[
\begin{array}{lcl}
  s \models p &\quad\textrm{iff}\ & p \in lab(s) \\
  s \models \sum_{i\geq 1} a_i \cdot x_i + b \rel 0 &\quad\textrm{iff}\ &
  v \models \sum_{i\geq 1} a_i \cdot x_i + b \rel 0
\end{array}
\]
and inductively:
\[
\begin{array}{lcl}
  s \models \varphi \wedge \psi &\quad\textrm{iff}\ & s \models \varphi\ \textrm{and}
  \ s \models \psi \\
  s \models \neg\varphi &\quad\textrm{iff}\ & s \not\models \varphi \\
  s \models \always \fee \until \psi &
\quad\textrm{iff}\ & \textrm{for all } \rho \in Exec(s), \ \rho \models \fee \until \psi\\
  s \models \expath \fee \until \psi &
\quad\textrm{iff}\ & \textrm{there exists } \rho \in Exec(s) 
\textrm{ s. t. } \rho \models \fee \until \psi\\
  \textrm{with } 
  \rho \models \fee \until \psi  &
\quad\textrm{iff}\ &\textrm{there is a position } \pi \in \rho 
\textrm{ s. t. } s_\pi \models \psi \\
& &\textrm{and } \forall \pi' <_{\rho} \pi, \ s_{\pi'} \models \fee \vee \psi.
\end{array}
\]
\MS{Reviewer 2, Rq 6. Done.}

The automaton $\A$ satisfies $\psi$ if the initial configuration $s_0$
of $\T_{\A}$ satisfies $\psi$.
\begin{theorem}\label{thm:mctctlintdec}
  Model checking \tctlcint on interrupt timed
  automata can be done in $2$-EXPTIME, and in PTIME when the number
  of clocks is fixed.
\end{theorem}

The proof relies on a refinement of the class graph according to the
\MS{Reviewer 2, Rq 7. Done.}
comparisons in the formula to model-check.  It is detailed in
Appendix~\ref{app:proofmctctlindec} and we show the resulting graph on
an example below.

\paragraph{Example.}
Consider the ITA $\A_1$ (\figurename~\ref{fig:exita1}) and the formula
$\varphi_1= \expath \eventually (q_1 \wedge (x_2 > x_1)$.  We assume
that $q_1$ is a propositional property true only in state $q_1$.
Initially, the set of expressions are $E_1 = \{x_1,0\}$ and $E_2 =
\{x_2,0\}$.  First the expression $-\frac12 x_1 + 1$ is added into
$E_2$ since $x_1 + 2x_2 =2$ appears on the guard in the transition
from $q_1$ to $q_2$.  Then expression $1$ is added to $E_1$ because
$x_1 - 1 < 0$ appears on the guard in the transition from $q_0$ to
$q_1$.  Finally expression $x_1$ is added to $E_2$ since $x_2 - x_1 >
0$ appears in $\varphi_1$.  The iterative part of the procedure goes
as follows.  Since there is a transition from $q_0$ of level $1$ to
state $q_1$ of level $2$, we compute all differences between
expressions of $E_2$, then normalize them:
\begin{itemize}[label=\textbullet] 
  \item $x_1 - 0$ and $x_2 - 0$ yield no new expression.
  \item $x_2 - (-\frac12 x_1 + 1)$ and $0 - (-\frac12 x_1 + 1)$ with update $x_2:=0$ both
    yield expression $2$, that is added to $E_1$.
  \item $x_1 - (-\frac12 x_1 + 1)$ yields expression $\frac23$, which
    is also added to $E_1$.
\end{itemize}
The sets of expressions are therefore $E_1 = \{x_1,0,1,\frac23,2\}$
and $E_2 = \{x_2,0,-\frac12 x_1 + 1, x_1\}$.  Remark that knowing the
order between $x_1$ and $\frac23$ will allow us to know the order
between $-\frac12 x_1 + 1$ and $x_1$.  The class graph $\mathcal{G}$
corresponding to $\A_1$ and $\varphi_1$ is depicted in
\figurename~\ref{fig:exita1classes}.  Note that we replaced $x_1$ by
its value, since it is not changed by any update at level $2$.  Some
time zone notations used in $\mathcal{G}$ are displayed in
\tablename~\ref{tab:tzita1}.  In the class graph, states where the
comparison $x_2 > x_1$ is \emph{true} are greyed.  Among these, the
ones in which the class corresponds to state $q_1$ are doubly circled,
\emph{i.e.} states in which $q_1 \wedge (x_2 > x_1)$ is \emph{true}.
Applying standard \ctl model checking procedure on this graph, one can
prove that one of these states is reachable, hence proving that
$\varphi_1$ is \emph{true} on $\A_1$.

\begin{table}
\scriptsize
\begin{mathpar}
Z_0^1=(0 = x_1 < \frac23 < 1 < 2) \and
Z_1^1=(0 < x_1 < \frac23 < 1 < 2) \and
Z_2^1=(0 < x_1 = \frac23 < 1 < 2) \and
Z_3^1=(0 < \frac23 < x_1 < 1 < 2) \and
Z_4^1=(0 < \frac23 < x_1 = 1 < 2) \and
Z_5^1=(0 < \frac23 < 1 < x_1 < 2) \and
Z_6^1=(0 < \frac23 < 1 < x_1 = 2) \and
Z_7^1=(0 < \frac23 < 1 < 2 < x_1) \\
Z_0^2=(0 = x_2 < x_1 < -\frac12 x_1 + 1)\and
Z_1^2=(0 < x_2 < x_1 < -\frac12 x_1 + 1)\and
Z_2^2=(0 < x_1 = x_2 < -\frac12 x_1 + 1)\and
Z_3^2=(0 < x_1 < x_2 < -\frac12 x_1 + 1)\and
Z_4^2=(0 < x_1 < -\frac12 x_1 + 1 = x_2)\and
Z_5^2=(0 < x_1 < -\frac12 x_1 + 1 < x_2)\and
Z_6^2=(0 = x_2 < -\frac12 x_1 + 1 < x_1)\and
Z_7^2=(0 < x_2 < -\frac12 x_1 + 1 < x_1)\and
Z_8^2=(0 < -\frac12 x_1 + 1 = x_2 < x_1)\and
Z_9^2=(0 < -\frac12 x_1 + 1 < x_2 < x_1)\and
Z_{10}^2=(0 < -\frac12 x_1 + 1 < x_1 = x_2)\and
Z_{11}^2=(0 < -\frac12 x_1 + 1 < x_1 < x_2)\and
\end{mathpar}
\caption{Time zones used in the class graph of $\A_1$ when checking $\varphi_1$.}
\label{tab:tzita1}
\end{table}

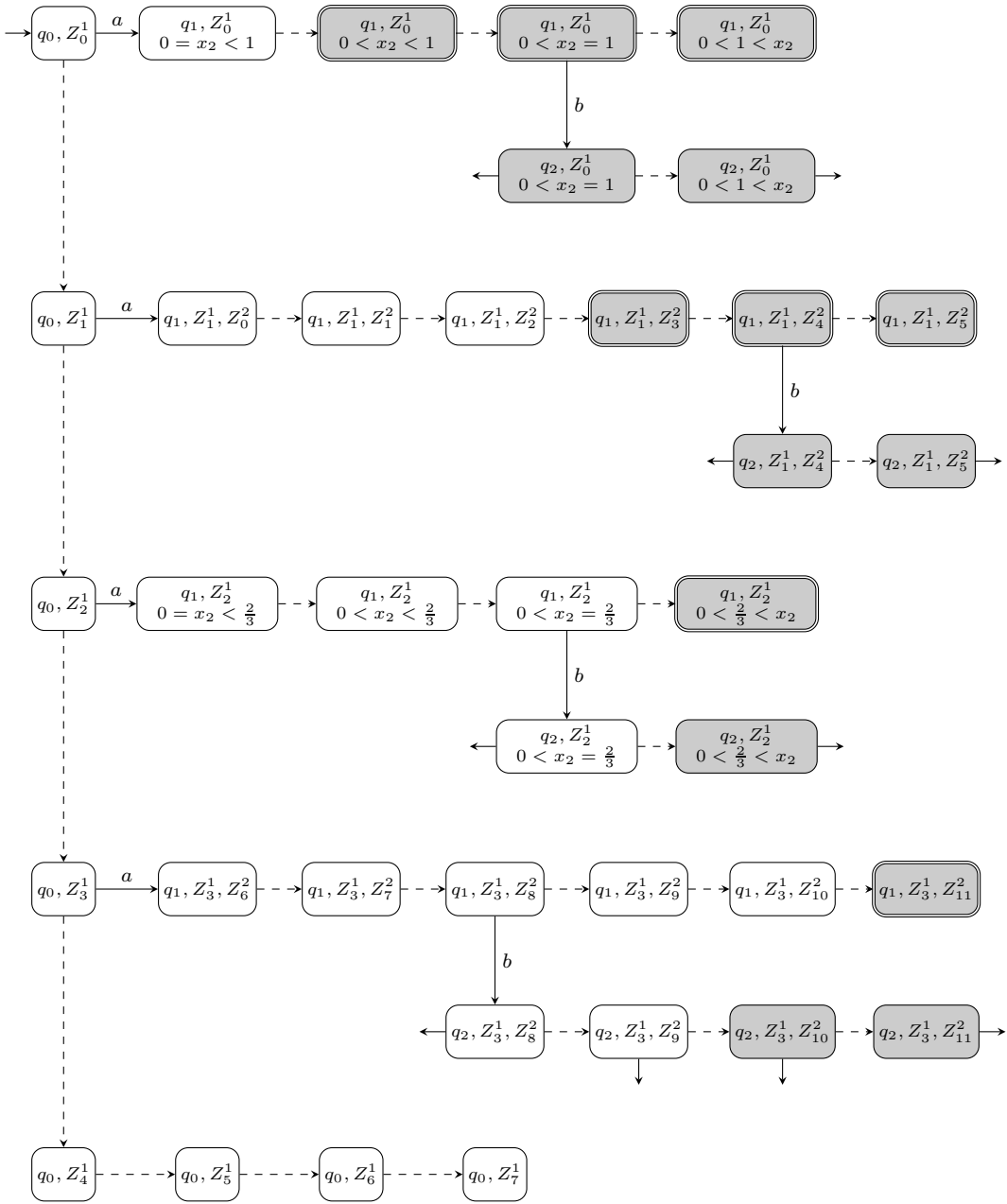
\begin{figure}
\begin{center}
\begin{tikzpicture}[node distance=2.5cm,auto,initial text=]
\tikzstyle{every state}=[inner sep=2pt,draw=black,shape=rectangle,rounded corners=5pt,font=\scriptsize]
\tikzstyle{time step}=[draw,->,dashed]
\tikzstyle{timeok}=[fill=black!20]
\tikzstyle{predicatetrue}=[timeok,double]

\useasboundingbox ($(-0.8,0.5) + (24.5pt,0)$) rectangle ($(13.1,-16.5) - (24.5pt,0)$); 

\node[state,initial] (r0) at (0,0)                          {$q_0,Z_0^1$};
\node[state]         (r1) [node distance=4cm,below of=r0]   {$q_0,Z_1^1$};
\node[state]         (r2) [node distance=4cm,below of=r1]   {$q_0,Z_2^1$};
\node[state]         (r3) [node distance=4cm,below of=r2]   {$q_0,Z_3^1$};
\node[state]         (r4) [node distance=4cm,below of=r3]   {$q_0,Z_4^1$};
\node[state]         (r5) [node distance=2cm,right of=r4] {$q_0,Z_5^1$};
\node[state]         (r6) [node distance=2cm,right of=r5] {$q_0,Z_6^1$};
\node[state]         (r7) [node distance=2cm,right of=r6] {$q_0,Z_7^1$};

\node[state] (p100) [node distance=2cm,right of=r0] {%
\begin{tabular}{cc}
$q_1,Z_0^1$\\$0 = x_2 < 1$
\end{tabular}
};
\node[state,predicatetrue] (p101) [right of=p100] {%
\begin{tabular}{cc}
$q_1,Z_0^1$\\$0 < x_2 < 1$
\end{tabular}
};
\node[state,predicatetrue] (p102) [right of=p101] {%
\begin{tabular}{cc}
$q_1,Z_0^1$\\$0 < x_2 = 1$
\end{tabular}
};
\node[state,predicatetrue] (p103) [right of=p102] {%
\begin{tabular}{cc}
$q_1,Z_0^1$\\$0 < 1 < x_2$
\end{tabular}
};
\node[state,accepting by arrow,accepting where=left,timeok] (p202) [node distance=2cm,below of=p102] {%
\begin{tabular}{cc}
$q_2,Z_0^1$\\$0 < x_2 = 1$
\end{tabular}
};
\node[state,accepting by arrow,timeok] (p203) [right of=p202] {%
\begin{tabular}{cc}
$q_2,Z_0^1$\\$0 < 1 < x_2$
\end{tabular}
};

\node[state] (p110) [node distance=2cm,right of=r1]   {$q_1,Z_1^1,Z_0^2$};
\node[state] (p111) [node distance=2cm,right of=p110] {$q_1,Z_1^1,Z_1^2$};
\node[state] (p112) [node distance=2cm,right of=p111] {$q_1,Z_1^1,Z_2^2$};
\node[state,predicatetrue] (p113) [node distance=2cm,right of=p112] {$q_1,Z_1^1,Z_3^2$};
\node[state,predicatetrue] (p114) [node distance=2cm,right of=p113] {$q_1,Z_1^1,Z_4^2$};
\node[state,predicatetrue] (p115) [node distance=2cm,right of=p114] {$q_1,Z_1^1,Z_5^2$};
\node[state,accepting by arrow,accepting where=left,timeok]
             (p214) [node distance=2cm,below of=p114]   {$q_2,Z_1^1,Z_4^2$};
\node[state,accepting by arrow,timeok]
             (p215) [node distance=2cm,right of=p214] {$q_2,Z_1^1,Z_5^2$};

\node[state] (p120) [node distance=2cm,right of=r2] {%
\begin{tabular}{cc}
$q_1,Z_2^1$\\$0 = x_2 < \frac23$
\end{tabular}
};
\node[state] (p121) [right of=p120] {%
\begin{tabular}{cc}
$q_1,Z_2^1$\\$0 < x_2 < \frac23$
\end{tabular}
};
\node[state] (p122) [right of=p121] {%
\begin{tabular}{cc}
$q_1,Z_2^1$\\$0 < x_2 = \frac23$
\end{tabular}
};
\node[state,predicatetrue] (p123) [right of=p122] {%
\begin{tabular}{cc}
$q_1,Z_2^1$\\$0 < \frac23 < x_2$
\end{tabular}
};
\node[state,accepting by arrow,accepting where=left] (p222) [node distance=2cm,below of=p122] {%
\begin{tabular}{cc}
$q_2,Z_2^1$\\$0 < x_2 = \frac23$
\end{tabular}
};
\node[state,accepting by arrow,timeok] (p223) [right of=p222] {%
\begin{tabular}{cc}
$q_2,Z_2^1$\\$0 < \frac23 < x_2$
\end{tabular}
};

\node[state] (p130) [node distance=2cm,right of=r3]   {$q_1,Z_3^1,Z_6^2$};
\node[state] (p131) [node distance=2cm,right of=p130] {$q_1,Z_3^1,Z_7^2$};
\node[state] (p132) [node distance=2cm,right of=p131] {$q_1,Z_3^1,Z_8^2$};
\node[state] (p133) [node distance=2cm,right of=p132] {$q_1,Z_3^1,Z_9^2$};
\node[state] (p134) [node distance=2cm,right of=p133] {$q_1,Z_3^1,Z_{10}^2$};
\node[state,predicatetrue] (p135) [node distance=2cm,right of=p134] {$q_1,Z_3^1,Z_{11}^2$};
\node[state,accepting by arrow,accepting where=left]
             (p232) [node distance=2cm,below of=p132]   {$q_2,Z_3^1,Z_8^2$};
\node[state,accepting by arrow,accepting where=below]
             (p233) [node distance=2cm,right of=p232] {$q_2,Z_3^1,Z_9^2$};
\node[state,accepting by arrow,accepting where=below,timeok]
             (p234) [node distance=2cm,right of=p233] {$q_2,Z_3^1,Z_{10}^2$};
\node[state,accepting by arrow,timeok]
             (p235) [node distance=2cm,right of=p234] {$q_2,Z_3^1,Z_{11}^2$};

\path[->,dashed] (r0) edge (r1);
\path[->,dashed] (r1) edge (r2);
\path[->,dashed] (r2) edge (r3);
\path[->,dashed] (r3) edge (r4);
\path[->,dashed] (r4) edge (r5);
\path[->,dashed] (r5) edge (r6);
\path[->,dashed] (r6) edge (r7);

\path[->]        (r0)   edge node {$a$} (p100);
\path[->,dashed] (p100) edge            (p101);
\path[->,dashed] (p101) edge            (p102);
\path[->,dashed] (p102) edge            (p103);
\path[->,dashed] (p202) edge            (p203);
\path[->]        (p102) edge node {$b$} (p202);

\path[->]        (r1)   edge node {$a$} (p110);
\path[->,dashed] (p110) edge            (p111);
\path[->,dashed] (p111) edge            (p112);
\path[->,dashed] (p112) edge            (p113);
\path[->,dashed] (p113) edge            (p114);
\path[->,dashed] (p114) edge            (p115);
\path[->]        (p114) edge node {$b$} (p214);
\path[->,dashed] (p214) edge            (p215);

\path[->]        (r2)   edge node {$a$} (p120);
\path[->,dashed] (p120) edge            (p121);
\path[->,dashed] (p121) edge            (p122);
\path[->,dashed] (p122) edge            (p123);
\path[->,dashed] (p222) edge            (p223);
\path[->]        (p122) edge node {$b$} (p222);

\path[->]        (r3)   edge node {$a$} (p130);
\path[->,dashed] (p130) edge            (p131);
\path[->,dashed] (p131) edge            (p132);
\path[->,dashed] (p132) edge            (p133);
\path[->,dashed] (p133) edge            (p134);
\path[->,dashed] (p134) edge            (p135);
\path[->]        (p132) edge node {$b$} (p232);
\path[->,dashed] (p232) edge            (p233);
\path[->,dashed] (p233) edge            (p234);
\path[->,dashed] (p234) edge            (p235);
\end{tikzpicture}
\end{center}
\caption{The class automaton for $\A_1$ and formula $\varphi_1$.}
\label{fig:exita1classes}
\end{figure}

\subsection{Model-checking \tctl with subscript}

%

Note that in \tctlcint, it is not possible to reason about time
evolution independently of the level in which actions are performed.
For example, properties \emph{(P2) the system is error free for at least 50 t.u.} or
\MS{Reviewer 1, Rq 20. Done.}
\emph{(P3) the system will reach a safe state within 7 t.u.} involve global time.
In order to verify such properties, we introduce the
fragment \tctlp.  This fragment is expressive enough to state
constraints on earliest (and latest) execution time of particular
sequences, like those reaching a recovery state after a crash.  \tctlp
is the set of formulas where satisfaction of an \emph{until} modality
over propositions can be parameterized by a restricted form of time
intervals.
\begin{definition}
Formulas of \tctlp are defined by the following grammar:
\[\varphi_p := p \mid \varphi_p \wedge \varphi_p \mid \neg \varphi_p \quad\mbox{and}\quad
\psi := 
\psi \wedge \psi \mid \neg \psi\mid \varphi_p  \mid
\always \varphi_p \untilsub{\rel a} \varphi_p \mid
\expath \varphi_p \untilsub{\rel a} \varphi_p\]
where $p \in AP$ is an atomic proposition, $a \in \mathbb{Q}^+$, and
$\rel \,\in \{>,\geq,\leq,<\}$ is a comparison operator.
\end{definition}
The properties given in introduction can be expressed by \tctlp
formulas as follows. Property $P2:$ \emph{the system is error free for
  at least 50 t.u.}  corresponds to $\always
(\neg$\emph{error})$\untilsub{\geq 50} \mathbf{t}$, while property
\MS{Reviewer 2, Rq 8. No need for correction.}
$P3:$ \emph{the system will reach a safe state within 7 t.u.} is
expressed by $\always \eventually_{\leq 7}$\emph{safe}.

Formulas of \tctlp are again interpreted over configurations of the
transition system associated with an ITA. For configuration
$s=(q,v,\beta)$, with $lab(s)=lab(q)$, the inductive definition is as
follows: 

\[\begin{array}{lcl}
s \models p &\ \textrm{iff}\ & p \in lab(s)\\
s \models \varphi \wedge \psi &\ \textrm{iff}\ & s \models \varphi \textrm{ and } s \models \psi \\
s \models \neg\varphi &\ \textrm{iff}\ & s \not\models \varphi \\
s \models \always\varphi_p \untilsub{\bowtie a} \psi_p &\ \textrm{iff}\ &\textrm{any execution }\rho \in Exec(s) \textrm{ is such that } \rho \models \varphi_p \untilsub{\bowtie a} \psi_p \\
s \models \expath\varphi_p \untilsub{\bowtie a} \psi_p &\ \textrm{iff}\ &\textrm{there exists an execution }\rho \in Exec(s) \textrm{ such that } \rho \models \varphi_p \untilsub{\bowtie a} \psi_p 
\end{array}\]

where
\[\begin{array}{lcl}
\rho \models \varphi_p \untilsub{\bowtie a} \psi_p &\quad\textrm{iff}\ &\textrm{there exists a position }\pi \textrm{ along } \rho \textrm{ such that } Dur(\rho^{\leq\pi}) \bowtie a,\\
&& s_{\pi} \models \psi_p, \textrm{ and for any position } \pi' <_\rho \pi,\, s_{\pi'} \models \varphi_p
\end{array}\]

Again $\A \models \psi$ if $s_0 \models \psi$.

 We now prove that:
\begin{theorem}\label{thm:mcsubscript}
  Model checking TCTL$_p$ on ITA is decidable.
\MS{Reviewer 2, Rq 9. Done.}
\end{theorem}
The proof consists in establishing procedures dedicated to the four different subcases:
\begin{itemize}
\item $\expath p \untilsub{\leq a} r$ and $\expath p \untilsub{< a} r$ (Proposition~\ref{lem:mcsubscriptexistinf}),
\item $\expath p \untilsub{\geq a} r$ and $\expath p \untilsub{> a} r$ (Proposition~\ref{lem:mcsubscriptexistsup}),
\item $\always p \untilsub{\geq a} r$ and $\always p \untilsub{> a} r$ (Proposition~\ref{lem:mcsubscriptallsup}),
\item $\always p \untilsub{\leq a} r$ and $\always p \untilsub{< a} r$ (Proposition~\ref{lem:mcsubscriptallinf}),
\end{itemize}
where $p$ and $r$ are boolean combinations of atomic propositions.

\begin{proposition}\label{lem:mcsubscriptexistinf}
  Model checking formulas $\expath p \untilsub{\leq a} r$ and $\expath
  p \untilsub{< a} r$ over ITA is decidable in NEXPTIME and in NP
  if the number of clocks is fixed.
\end{proposition}

\begin{proof}
  First consider the case of ITA$_-$.
  Both formulas are variants of reachability, with the addition of a
  time bound. Therefore, the proof is similar to the one of
  Proposition~\ref{proposition:reachitamoins}. Again using
  Lemma~\ref{lemma:counting} on an ITA$_-$ with $E$ transitions,
  we can look for a run satisfying one of these formulas and bounded
  by $B = (E + n)^{3n}$, because shortening longer runs can be can be
  done while preserving the property.  Thus, the decision procedure
  again consists in guessing a path and building a linear program.
  The satisfaction of the formula is then checked by separately
  verifying on one side that the run satisfies $p \until r$, and on
  the other side, that the sum of all delays $d_j$ satisfies the
  constraint in the formula. The complexity is the same as in
  Proposition~\ref{proposition:reachitamoins}.

In the case of ITA, the exponential blowup of the transformation into an equivalent ITA$_-$ does not affect the complexity of the model-checking procedure above, as in Theorem~\ref{thm:optreach}.
\qed
\end{proof}
Note that this problem can be compared with bounded reachability as
studied in~\cite{brihaye11}. However, the models seem incomparable:
while the variables (that have fixed non-negative rates in a state) 
are more powerful than interrupt clocks, 
the guards and updates are rectangular,
which in particular forbids additive and diagonal constraints.

\begin{proposition}\label{lem:mcsubscriptexistsup}
  Model checking a formula $\expath p \untilsub{\geq a} r$ and
  $\expath p \untilsub{> a} r$ on an ITA is decidable in NEXPTIME
  and in NP if the number of clocks is fixed.
\end{proposition}

\begin{proof}
  Let $\A$ be an ITA$_-$ with $n$ interrupt clocks and $E$
  transitions, and $B = (E + n)^{3n}$.
  The algorithm to decide whether $\expath p \untilsub{\geq a} r$ (or $\expath p \untilsub{> a} r$) works as follows.
  It nondeterministically guesses a path of length smaller than or equal to $B$ and builds the associated linear program (as in the proof of Proposition~\ref{proposition:reachitamoins}), then checks that:
  \begin{itemize}
  \item this path yields a run, which can be done by solving the linear program;
  \item there is a position $\pi$ in this run at which $r$ holds and before which $p$ holds continuously;
  \item the sum of delays before $\pi$ exceeds $a$ (or strictly exceed in the case of $\expath p \untilsub{> a} r$).
  \end{itemize}
  If this first procedure fails, the algorithm nondeterministically guesses a path of length smaller or equal to $2B+1$ and checks that:
  \begin{itemize}
  \item this path yields a run, which can be checked by a linear program
  as before,
  \item $p$ holds on this path, but not necessarily in the last state reached,
  \item $r$ holds in the last state of this path,
  \item either there is a transition $e$ of level $k$ that updates $x_k$ appearing twice and
    separated by a sequence $\sigma$ of transitions of level higher than $k$ during which time elapses (globally) ;
    this last part can be checked with a linear program on the delays
    corresponding to this subrun.
  \item or there is a transition $e$ of level $k$ that does not update $x_k$ appearing twice and
    separated by a sequence $\sigma$ of transitions of level higher than $k$ not updating $x_k$ during which time elapses at levels strictly higher than $k$ but not at level $k$.
  \end{itemize}
The algorithm returns \emph{true} if one of the previous procedure succeeds, and \emph{false} otherwise.
We shall now prove that this algorithm is both sound and complete.

\paragraph{Soundness.}
If the first procedure succeeds, then the path guessed is trivially a
witness of $\expath p \untilsub{\geq a} r$ (or $\expath p
\untilsub{>a} r$, accordingly).  If the second procedure succeeds,
then a witness for the formula can be built from the path guessed.
Indeed, the path guessed satisfies $p \until r$, but not necessarily
$p \untilsub{\geq a} r$.  Assume the sequence $\sigma$ lets elapse
$\delta$ time units ($\delta >0$), by repeating
$\lceil\frac{a}{\delta}\rceil$ times\footnote{This sequence may be
  repeated once more in the case of $p \untilsub{> a} r$.} the
sequence $\sigma e$, we obtain a run satisfying $p \untilsub{\geq a}
r$.  Note that since either $e$ updates the clock $x_k$ or there are no updates nor time elapsing at level $k$, and $\sigma$ happens
at higher levels, the clock values in each instance of $\sigma e$ will
be identical, hence this repetition will always be possible.

\paragraph{Completeness.}
Now consider a minimal witness $\rho$ of length $h$ for $\expath p
\untilsub{\geq a} r$.  Since $\rho$ is minimal, $r$ holds in the last
state of $\rho$ and $p$ holds (at least) in every position before.  If
$h\leq B$, then the first procedure will consider $\rho$.  Otherwise,
$h>B$, it means that one of the following cases of
Lemma~\ref{lemma:counting} happens:
\begin{itemize}
\item The same transition $e$ of level $k$ leaving $x_k$ unchanged appears twice
  separated by lazy or delayed transitions between states of level
  greater than or equal to $k$.  In that case, the corresponding
  subrun can be replaced by a time step of the same duration, not
  changing the truth value of $p \untilsub{\geq a} r$ on this new
  smaller run, thus violating the minimality hypothesis.
\item The same transition $e$ of level $k$ updating clock $x_k$
  appears twice on the subrun $e_1 \dots e_{B+1}$, at positions $i$
  and $j$.  In that case we have to distinguish two subcases either
  some time has elapsed between the two occurrences $e_i$ and
  $e_j$ of $e$, or the transitions were all instantaneous.
  \begin{itemize}
    \item If no time has elapsed, the subrun between $e_i$ and $e_j$
    can be removed without altering the truth value of $p \untilsub{\geq
      a} r$ on this new run, which is smaller than $\rho$.  Hence there
    is a contradiction with the minimality hypothesis.
  \item Or some time elapsed during this subrun.  Let $\rho$ be
    decomposed into $\rho_0 e_i \sigma e_j \rho_j$.  Then by applying
    Lemma~\ref{lemma:counting} to $\rho_j$ there exists a run
    $\rho_j'$ of length smaller or equal to $B$ such that
    $\rho'=\rho_0 e_i \sigma e_j \rho_j'$ is also a run.  Note that
    $|\rho'| \leq 2B+1$, that the last state of $\rho'$ will be the
    same as the last state of $\rho$ hence will satisfy $r$, and that
    $p$ will also hold along $\rho'$.  As a result $\rho'$ will be
    considered by the second procedure.
  \end{itemize}
\item The same transition $e$ of level $k$ leaving $x_k$ unchanged appears twice, with no time elapsing at level $k$ between these occurrences.  In that case, we again distinguish two subcases:
\begin{itemize}
\item either no time elapsed (globally) the corresponding subrun can be removed,
  not changing anything to the rest of the execution nor to the
  satisfaction of $p \untilsub{\geq a} r$, thus violating the
  hypothesis of minimality of $\rho$;
\item or time elapsed at higher levels and, by minimizing the subrun after the second occurrence as above, we deduce that the run will be considered by the second procedure.
\end{itemize}
\end{itemize}
The completeness proof is similar in the case of $\expath p \untilsub{> a} r$.

When $\A$ is an ITA, the exponential blowup of the transformation from ITA to ITA$_-$ does not affect the above complexity.
\qed
\end{proof}

\medskip While a witness is a finite path in the previous cases, it is
potentially infinite for $\always p \untilsub{\geq a} r$ or $\always p
\untilsub{> a} r$.
The generation of an infinite run relies on the (nondeterministic)
exploration of the class graph built in Section~\ref{sec:regular},
thus has a much greater computational complexity.

\begin{proposition}\label{lem:mcsubscriptallsup}
  Model checking a formula $\always p \untilsub{\geq a} r$ and
  $\always p \untilsub{> a} r$ on an ITA is decidable in 2-EXPTIME
  and in co-NP if the number of clocks is fixed.
\end{proposition}
\begin{proof} We consider an ITA $\A$ with $n$ interrupt
  clocks, $E$ transitions and the bound $B = (n+2)^{12 b \cdot E \cdot n^3}$ where $b$ is the number of bits coding the constants in $\A$.
  
  The algorithm to verify $\always p \untilsub{\geq a} r$ (or $\always
  p \untilsub{> a} r$) works as follows.  It nondeterministically
  guesses a path of length smaller than or equal to $B$, builds its
  associated linear program, and checks that:
\begin{itemize}
\item this path yields a run $\rho$ (by solving the linear program);
\item this path is maximal, that means no transition can be fired from
  the last configuration of the run;
\item there is a position $\pi$ in $\rho$ occurring at a time stricly less
  than\footnote{Less than or equal to $a$ in the case of $\always p
    \untilsub{> a} r$.} $a$ such that
\begin{description}
\item[Case 1:] either $r$ does not hold from $\pi$ (see
  \figurename~\ref{fig:finiteneverr})
\item[Case 2:] or there is a position $\pi'$ where neither $p$ not $r$
  hold, and $r$ does not hold between $\pi$ and $\pi'$ (see
  \figurename~\ref{fig:finitenotpandnotr}).
\end{description}
\end{itemize}
  \begin{figure}
  \centering
  \begin{tikzpicture}[auto]
  \node (zero) at (0,0) {};
  \node (pi) at (4,0) {};
  \node (notboth) at (8,0) {};
  \node (end) at (10,0) {};
  \node (below) at (0,-0.25) {};
  \node (verybelow) at (0,-0.75) {};

\path[draw,|-|] (zero.center) -- (pi.center);
\path[draw,-|] (pi.center) -- node[swap] {$\neg r$} (end.center);
\node[anchor=east] at (zero) {$\rho$:};

\node[anchor=south] at (pi.north) {$\pi$};

\path[interval duration] (below -| zero) -- node[swap] {\timespent\,$<a$} (below -| pi);
\path[interval length] (verybelow -| zero.west) -- node[swap] {\ruler$\leq B$} (verybelow -| end.east);
  \end{tikzpicture}
  \caption{Proof of Proposition~\ref{lem:mcsubscriptallsup}: finite counterexample (Case~1).}
  \label{fig:finiteneverr}
  \end{figure}
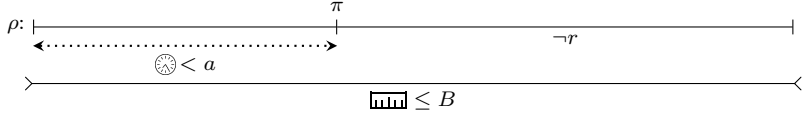
  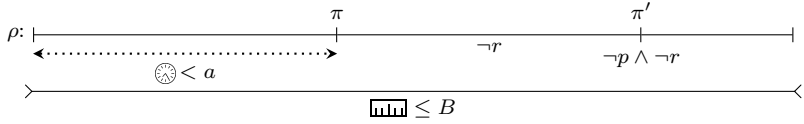
\begin{figure}
  \centering
  \begin{tikzpicture}[auto]
  \node (zero) at (0,0) {};
  \node (pi) at (4,0) {};
  \node (notboth) at (8,0) {};
  \node (end) at (10,0) {};
  \node (below) at (0,-0.25) {};
  \node (verybelow) at (0,-0.75) {};

  \path[draw,|-|] (zero.center) -- (pi.center);
   \path[draw,-|] (pi.center) -- node[swap] {$\neg r$} (notboth.center);
   \path[draw,-|] (notboth.center) -- (end.center);

\node[anchor=south] at (pi.north) {$\pi$};
\node[anchor=north] at (notboth.south) {$\neg p \wedge \neg r$};
\node[anchor=south] at (notboth.north) {$\pi'$};
\node[anchor=east] at (zero) {$\rho$:};

\path[interval duration] (below -| zero) -- node[swap] {\timespent\,$<a$} (below -| pi);
\path[interval length] (verybelow -| zero.west) -- node[swap] {\ruler$\leq B$} (verybelow -| end.east);
  \end{tikzpicture}
  \caption{Proof of Proposition~\ref{lem:mcsubscriptallsup}: finite counterexample (Case~2).}
  \label{fig:finitenotpandnotr}
  \end{figure}
If this first procedure fails, then the algorithm guesses:
\begin{itemize}
\item a class $K$ and a cycle $\cal C$ starting from $K$
  in the class graph (without building neither the graph nor the cycle), 
  such that $\cal C$ contains
  at least a discrete step and only traverses classes where $\neg r$
  holds;
\item a path in the automaton of length smaller than or equal
  to the bound $B$;
\end{itemize}
and checks that:
\begin{itemize}
\item the path does yield a run $\rho$, that reaches a configuration
  $(q,v,\beta)$ in class $K$ (through a linear program);
\item there is a position $\pi$ in $\rho$ occurring at time strictly
  less than\footnote{Less than or equal to $a$ in the case of $\always
    p \untilsub{> a} r$.} $a$ after which $r$ no longer holds.
\end{itemize}
Remark that the procedure cannot use solely the class graph, since the
abstraction is not precise enough to check the existence of position
$\pi$.

\paragraph{Soundness.}
We prove that the algorithm is sound: when one of the procedures
succeeds, there exists a counterexample for formula $\always p
\untilsub{\geq a} r$ (or $\always p \untilsub{> a} r$).  In the case
of the first procedure, it is trivial that the guessed run does not
satisfy $p \untilsub{\geq a} r$ (or $p \untilsub{> a} r$).  In the
case of the second one, we show that there exists an infinite
counterexample.  Consider
configuration $(q,v,\beta)$, which is reachable by $\rho$.  Since
$(q,v,\beta)$ belongs to class $K$, for any path $\sigma$ starting
from $K$ in the class graph, there is a run in the automaton starting
from $(q,v,\beta)$ traversing configurations which belong to the
classes traversed by $\sigma$. Since there is a cycle in the class
graph, there is an infinite path in the class graph (iterating on this
cycle), so there exists an infinite run in the ITA.  Also, since $\neg r$
holds in the infinite path of the class graph, it holds in the run of
the ITA, and the run is a counterexample for the formula.

\paragraph{Completeness.}
Assume there exists a finite counterexample $\rho$.  Let $\A'$ be the
ITA$_-$ accepting the same timed language as $\A$ and let $E'$ denote
the number of its transitions. Let $B' = (E'+2n)^{3n}$ (the bound of Lemma~\ref{lemma:counting}), 
we have $B'\leq B$.  If  $|\rho| \leq B$, it will be detected by procedure 1.
Otherwise let $\rho'$ be the run corresponding to $\rho$ in $\A'$.
This run accepts the same timed word as $\rho$ and its
sequence of traversed states can be projected onto the sequence of
corresponding states of $\rho$, by omitting states of the form
$(q^-,-)$: any subsequence $(q^+_0,-) \rightarrow \cdots \rightarrow
(q^+_{m-1},-) \rightarrow (q^-_m,-) \rightarrow (q^+_m,-)$ in $\rho'$
corresponds to the subsequence $q_0 \rightarrow \cdots \rightarrow
q_{m-1} \rightarrow q_m$ in $\rho$.  Note that $|\rho| \leq |\rho'|$ 
and that $\rho'$ is also a counterexample for the
formula (although in $\A'$).  Since $|\rho'| > B \geq B'$, then one of
the cases of case of Lemma~\ref{lemma:counting} occurs.  By removing
transitions and maybe replacing them by some time elapsing, as in the
proof of Proposition~\ref{lem:mcsubscriptexistsup}, a counterexample
$\sigma'$ of size $|\sigma'| \leq B' \leq B$ exists in $\A'$.  Now
consider the run $\sigma$ in $\A$ which corresponds to $\sigma'$.  We
have $|\sigma| \leq |\sigma'| \leq B' \leq B$ and
$\sigma$ is still a counterexample.  Therefore $\sigma$ can be guessed
by the first procedure.

If there exists an infinite counterexample $\rho$, consider its
counterpart $\sigma$ in the class graph.  This counterpart is also
infinite.  More precisely, $\sigma$ contains an infinite number of
discrete transitions.  Since $\sigma$ traverses a finite number of
classes, it contains a cycle $\cal C$ with at least one discrete
transition.  Choose any class $K$ of this cycle
and consider the prefix $\rho_0$ of $\rho$ leading to a configuration
in $K$.  As in the case of a finite counterexample, there exists
$\rho_0'$ of length smaller than $B$ reaching the same configuration.
All $\cal C$, $K$ and $\rho_0'$ can be guessed by the second
procedure, which will therefore succeed.

Procedure 1 operates in NEXPTIME (guessing a path of length $B$
and solving a linear program of size polynomial w.r.t. $B$). 
Procedure 2 consists in
looking for a specific cycle in the class graph which in can be done
in time polynomial w.r.t. the size of the graph thus in 2-EXPTIME. 
The case where the clocks are fixed, is handled as usual.
\qed
\end{proof}

For formulas in case 4, a specific procedure can be avoided,
since the algorithms of cases 2 and 3 can be reused:
\begin{proposition}\label{lem:mcsubscriptallinf}
  Model checking a formula $\always p \untilsub{\leq a} r$ and
  $\always p \untilsub{< a} r$ on an ITA is decidable in
  2-EXPTIME and in co-NP if the number of clocks is fixed.
\end{proposition}
\begin{proof}
  Notice that $\always p \untilsub{\leq a} r = (\always p \untilsub{\geq0} r)
  \wedge \neg (\expath \neg r \untilsub{> a} \mathbf{t})$, and $\always p
  \untilsub{< a} r = (\always p \untilsub{\geq0} r) \wedge \neg (\expath \neg r
  \untilsub{\geq a} \mathbf{t})$.
\qed
\end{proof}

\section{Language properties}\label{sec:exp}

In this section, we compare the expressive power of the previous
models with respect to language acceptance. Recall that TL is strictly
contained in CRTL. We prove that:
\begin{theorem}
The families TL and ITL are incomparable. The families CRTL
and ITL are incomparable.
\end{theorem}

\subsection{ITL is not contained in TL, nor in CRTL}
The next proposition shows that ITA cannot be reduced to TA or
CRTA. Observe that the automata used in the proof belong to
ITA$_-$. Also, the language given for the first point of the
proposition is very simple since it contains only words of length 2.
\begin{proposition}\label{prop:comp}~\\
1. There exists a language in ITL whose words have bounded length
 which is not in TL.\\
2. There exists a language in ITL which is not in CRTL.
\end{proposition}

\begin{proof}
  To prove the first point, consider the ITA $\A_3$ in
  Fig.~\ref{fig:ctrex1}.  Suppose, by contradiction, that $L_3 =
  \La(\A_3)$ is accepted by some timed automaton $\Ba$ (possibly
\MS{Reviewer 1, Rq 22. Done.}
  with $\eps$-transitions). 
Note that since we consider timed languages, we cannot assume that the granularity of $\Ba$ is $1$.
\MS{Reviewer 1, Rq 23. Done.}
Let $d$ be the granularity of $\Ba$,
  \textit{i.e.} the gcd of all rational constants appearing in the
  constraints of $\Ba$ (thus each such constant can be written $k/d$
  for some integer $k$). Then the word $w=(a, 1-1/d)(b, 1-1/2d)$ is
  accepted by $\Ba$ through a finite path. Consider now the automaton
  $\Ba'$ in TA, consisting of this single path (where states may have
  been renamed).  We have $w \in \La(\Ba') \subseteq \La(\Ba)=L_3$ and
  $\Ba'$ contains no cycle.  Using the result in~\cite{berard98}, we
  can build a timed automaton $\Ba''$ without $\eps$-transition and
  with same granularity $d$ such that $\La(\Ba'')=\La(\Ba')$, so that $w
  \in \La(\Ba'')$. The accepting path for $w$ in $\Ba''$ contains two
  transitions : $p_0 \tr{\fee_1, a, r_1} p_1 \tr{\fee_2, b, r_2}
  p_2$. After firing the $a$-transition, all clock values are $1- 1/d$
  or $0$, thus all clock values are $1-1/2d$ or $1/2d$ when the
  $b$-transition is fired. Let $x \rel c$ be an atomic proposition
  appearing in $\fee_2$. Since the granularity of $\Ba''$ is $d$, the
  $\rel$ operator cannot be $=$ otherwise the constraint would be $x =
  1/2d$ or $x= 1- 1/2d$. If the constraint is $x<c$, $x\leq c$, $x>c$,
  or $x\geq c$, the path will also accept some word $(a, 1-1/d)(b, t)$
  for some $t \neq 1-1/2d$. This is also the case if the constraint
  $\fee_2$ is true. We thus obtain a contradiction with $\La(\Ba'')
  \subseteq L_3$, which ends the proof.
\begin{figure}
\centering
\begin{tikzpicture}[node distance=4cm,auto]
\node[state,initial] (q0) at (0,0) {$q_0,1$};
\node[state] (q1) [right of=q0] {$q_1,2$};
\node[state,accepting] (q2) [right of=q1] {$q_2,2$};

\path[->] (q0) edge node {$x_1 < 1$, $a$, $(x_2:=0)$} (q1);
\path[->] (q1) edge node (tr) {$x_1 + 2 x_2 = 1$, $b$} (q2);
\end{tikzpicture}
\caption{An ITA $\A_3$ for $L_3$}
\label{fig:ctrex1}
\end{figure}
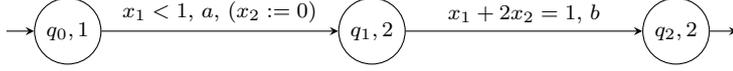

\medskip To prove the second point, consider the language: $$L_4= \{
(c,\tau) (c, 2\tau) \ldots (c, n\tau) \mid n \in \N, \tau > 0 \}$$
accepted by the ITA $\A_4$ in Fig.~\ref{fig:ctrex2}.
This language cannot be accepted by a CRTA (see~\cite{Zielonka}).
\qed
\begin{figure}
\centering
\begin{tikzpicture}[node distance=4cm,auto]
\useasboundingbox (-0.825,0.8) rectangle (8,-0.5); 

\node[state,initial] (p0) at (0,0) {$q_0, 1$};
\node[state,accepting,accepting where=above] (p1) [right of=p0] {$q_1,2$};

\path[->] (p0) edge node {$x_1>0$, $c$, $x_2:=0$} (p1);
\path[->] (p1) edge [out=30,in=-30,min distance=1.5cm,looseness=2] node {$x_2=x_1$, $c$, $x_2:=0$} (p1);
\end{tikzpicture}
\caption{An ITA $\A_4$ for $L_4$}
\label{fig:ctrex2}
\end{figure}
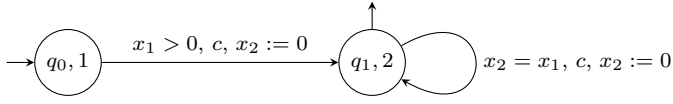
\end{proof}

\subsection{TL is not contained in ITL} 
\label{subsec:incomp}

\MS{\`A relire; traduction depuis la th\`ese.}
We now prove that there exists a language in TL that does not belong
to ITL.
Let $L_5$ be the language defined by
\begin{eqnarray*}
L_5=\big\{ (a,\tau_1)(b,\tau_2) &\ldots& (a,\tau_{2p+1})(b,\tau_{2p+2}) 
\mid p \in \N,\\
&&\forall 0\leq i\leq p,\ \tau_{2i+1}=i+1 \mbox{ and }  i+1<\tau_{2i+2}<i+2, \\
&&\forall 1\leq i\leq p\ \tau_{2i+2} - \tau_{2i+1}< \tau_{2i} - \tau_{2i-1} \big\}
\end{eqnarray*}
Hence, the untimed language of $L_5$ is $(ab)^*$, there is an
occurrence of $a$ at each time unit and the successive occurrences of
$b$ come each time closer to the occurrence of $a$ than previously.
This language is in TL as can be checked on the TA $\A_5$ of
\figurename~\ref{fig:ctrex3} (first proposed in~\cite{alur94a}).

\begin{figure}[ht]
\centering
\begin{tikzpicture}[node distance=4.75cm,auto]
\tikzstyle{every state}+=[scale=0.75]
\node[state,initial] (q0) at (0,0) {};
\node[state] (q1) [right of=q0] {};
\node[state,accepting,accepting where=above] (q2) [right of=q1] {};
\node[state] (q3) [node distance=5cm,right of=q2] {};

\path[->] (q0) edge node {$z=1$, $a$, $z:=0$} (q1);
\path[->] (q1) edge node {$0<z<1$, $b$, $y:=0$} (q2);
\path[->] (q2) edge [bend left,looseness=0.75] node {$z=1$, $a$, $z:=0$} (q3);
\path[->] (q3) edge [bend left,looseness=0.75] node {$0<z \wedge y<1$, $b$, $y:=0$} (q2);
\end{tikzpicture}
\caption{A timed automaton $\A_5$ for $L_5$}
\label{fig:ctrex3}
\end{figure}
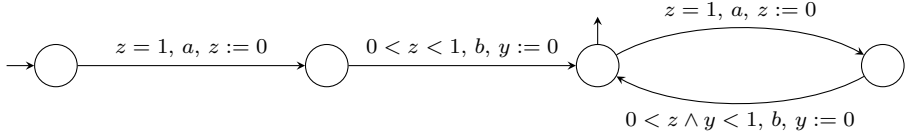

\begin{proposition}
The language $L_5$ does not belong to ITL.
\end{proposition}
%
%
\begin{proof}
Assume, by contradiction, that $L_5$ belongs to ITL.
Then $L_5$ is accepted by an ITA$_-$ $\A$ with $n$ clocks and $E$ transitions.
Let $B = (E+n)^{3n}$ and consider the timed word $w = (a,\tau_1) (b,\tau_2) \cdots (a,\tau_{2B+1}) (b,\tau_{2B+2}) \in L_5$.
Word $w$ is accepted by a run $\rho$ of $\A$, which can be assumed of minimal size.
However, we know that $|\rho| > B$, so one of the three cases of Lemma~\ref{lemma:counting} occurs in the $B$ first transitions.
\begin{itemize}[topsep=-\baselineskip]
\item Suppose a transition $e$ of level $k$ that updates $x_k$ appears twice, separated by a subrun $\sigma$ of level greater than or equal to $k$.
Remark that the valuations after the first and the second occurrence of $e$ are identical.
We distinguish several subcases, depending on the word read along $\sigma e$.
\begin{itemize}
\item If $\sigma e$ reads the empty word $\eps$, we write $\delta$ for the time spent during $\sigma e$.
If $\delta = 0$, then $\sigma e $ can be deleted without affecting neither the remainder of the run nor the accepted word, which contradicts the minimality of $\rho$.
If $\delta \geq 1$, then some interval $[i,i+1]$ does not contain any $b$, which contradicts the definition of $L_5$.
Otherwise, $0 < \delta < 1$.
By deleting $\sigma e$, we obtain an execution $\rho'$ (accepted by $\A$) in which the suffix after $e$ is shifted by $\delta$.
Therefore the following occurrence of letter $a$, which appeared in $\rho$ at date $i \in \N\setminus\{0\}$, appears in $\rho'$ at date $i - \delta$ which is not integral.
So the word accepted by $\rho'$ is not in $L_5$, which is a contradiction.
\item If $\sigma e$ reads more $a$s than $b$s or more $b$s than $a$s, by deleting $\sigma e$ we obtain a run accepting a word whose untiming is not in $(ab)^*$ thus does not belong to $L_5$.
\item If $\sigma e$ reads as many $a$s as $b$s (and both letters at least once), by duplicating $\sigma e$ we obtain a run accepting a word where a same duration separates an $a$ from the following $b$ is repeated, thus violating the definition of $L_5$.
\end{itemize}
\item Suppose a transition $e$ of level $k$ delayed or lazy occurs twice, separated by a subrun $\sigma$ of level greater than or equal to $k$, such that $\sigma e$ does not update $x_k$.
Then we can replace $e \sigma$ by a time step of the same duration and obtain a new run $\rho'$, accepted by $\A$.
\begin{itemize}
\item If $e \sigma$ reads $\eps$, then $\rho'$ contradicts the minimality of $\rho$.
\item If $e \sigma$ reads the word $b$, then $\rho'$ accepts a word where $a$ and $b$ do not alternate, thus not in $L_5$.
\item If $e \sigma$ reads at least an $a$, then $\rho'$ accepts a word with no $a$ at a given integral date, therefore not in $L_5$.
\end{itemize}
\item Otherwise, a transition $e$ of level $k$ appears twice separated by a subrun $\sigma$ of level greater than or equal to $k$, such that $\sigma e$ does not update $x_k$ nor lets time elapse at level $k$.
The same disjunction as in the case of an update of $x_k$ can be applied, since $\sigma e$ can either be deleted or duplicated.
\end{itemize}
\end{proof}

Note that the feature preventing $L_5$ to be in ITL lies in the
decreasing delays between the $a$'s and their immediately following
$b$. A language in ITL can record $k$ different constant delays, using
$k+1$ clocks.  For instance on the alphabet $\Sigma=\{a_1, \ldots,
a_k\}$, the language
\begin{eqnarray*}
M_k &=& \{(a_1,\tau_1) \ldots (a_k,
\tau_k)(a_1,\tau_1+1) \ldots (a_k, \tau_k+1) \ldots (a_1,\tau_1+n)
\ldots (a_k, \tau_k+n) \\&&\qquad\mid n \geq 1, \tau_1 \leq \tau_2 \leq \cdots \leq \tau_k \leq
\tau_1 + 1\}
\end{eqnarray*} is accepted by an ITA$_-$ with $k+1$ clocks. \figurename~\ref{fig:ex4} illustrates the case where $k=3$, with all states lazy.
We conjecture that $M_k$ cannot be accepted by an ITA with $k$ clocks.
\MS{Reviewer 1, Rq 21. Done.}

\begin{figure}[ht]
\centering
\begin{tikzpicture}[auto,node distance=5cm]
\node[state,initial] (q0) {$q_0,1$};
\node[state,right of=q0] (q1) {$q_1,2$};
\node[state,right of=q1] (q2) {$q_2,3$};
\node[state,accepting,node distance =3cm,below of=q2] (q3) {$q_3,4$};
\node[state,left of=q3] (q4) {$q_4,4$};
\node[state,left of=q4] (q5) {$q_5,4$};

\path[->] (q0) edge node {$a_1$, $(x_2:=0)$} (q1);
\path[->] (q1) edge node {$x_1 + x_2 < 1$, $a_2$, $(x_3:=0)$} (q2);
\path[->] (q2) edge node [swap,pos=0.25] {$x_1 + x_2 + x_3 < 1$, $a_3$, $(x_4:=0)$} (q3);
\path[->] (q3) edge node {$x_4 = 1 - x_2 - x_3$, $a_1$} (q4);
\path[->] (q4) edge node {$x_4 = 1 - x_3$, $a_2$} (q5);
\path[->] (q5) edge [bend left,looseness=0.6] node {$x_4 = 1$, $a_3$, $x_4:=0$} (q3);
\end{tikzpicture}
\caption{An interrupt timed automaton for $M_3$}
\label{fig:ex4}
\end{figure}
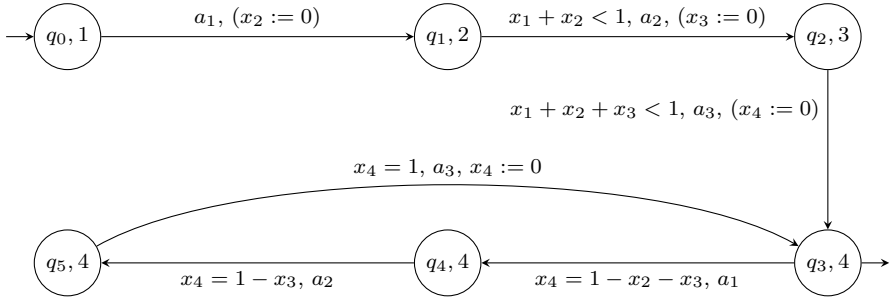

\subsection{Closure under complementation and intersection} 
\label{subsec:complementation}

\begin{proposition}
ITL is not closed under complementation. 
\end{proposition}
\begin{proof}
  We prove that the complement $L_5^c$ of $L_5$ belongs to ITL.  A
  timed word belongs to $L_5^c$ iff one of the following assertions
  hold:
\begin{enumerate}
\item An $a$ occurs not at a time unit.
\item An $a$ is missing at some time unit that precedes some letter
 of the word.
\item A $b$ occurs at a time unit.
\item There is no $b$ in an interval $[i,i+1]$ with an $a$ at time 
$i \in \N$.
\item There are two $b$s in an interval $[i,i+1]$ with an $a$ at
 time $i \in \N$.
\item There is an occurrence of $abab$ such that the time difference
 between the two first occurrences is smaller than or equal to the
 time difference between the two last occurrences.
\end{enumerate}
Since ITL is trivially closed under union, it is enough to prove that
each assertion from the set above can be expressed by an ITA. The five
first assertions are straightforwardly modeled by an ITA with a
single clock (and $\eps$-transitions) and we present in
\figurename~\ref{fig:cas6} an ITA with two clocks corresponding to the
last one.\qed

\begin{figure}[ht]
\centering
\begin{tikzpicture}[auto,node distance=3cm]
\node[state,initial] (q0) {$q_0,1$};
\node[state,right of=q0] (q1) {$q_1,1$};
\node[state,node distance=2.25cm,above of=q1] (q2) {$q_2,2$};
\node[state,right of=q2] (q3) {$q_3,2$};
\node[state,right of=q3,accepting] (q4) {$q_4,2$};

\path[->] (q0) edge [loop above,out=120,in=60,min distance=1.5cm] node {$a,b$} (q0);
\path[->] (q0) edge node {$a$, $x_1:=0$} (q1);
\path[->] (q1) edge node [swap] {$b$, $x_2:=0$} (q2);
\path[->] (q2) edge node {$a$, $x_2:=0$} (q3);
\path[->] (q3) edge node {$x_1 \leq x_2$, $b$} (q4);
\path[->] (q4) edge [loop below,out=-60,in=-120,min distance=1.5cm] node {$a,b$} (q4);
\end{tikzpicture}
\caption{An ITA for the language defined by assertion $6$}
\label{fig:cas6}
\end{figure}
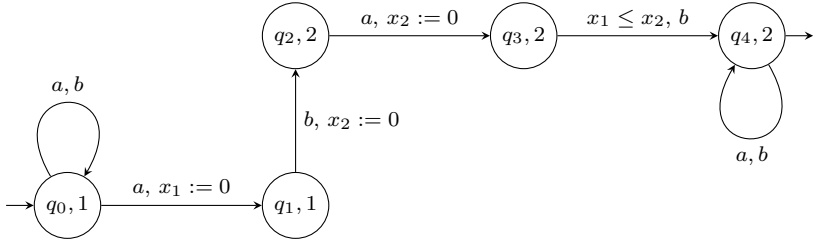
\end{proof}

\begin{proposition}
ITL is not closed under intersection. 
\end{proposition}
\begin{proof}
$L_5$ is the intersection of $L'_5$ and $L''_5$ defined as follows:
\begin{itemize}
\item The words of $L'_5$ are $(a,1)(b,\tau_1)\ldots(a,n)(b,\tau_n)$, 
with $i<\tau_i<i+1$ for all $i$, $1 \leq i \leq n$. 
\item The words of $L''_5$ are 
$(a,\tau'_1)(b,\tau_1)\ldots(a,\tau'_n)(b,\tau_n)$, 
with $\tau_{i+1}-\tau_i<1$ for all $i$, $1 \leq i \leq n-1$.
\end{itemize}
Both languages are accepted by one-clock ITA (which are also one-clock
TA).  In case of $L'_5$, (1) the clock is reset at every occurrence of
an $a$; (2) an $a$ must occur when the clock is 1 and (3) a single $b$
must occur when the clock is in $(0,1)$.  In case of $L''_5$, (1) the
clock is reset at every occurrence of a $b$ (2) a $b$ must occur when
the clock is less than 1 except for the first $b$ and (3) a single $a$
must occur before every occurrence of a $b$.  \qed
\end{proof}

\section{Combining ITA with CRTA}\label{sec:combination}

In the previous section, we proved that the class of languages defined
by ITA and CRTA are incomparable.  Here we provide a class containing
\MS{Reviewer 2, Rq 10. Done.}
both ITL and CRTL.  In order to do so, we combine the models of ITA
with CRTA.

\subsection{An undecidable product}\label{subsec:mctctlundec}


The first kind of combination possible is through synchronized product
between an ITA and a CRTA.  However, this turns out to be a too
powerful model, since combining even a TA with an ITA yields the
undecidability of the reachability problem.

\begin{definition}
If $\I=\langle\Sigma, Q_\I, q_0^\I, F_\I,pol_\I, X, \lambda_\I, \Delta_\I\rangle$ is an ITA (propositional variables and labeling are omitted) and $\T=\langle\Sigma, Q_\T, q_0^T, F_\T, Y,\Delta_\T\rangle$ is a TA, then $\I\times\T = \langle\Sigma, Q_\I\times Q_\T, (q_0^\I,q_0^\T), F,pol, X, Y, \lambda, \Delta\rangle$ is an ITA$\times$TA where:
\begin{itemize}
\item $pol(q_\I,q_\T) =  pol_\I(q_\I)$ and $\lambda(q_\I,q_\T) =  \lambda_\I(q_\I)$ are lifted from the ITA
\item if $q_\I \xrightarrow{\varphi,a,u} q_\I' \in \Delta_\I$ and $q_\T \xrightarrow{\psi,a,v} q_\T' \in\Delta_\T$, then \[(q_\I,q_\T) \xrightarrow{\varphi\wedge\psi,a,u\wedge v} (q_\I',q_\T') \in \Delta.\]
\end{itemize}
\end{definition}
The semantics of an ITA$\times$TA is a transition system over configurations \[\left\{(q,v,w,\beta) \mid q\in Q, v \in \R^X, w \in \R^Y, \beta \in\{\top,\bot\}\right\}.\]
Discrete steps are defined analogously as in ITA (see Definition~\ref{def:semantics}).
In time steps, clocks of $X$ evolve as in an ITA and clocks of $Y$ as in a TA.
More precisely, a time step of duration $d>0$ is defined by $(q,v,w,\beta) \xrightarrow{d} (q,v',w',\top)$ where $v'(x_{\lambda(q)})=v(x_{\lambda(q)})+ d$ and $v'(x)=v(x)$ for any other clock $x \in X$, and $w'(y)=w(y)+d$ for $y\in Y$.

\begin{theorem}\label{prop:reachundec}
  Reachability is undecidable in the class ITA$\times$TA.
\end{theorem}

\begin{proof}[Sketch]
  The proof consists in encoding a two counter machine into an
  ITA$\times$TA.  Two classical clocks $\{y_c, y_d\}$ will keep the
  value of the counters by retaining a value $1-\frac1{2^n}$ to encode
  $n$.  Three interrupt clocks are used to change the value of the
  classical clocks through appropriate resets.  The ITA$\times$TA
  is defined through basic modules, corresponding to the four possible
  actions (incrementation or decrementation of $c$ or $d$).  Each
  module is itself composed of submodules: the first one compares the
  value of $c$ to the one of $d$.  The other one performs the action,
  but depends on the order between $c$ and $d$.

  For example, the submodule incrementing $c$ when $c \geq d$ is
  depicted in \figurename~\ref{fig:modincremc}.  In this module, the
  value\footnote{Or rather the complement to $1$ of the value.} of
  classical clocks is copied into interrupt clocks, updated thanks to
  linear updates allowed by ITA.  the new values are copied into
  classical clocks by resetting them at the appropriate moment.  The
  valuations of clocks during an execution of this module are given in
  \tablename~\ref{tab:valhorloges}.

Note that the policies are used in this product but they could be replaced by classical clocks.
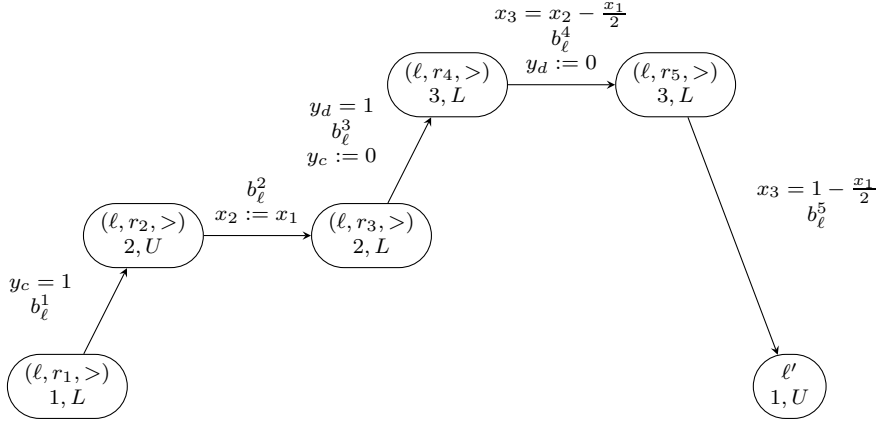
\begin{figure}
\centering
\begin{tikzpicture}[node distance=4.5cm,initial text=,auto,scale=0.5]
\tikzstyle{every state}=[draw=black,fill=white,inner xsep=-4pt,font=\small,shape=rounded rectangle]




\node[state] (q1) at (0,1)            {\etat{(\ell,r_1,>)}{1,L}};
\node[state] (q2) at (2,5) {\etat{(\ell,r_2,>)}{2,U}};
\node[state] (q3) at (8,5)       {\etat{(\ell,r_3,>)}{2,L}};
\node[state] (q4) at (10,9) {\etat{(\ell,r_4,>)}{3,L}};
\node[state] (q5) at (16,9)       {\etat{(\ell,r_5,>)}{3,L}};
\node[state] (q6) at (19,1)           {\etat{\ell'}{1,U}};

\path[->] (q1) edge node [pos=0.3] {\timedtransnoreset{y_c = 1}{b_\ell^1}} (q2);
\path[->] (q2) edge node {\timedtrans{}{b_\ell^2}{x_2:=x_1}} (q3);
\path[->] (q3) edge node [pos=0.3] {\timedtrans{y_d = 1}{b_\ell^3}{y_c := 0}} (q4);
\path[->] (q4) edge node {\timedtrans{x_3 = x_2 - \frac{x_1}{2}}{b_\ell^4}{y_d:= 0}} (q5);
\path[->] (q5) edge node {\timedtransnoreset{x_3 = 1 - \frac{x_1}{2}}{b_\ell^5}} (q6);
\end{tikzpicture}
\caption{Module $\mathcal{A}^{c++}_{c \geq d}(\ell,\ell')$
  incrementing the value of $c$ when $c \geq d$.}
\label{fig:modincremc}
\end{figure}
\begin{table}
\[
\begin{array}{|c|}
\hline
\mathbf{(\ell,r_1,>)}\\
\hline
y_c = 1-\frac{1}{2^n}\\
y_d = 1-\frac{1}{2^p}\\
x_1 = 0\\
\color{black!50}{x_2 = 0}\\
\color{black!50}{x_3 = 0}\\
\hline
\end{array}
\stackrel{\frac{1}{2^n}}{\longrightarrow}
\begin{array}{|c|}
\hline
\mathbf{(\ell,r_1,>)}\\
\hline
y_c = 1\\
y_d = 1-\frac{1}{2^p}+\frac{1}{2^n}\\
x_1 = \frac{1}{2^n}\\
\color{black!50}{x_2 = 0}\\
\color{black!50}{x_3 = 0}\\
\hline
\end{array}
\stackrel{b_\ell^1}{\longrightarrow}
\begin{array}{|c|}
\hline
\mathbf{(\ell,r_2,>)}\\
\hline
y_c = 1\\
y_d = 1-\frac{1}{2^p}+\frac{1}{2^n}\\
x_1 = \frac{1}{2^n}\\
x_2 = 0\\
\color{black!50}{x_3 = 0}\\
\hline
\end{array}
\stackrel{b_\ell^2}{\longrightarrow}
\begin{array}{|c|}
\hline
\mathbf{(\ell,r_3,>)}\\
\hline
y_c = 1\\
y_d = 1-\frac{1}{2^p}+\frac{1}{2^n}\\
x_1 = \frac{1}{2^n}\\
x_2 = \frac{1}{2^n}\\
\color{black!50}{x_3 = 0}\\
\hline
\end{array}
\]\[
\xrightarrow{\frac{1}{2^p} - \frac{1}{2^n}}
\begin{array}{|c|}
\hline
\mathbf{(\ell,r_3,>)}\\
\hline
y_c = 1 + \frac{1}{2^p} - \frac{1}{2^n}\\
y_d = 1\\
x_1 = \frac{1}{2^n}\\
x_2 = \frac{1}{2^p}\\
\color{black!50}{x_3 = 0}\\
\hline
\end{array}
\stackrel{b_\ell^3}{\longrightarrow}
\begin{array}{|c|}
\hline
\mathbf{(\ell,r_4,>)}\\
\hline
y_c = 0\\
y_d = 1\\
x_1 = \frac{1}{2^n}\\
x_2 = \frac{1}{2^p}\\
x_3 = 0\\
\hline
\end{array}
\xrightarrow{\frac{1}{2^p} - \frac{1}{2^{n+1}}}
\begin{array}{|c|}
\hline
\mathbf{(\ell,r_4,>)}\\
\hline
y_c = \frac{1}{2^p} - \frac{1}{2^{n+1}}\\
y_d = 1 + \frac{1}{2^p} - \frac{1}{2^{n+1}}\\
x_1 = \frac{1}{2^n}\\
x_2 = \frac{1}{2^p}\\
x_3 = \frac{1}{2^p} - \frac{1}{2^{n+1}}\\
\hline
\end{array}
\]\[
\stackrel{b_\ell^4}{\longrightarrow}
\begin{array}{|c|}
\hline
\mathbf{(\ell,r_5,>)}\\
\hline
y_c = \frac{1}{2^p} - \frac{1}{2^{n+1}}\\
y_d = 0\\
x_1 = \frac{1}{2^n}\\
x_2 = \frac{1}{2^p}\\
x_3 = \frac{1}{2^p} - \frac{1}{2^{n+1}}\\
\hline
\end{array}
\stackrel{1 - \frac{1}{2^p}}{\longrightarrow}
\begin{array}{|c|}
\hline
\mathbf{(\ell,r_5,>)}\\
\hline
y_c = 1 -\frac{1}{2^{n+1}}\\
y_d = 1 - \frac{1}{2^p}\\
x_1 = \frac{1}{2^n}\\
x_2 = \frac{1}{2^p}\\
x_3 = 1 - \frac{1}{2^{n+1}}\\
\hline
\end{array}
\stackrel{b_\ell^5}{\longrightarrow}
\begin{array}{|c|}
\hline
\mathbf{\ell'}\\
\hline
y_c = 1 -\frac{1}{2^{n+1}}\\
y_d = 1 - \frac{1}{2^p}\\
x_1 = \frac{1}{2^n}\\
\color{black!50}{x_2 = 0}\\
\color{black!50}{x_3 = 0}\\
\hline
\end{array}
\]
\caption[Clock values in the unique run of $\mathcal{A}^{c++}_{c \geq
d}(\ell,\ell')$]{Clock values in the unique run of $\mathcal{A}^{c++}_{c \geq
d}(\ell,\ell')$. Irrelevant values of interrupt clocks are greyed.}
\label{tab:valhorloges}
\end{table}

The detailed proof can be found in Appendix~\ref{app:proofreachundec}.
\end{proof}

Other proofs of undecidability for hybrid systems mixing clocks and
stopwatches have been developed (see for
instance~\cite[Theorem 4.1]{hen98} for a construction with a single
stopwatch and $5$ clocks). While this construction could have been
adapted to our setting, this would have led to an ITA$\times$TA with $5$
classical clocks and $2$ interrupt clocks.

\subsection{A decidable product of ITA and CRTA: ITA$^+$}
\label{sec:comb}
We define another synchronized product between ITA and CRTA, in the
spirit of multi-level systems, for which reachability is decidable.
This class, denoted by ITA$^+$, includes a set of clocks at an
implicit additional level $0$, corresponding to a basic task described
as in a CRTA. In the definition below, since no confusion can occur,
we aggregate the coloring function of CRTA and the level function of
ITA, into a single function $\lambda$.

\begin{definition}[ITA$^+$]
An \emph{extended interrupt timed automaton} is a tuple $\A=\langle Q, q_0,$
$F, pol, X\uplus Y, \Sigma, \Omega, \lambda, up, low, vel, \Delta\rangle$,
where:
\begin{itemize}
\item $Q$ is a finite set of states, $q_0$ is the initial state and $F
\subseteq Q$ is the set of final states.
\item $pol: Q \mapsto \{L,U,D\}$ is the timing policy of states.
\item   $X=\{x_1, \ldots, x_n\}$ consists of $n$ interrupt clocks and  
$Y$ is a set of basic clocks,
\item $\Sigma$ is a finite alphabet,
\item $\Omega$ is a set of colors, the mapping $\lambda : Q \uplus Y
  \mapsto \{1, \ldots, n\} \uplus \Omega$ associates with each state
  its level or its color, with $x_{\lambda(q)}$ the active clock in
  state $q$ for $\lambda(q)\in \N$ and $\lambda(y) \in \Omega$ for $y
  \in Y$. For every state $q \in \lambda^{-1}(\Omega)$, the policy is
  $pol(q)=L$.
\item $up$ and $low$ are mappings from $Y$ to $\Q$ with the same
constraints as CRTA (see Definition~\ref{def:crta}), and $vel:
Q \mapsto \Q$ is the clock rate with $\lambda(q) \notin \Omega
\Rightarrow vel(q)=1$
\item $\Delta \subseteq Q \times [\C(X\cup Y) \times (\Sigma \cup
\{\eps\}) \times \U(X\cup Y)] \times Q$ is the set of transitions.
Let $q \tr{\fee, a, u} q'$ in $\Delta$ be a transition.
\begin{enumerate}
\item The guard $\fee$ is of the form $\fee_1 \wedge\fee_2$ with the
 following conditions.  If $\lambda(q)\in \N$, $\fee_1$ is an ITA
 guard on $X$ and otherwise $\fee_1=true$. Constraint $\fee_2$ is a
 CRTA guard on $Y$ (also possibly equal to $true$).
\item The update $u$ is of the form $u_1 \wedge u_2$ fullfilling the
 following conditions. Assignments from $u_1$ update the clocks in
 $X$ with the constraints of ITA when $\lambda(q)$ and $\lambda(q')$
 belong to $\N$.  Otherwise it is a global reset of clocks in $X$.
 Assignments from $u_2$ update clocks from $Y$, like in CRTA.
\end{enumerate}

\end{itemize}
\end{definition}

Any ITA can be viewed as an ITA$^+$ with $Y$ empty and $\lambda(Q)
\subseteq \{1, \ldots, n\}$, and any CRTA can be viewed as an ITA$^+$
with $X$ empty and $\lambda(Q) \subseteq \Omega$. Class ITA$^+$
combines both models in the following sense. When the current state
$q$ is such that $\lambda(q) \in \Omega$, the ITA part is inactive.
Otherwise, it behaves as an ITA but with additional constraints about
clocks of the CRTA involved by the extended guards and updates.  The
semantics of ITA$^+$ is defined as usual but now takes into account
the velocity of CRTA clocks.

\begin{definition}[Semantics of ITA$^+$]
The semantics of an automaton $\A$ in ITA$^+$ is defined by the
transition system $\T_{\A}= (S, s_0, \rightarrow)$. The set $S$ of
configurations is $\left\{ (q,v) \mid q \in Q, \ v \in \R^{X\cup Y}, \ \beta \in
    \{\top,\bot\} \right\}$,
with initial configuration $(q_0, \vect{0}, \bot)$. An accepting
configuration of $\T_{\A}$ is a pair $(q,v)$ with $q$ in $F$. The
relation $\rightarrow$ on $S$ consists of time steps and discrete
steps, the definition of the latter being the same as before:
\begin{description}
\item[Time steps:] Only the active clocks in a state can evolve, all
other clocks are suspended.  For a state $q$ with $\lambda(q) \in
\N$ (the active clock is $x_{\lambda(q)}$), a time step of duration
$d>0$ is defined by $(q,v,\beta) \tr{d} (q, v',\top)$ with
$v'(x_{\lambda(q)})=v(x_{\lambda(q)})+ d$ and $v'(x)=v(x)$ for any
other clock $x$.  For a state $q$ with $\lambda(q) \in \Omega$ (the
active clocks are $Y'=Y \cap \lambda^{-1}(\lambda(q))$), a time step
of duration $d>0$ is defined by $(q,v,\beta) \tr{d} (q, v',\top)$ with
$v'(y)=v(y)+ vel(q)d$ for $y \in Y'$ and $v'(x)=v(x)$ for any other
clock $x$. In all states, time steps of duration $d=0$ leave the system $\T_{\A}$ in the same configuration. When $pol(q)=U$, only time steps of duration $0$ $q$ are allowed.
\item[Discrete steps:] A discrete step $(q, v) \tr{a} (q', v')$ occurs
if there exists a transition $q \tr {\fee, a, u} q'$ in $\Delta$
such that $v \models \fee$ and $v' = v[u]$. When $pol(q)=D$ and $\beta=\bot$, discrete
steps are forbidden.
\end{description}
\end{definition}

In order to illustrate the interest of the combined models, an example
of a (simple) login procedure is described in \figurename~\ref{fig:excomb}
as a TA with interruptions at a single level.


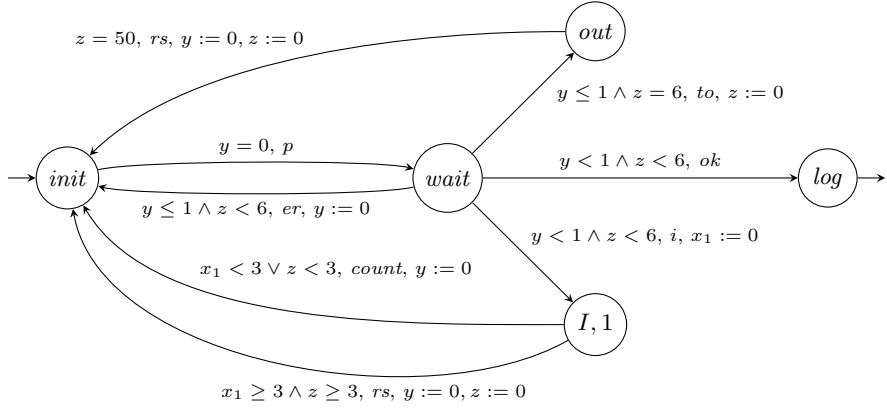
\begin{figure}
\centering
\begin{tikzpicture}[auto,node distance=2.75cm]
\tikzstyle{every node}+=[font=\scriptsize]
\tikzstyle{every state}+=[font=\normalsize]

\node[state] (wait) at (0,0) {\emph{wait}};
\node[state,initial,node distance=5cm,left of=wait] (init) {\emph{init}};
\node[state,accepting,node distance=5cm,right of=wait] (log) {\emph{log}};
\node[state,above right of=wait] (out) {\emph{out}};
\node[state,below right of=wait] (i1) {$I,1$};

\path[->] (init) edge [out=15,in=165,looseness=0.3] node {$y=0$, $p$} (wait);
\path[->] (wait) edge [out=-165,in=-15,looseness=0.3] node {$y \leq 1 \wedge z < 6$, \emph{er}, $y:=0$} (init);
\path[->] (wait) edge node {$y < 1 \wedge z < 6$, \emph{ok}} (log);
\path[->] (wait) edge node [swap,pos=0.75] {$y \leq 1 \wedge z=6$, \emph{to}, $z:=0$} (out);
\path[->] (wait) edge node {$y < 1 \wedge z < 6$, $i$, $x_1:=0$} (i1);
\path[->] (out) edge [bend right,looseness=0.75,out=-15] node [swap] {$z=50$, \emph{rs}, $y:=0, z:=0$} (init);
\path[->] (i1) edge [out=180,in=-60,looseness=0.75] node [swap,pos=0.7] {$x_1 < 3 \vee z < 3$, \emph{count}, $y:=0$} (init);
\path[->] (i1) edge [out=-150,in=-80,looseness=0.8] node [pos=0.35] {$x_1 \geq 3 \wedge z \geq 3$, \emph{rs}, $y:=0, z:=0$} (init);
\end{tikzpicture}
\caption{An automaton for login in ITA$^+$}
\label{fig:excomb}
\end{figure}

First it immediately displays a prompt and arms a time-out of $1$
t.u. handled by clock $y$ (transition $init \tr{p} wait$).  Then
either the user answers correctly within this delay (transition $wait
\tr{ok} log$) or he answers incorrectly or let time elapse, both
cases with transition $wait \tr{er} init$, and the system prompts
again.  The whole process is controlled by a global time-out of $6$
t.u. (transition $wait \tr{to} out$) followed by a long suspension
($50$ t.u.) before reinitializing the process (transition $out \tr{rs}
init$).  Both delays are handled by clock $z$.  At any time during the
process (in fact in state $wait$), a system interrupt may occur
(transition $wait \tr{i} I$).  If the time spent (measured by clock
$x_1$) during the interrupt is less than $3$ t.u. or the time already
spent by the user is less than $3$ t.u., the login process resumes
(transition $I \tr{cont} init$). Otherwise the login process is
reinitialized allowing again the $6$ t.u. (transition $I \tr{rs}
init$).  In both cases, the prompt will be displayed again.  Since
invariants are irrelevant for the reachability problem we did not
include them in the models. Of course, in this example state $wait$
should have invariant $y \leq 1 \wedge z \leq 6$ and state $out$
should have invariant $z\leq 50$.

\bigskip
We extend the decidability and complexity results of the previous
models when combining them with CRTA. Class ITA$^+_-$ is obtained in a
similar way by combining ITA$_-$ with CRTA.
\begin{proposition}\label{prop:itaplusmoins}
\begin{enumerate}[label=\arabic*.]
\item The reachability problem for ITA$^+_-$ is decidable in
  \emph{NEXPTIME} and is \emph{PSPACE}-complete when the number of
  interrupt clocks is fixed.
\item The reachability problem for ITA$^+$ is decidable in
  \emph{NEXPTIME} and is \emph{PSPACE}-complete when the number of
  interrupt clocks is fixed.
\end{enumerate}
\end{proposition}
\begin{proof}~
\vspace{-\baselineskip}
\paragraph{Case of ITA$^+_-$.}
Let $\A= \langle Q, q_0, F, pol, X\uplus Y, \Sigma, \Omega, \lambda,
up, low, vel, \Delta \rangle$ be an ITA$^+_-$, with $n = |X|$ the
number of ITA clocks, $p= |Y|$ the number of CRTA clocks and
$E=|\Delta|$ the number of transitions.

We first consider the reachability problem for two states $q_i$ and
$q_f$ on the CRTA level (with $\lambda(q_i) \in \Omega$ and
$\lambda(q_f) \in \Omega$).  The procedure consists in performing a
non deterministic search along an elementary path where the vertices
are graph classes of the CRTA. Let $(q,Z)$ be the current class, the
procedure chooses non deterministically the next class $(q',Z')$ and
checks that there exists a configuration of $(q,Z)$ and an execution
only through states $q''$ with $\lambda(q'') \in \N$ that leads to a
configuration of $(q',Z')$.  This is solved as previously by non
deterministically choosing an execution path, building a linear
program related to the path (of exponential size) and solving it. Let
us prove that such a path can be chosen whose length is in
$O(p(E+2n)^{3n})$.

\bigskip Assume that there is a run $\pi$ from $(q,v) \in (q,Z)$ to
some configuration $(q',v') \in (q',Z')$ such that all intermediate
states $q''$ are such that $\lambda(q'')\in \N$.  We say that a
transition $e$ of $\pi$ \emph{usefully resets} a clock $y \in Y$ if it
is the first transition of $\pi$ that resets $y$. Observe that there
are at most $p$ useful resetting transitions and that between two such
successive transitions (or before the first one or after the last one)
the value of the clocks of $Y$ are unchanged when transitions are
fired.

We consider a subrun $\rho$ between two such successive transitions
(or before the first one or after the last one) from $(q_1,v_1)$ to
$(q_2,v_2)$, with $m_k$ the number of transitions of level $k$. 

Using Lemma~\ref{lemma:counting}, 
we build a subrun $\rho'$ from $(q_1,v_1)$ to $(q_2,v_2)$ 
of length smaller than $(E+2n)^{3n}$.
Concatenating the subruns, the useful resetting transitions
and the initial transition,
one obtains a run $\pi'$
from $(q,v)$ to $(q',v')$ of length in $O(p(E+2n)^{3n})$. 

The key point ensuring correctness of the procedure is that the
existence of a solution depends only on the starting class $(q,Z)$
and not on the configuration inside this class.  This is due to the
separation of guards and updates between the two kinds of clocks on
the transitions.

When state $q_i$ (resp. $q_f$) is not at the basis level, the
procedure adds an initial (resp. final) guess also checked by a
linear program. When the number of clocks is fixed the dominant
factor is the path search in the class graph 
and PSPACE hardness follows from the result in TA.

\paragraph{Case of ITA$^+$.}
We transform the ITA part of the automaton in ITA$_-$ \emph{via} the
procedure of proposition~\ref{proposition:itaitamoins} and apply the
procedure for ITA$^+_-$.  \qed
\end{proof}

It is also possible to build a class graph for $ITA^+$, combining a class graph for ITA and a region graph for TA.
This yields the regularity of the untimed language of an ITA$^+$, hence the strict inclusion in the languages accepted by a stopwatch automaton.

Let ITL$^+$ be the family of timed languages defined by ITA$^+$.
The class ITL$^+$ syntactically contains ITL$\,\cup\,$CRTL.
We can however have a stronger result:
\begin{proposition}
The class ITL$^+$ \emph{strictly} contains ITL$\,\cup\,$CRTL.
\end{proposition}

\begin{proof}
  Recall ITA $\A_4$ of \figurename~\ref{fig:ctrex2}, whose language
  $L_4$ is not in CRTL, and let $Q_4$ be its set of states.  Also
  recall TA $\A_5$ of \figurename~\ref{fig:ctrex3}, whose language
  $L_5$ is not in ITL, with set of states $Q_5$.  Let
  $\A_4\otimes\A_5$ be the ITA$^+$ having $\A_5$ at level $0$ and
  $\A_4$ at levels $1$ and $2$.

  Formally, $\A_4\otimes\A_5$ has set of states $Q_4 \cup Q_5$, which
  are all lazy.  Interrupt clocks of $\A_4\otimes\A_5$ are
  $\{x_1,x_2\}$ (active according to $\A_4$).  Its basic clocks are
  $\{z,y\}$ of velocity $1$.  Both have the same color as states of
  $Q_5$. The bounding functions $up$ (resp. $low$) map both $z$ and
  $y$ to $1$ (resp. $0$).  Transitions of $\A_4\otimes\A_5$ are the
  ones of $\A_4$ and $\A_5$, adding an unguarded, unlabeled transition
  from $\A_5$'s final state to $\A_4$'s initial one.

  $\A_4\otimes\A_5$ accepts timed words which start with an
  alternation of $a$s and $b$s, with the $b$ drawing always closer to
  its preceding $a$ (as in $\A_5$), and then contains only $c$s
  separated by the same amount of time (as in $\A_4$).  Since both
  CRTL and ITL are closed under projection, $\La(\A_4\otimes\A_5)$
  cannot be accepted by a CRTA nor an ITA.  \qed
\end{proof}

\section{Conclusion}

In this paper, we introduced and studied the model of Interrupt Timed
Automata.  This model is useful to represent timed systems with tasks
organized over priority levels.

\begin{figure}[h]
\centering
\begin{tikzpicture}
\path[thick,draw,fill=black!10] (0,-0.2) ellipse (4.125cm and 2.5cm);
\node at (0,1.9) {LHA = SWA};
\path[thick,draw,fill=white] (0,-0.5) ellipse (3.75cm and 2cm);
\node at (0,-2) {ITA$^+$ = ITA$^+_-$};
\path[thick,draw,xshift=-1cm,yshift=-0.35cm,rotate=-30] (0,0) ellipse (2cm and 1cm);
\node at (-2,0.5) {CRTA};
\path[thick, draw, xshift=-0.5cm,yshift=-0.6cm,rotate=-30] (0,0) ellipse (1.25cm and 0.625cm);
\node at (-1,-0.25) {TA};
\path[thick,draw,xshift=1cm,yshift=-0.35cm,rotate=30] (0,0) ellipse (2cm and 1cm);
\node at (1.75,0.25) {ITA = ITA$_-$};

\node[anchor=west,draw,dashed,inner sep=5pt] at (4.45,-0.25) {\parbox{2.675cm}{%
\raggedright
\begin{description}[itemsep=0.5em,topsep=0pt,parsep=0pt,leftmargin=1.3cm,style=nextline,font=\bfseries]
\item[LHA] Linear Hybrid Automata
\item[SWA] Stopwatch Automata
\item[CRTA] Controlled Real-Time Automata
\item[TA] Timed Automata
\item[ITA] Interrupt Timed Automata
\item[ITA$^+$] Extended ITA
\end{description}
}};
\end{tikzpicture}
\caption{Expressiveness of several timed formalisms with respect to timed languages.}
\label{fig:expressiveness}
\end{figure}
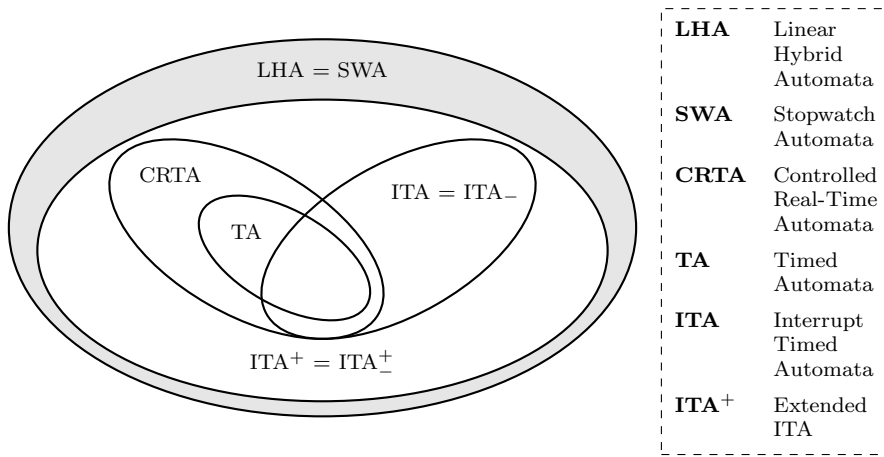

While ITA fall into the more general class of hybrid systems, the
reachability problem is proved decidable for this subclass.  For ITA,
the reachability is in NEXPTIME, and PTIME when the number of clocks
is fixed by building a class graph.  Similar constructions yield
decidability of the reachability problem on an extension of ITA where
the lowest priority level can behave as a Controlled Real-Timed
Automata.  It also yields procedure for model checking \ctlstar
formulas and timed \ctl formulas constraining only the clocks of the
system. Another fragment of interest was identified in timed \ctl as
decidable: the one where the only time constraints concern global
earliest or latest execution times. On the other hand, model checking
the linear time logic \scl is proved undecidable on ITA, implying that
this is also the case for \mitl.

On the expressiveness point of view, the class ITL is proved
incomparable with both TL and CRTL, and is neither closed under
complementation nor intersection. 
\MS{Reviewer 2, Rq 11. Done.}The expressiveness results are
summed up in \figurename~\ref{fig:expressiveness}, where the grey zone
represents undecidability of the reachability problem.

\MS{Reviewer 1, Rqs 3, 4, 7, 8, 13, 18; reviewer 2 long commentaires 2
  et 3. \`A traiter dans la conclu. + open issue \scl without policy +
  hierarchie sur le nb d'horloges.}  Several problems remain open on
the class of ITA.  First of all, the effect of having both (limited)
stopwatches and linear expressions in guards is combined in ITA, and
it is not known which is the cause of the undecidability results
presented in this paper.  For instance, the undecidability of \scl may
not hold without the possibility of complex updates. More generally,
the expressive power of the subclass of ITA restricted with
rectangular guards ($x + b \rel 0$) and only resets ($x:= 0$) should be
investigated. Also, it is conjectured that the class of ITA with $n+1$
clocks is strictly more expressive than the class of ITA with $n$
clocks.  Regarding model-checking, the undecidability of full \tctl
remains to be established. Finally, complexity bounds presented in
this paper are only upper-bounds, and matching lower-bounds are still
missing.

\paragraph{Acknowledgments}
The authors would like to thank the anonymous reviewers for their insightful comments.  This work was supported by
projects \textsc{Dots} (ANR-06-SETI-003, French government), \textsc{ImpRo} (ANR-2010-BLAN-0317, French government) and
\textsc{CoChaT} (Digiteo 2009-27HD, R\'egion \^Ile de France).

\bibliographystyle{splncs2}
\bibliography{journalITA}

\newpage
\appendix
\section{Proof of Theorem~\ref{thm:mctctlintdec}}\label{app:proofmctctlindec}

Let $\varphi$ be a formula in
\textsf{TCTL}$_c^{\textrm{int}}$ and $\mathcal{A}$ an ITA with $n$
levels and $E$ transitions. Like in Section~\ref{sec:regular}, the proof relies on the
construction of a finite class graph. The main difference is in the
computation of the $n$ sets of expressions $E_1, \dots, E_n$. Like
before, each set $E_k$ is initialized to $\{x_k, 0\}$ and expressions
in this set are those which are relevant for comparisons with the
current clock at level $k$. In this case, they include not only guards
but also comparisons with the constraints from the formula. Recall
that the sets are computed top down from $n$ to $1$, using the
normalization operation.
\begin{itemize}
\item At level $k$, we may assume that expressions in guards of an
  edge leaving a state are of the form $\alpha x_k+\sum_{i<k}a_ix_i+b$
  with $\alpha \in \{0,1\}$. We add $-\sum_{i<k}a_ix_i-b$ to $E_k$.
\item To take into account the constraints of formula $\varphi$, we
  add the following step: For each comparison $C \rel 0$ in $\varphi$,
  and for each $k$, with $\mathtt{norm}(C,k) = \alpha
  x_k+\sum_{i<k}a_ix_i+b$ ($\alpha \in \{0,1\}$), we also add expression
  $-\sum_{i<k} a_ix_i - b$ to $E_k$.
\item Then we iterate the following procedure until no new term is
added to any $E_i$ for $1\leq i\leq k$.
	\begin{enumerate}
      \item Let $q \tr{\fee,a,u} q'$ with $\lambda(q')\geq k$ and
        $\lambda(q)\geq k$. If $C \in E_{k}$, then we add $C[u]$ to
        $E_{k}$.
      \item Let $q \tr{\fee,a,u} q'$ with $\lambda(q') \geq k$ and
        $\lambda(q) < k$. For $C,C' \in E_{k}$, we compute
        $C''=\mathtt{norm}(C[u]-C'[u],\lambda(q))$. If $C''= \alpha
        x_{\lambda(q)}+\sum_{i<\lambda(q)}a_ix_i+b$ with $\alpha \in
        \{0,1\}$, then we add $-\sum_{i<\lambda(q)}a_ix_i-b$ to
        $E_{\lambda(q)}$.
	\end{enumerate}
\end{itemize}
The proof of termination for this construction is similar to the one
in Section~\ref{sec:regular}.

We now consider the transition system $\G_{\A}$ whose set of
configurations are the classes $R = (q,\{\preceq_k\}_{1 \leq k \leq
  \lambda(q)})$, where $q$ is a state and $\preceq_k$ is a total
preorder over $E_k$.  The class $R$ describes the set of valuations
$\sem{R}=\{(q,v) \mid \forall k \leq \lambda(q)\ \forall (g,h) \in
E_k,\ g[v] \leq h[v]$ iff $g \preceq_k h\}$.  The set of transitions
is defined as in Section~\ref{sec:regular}. The transition system
$\G_{\A}$ is again finite and time abstract bisimilar to $\T_{\A}$.
Moreover, the truth value of each comparison $C = \sum_{i \geq 1} a_i
\cdot x_i + b \rel 0$ appearing in $\varphi$ can be set for each class
$R$.  Indeed, since for every $k$, both $0$ and $\sum_{i \geq 1}^{k-1}
a_i \cdot x_i + b$ are in the set of expressions $E_k$, the truth
value of $C \rel 0$ does not change inside a class.  Therefore,
introducing a fresh propositional variable $q_C$ for the constraint $C
\rel 0$, each class $R$ can be labelled with a truth value for each
$q_C$.  Deciding the truth value of $\varphi$ can then be done by a
classical \textsf{CTL} model-checking algorithm on $\G_{\A}$.

The complexity of the procedure is obtained by bounding the number of
expressions for each level $k$ by
$(E+|\varphi|+2)^{2^{n(n-k+1)}+1}$, and applying the same
reasoning as for proposition~\ref{prop:reachita}.

\section{Proof of Theorem~\ref{prop:reachundec}}\label{app:proofreachundec}

  We build an automaton in ITA$\times$TA which simulates a
  deterministic two counter machine $\M$ (as in proof of Theorem~\ref{thm:mcsclundec}).

Let $L_\M$ be the set of labels of $\M$.
The automaton $\mathcal{A}_{\mathcal{M}} =
\langle\Sigma,Q,q_0,F,pol,X \cup Y,\lambda,\Delta\rangle$ is built
to reach its final location $Halt$ if and only if $\mathcal{M}$
stops. It is defined as follows:
\begin{itemize}
\item $\Sigma$ consists of one letter per transition.
\item \(Q = L_\M \cup (L _\M\times \{k_0\}) \cup (L_\M \times
  \{k_1,k_2,r_1,\dots,r_5\} \times \{>,<\})\), \(q_0 = \ell_0\) (the
  initial instruction of $\mathcal{M}$) and \(F = \{Halt\}\).
  \item \(pol: Q \rightarrow \{Urgent,Lazy,Delayed\}\) is such that
    \(pol(q) = Urgent\) iff either $q \in L_\M$ or \(q =
    (\ell,q_2,\bowtie)\), and \(pol(q) = Lazy\) in most other cases:
    some states \((\ell,k_i,\bowtie)\) are \emph{Delayed}, as shown on
    \figurename~\ref{fig:incremc} and~\ref{fig:decremc}.
  \item \(X = \{x_1,x_2,x_3\}\) is the set of interrupt clocks and \(Y
    = \{y_c,y_d\}\) is the set of standard clocks with rate $1$.
  \item \(\lambda: Q \rightarrow \{1,2,3\}\) is the interrupt level of
    each state.  All states in $L_\M$ and $L_\M \times \{k_0,k_1,k_2\}$ are
    at level $1$; so do all states corresponding to $r_1$.  States
    corresponding to $r_2$ and $r_3$ are in level $2$, while the ones
    corresponding to $r_4$ and $r_5$ are in level $3$.
  \item $\Delta$ is defined through basic modules in the sequel.
\end{itemize}

The transitions of $\mathcal{A}_{\mathcal{M}}$ are built within small
modules, each one corresponding to one instruction of $\mathcal{M}$.
The value $n$ of $c$ (resp. $p$ of $d$) in a state of $L_\M$ is encoded
by the value \(1 - \frac{1}{2^n}\) of clock $y_c$ (resp. \(1 -
\frac{1}{2^p}\) of $y_d$).

The idea behind this construction is that for any standard clock $y$,
it is possible to ``copy'' the value of $k-y$ in an interrupt clock
$x_i$, for some constant $k$, provided the value of $y$ never exceeds
$k$.  To achieve this, we start and reset the interrupt clock, then
stop it when \(y = k\).  Note that by the end of the copy, the value
of $y$ has changed. Conversely, in order to copy the content of an
interrupt clock $x_i$ into a clock $y$, we switch from level $i$ to
level $i+1$ and reset $y$ at the same time.  When $x_{i+1} = x_i$, the
value of $y$ is equal to the value of $x_i$. Remark that the form of
the guards on $x_{i+1}$ allows us to copy the value of a linear
expression on \(\{x_1,\dots,x_i\}\) in $y$.

For instance, consider an instruction labeled by $\ell$ incrementing
$c$ then going to $\ell'$, with the respective values $n$ of $c$ and
$p$ of $d$, from a configuration where $n \geq p$.  The corresponding
module \(\mathcal{A}_{c\geq d}^{c++}(\ell,\ell')\) is depicted on
\figurename~\ref{fig:modincremc} (see main text). In this module, interrupt clock
$x_1$ is used to record the value $\frac{1}{2^n}$ while $x_2$ keeps
the value $\frac{1}{2^p}$.  Assuming that \(y_c = 1 - \frac{1}{2^n}\),
\(y_d = 1 - \frac{1}{2^p}\) and \(x_1 = 0\) in state \((\ell,r_1,>)\),
the unique run in \(\mathcal{A}_{c\geq d}^{c++}(\ell,\ell')\) will end
in state $\ell'$ with \(y_c = 1 - \frac{1}{2^{n+1}}\) and \(y_d = 1 -
\frac{1}{2^p}\).  The intermediate clock values are shown in
\tablename~\ref{tab:valhorloges} (see main text).

The module on \figurename~\ref{fig:modincremc} can be adapted for the
case of decrementing $c$ by just changing the linear expressions in
guards for $x_3$, provided that the final value of $c$ is still greater
than the one of $d$.  It is however also quite easy to adapt the same
module when \(n < p\): in that case we store $\frac{1}{2^p}$ in $x_1$
and $\frac{1}{2^n}$ in $x_2$, since $y_d$ will reach $1$ before $y_c$.
We also need to start $y_d$ before $y_c$ when copying the adequate
values in the clocks. The case of decrementing $c$ while \(n \leq p\)
is handled similarly.  In order to choose which module to use
according to the ordering between the values of the counters, we use
the modules of \figurename~\ref{fig:incremc} and~\ref{fig:decremc}.
\figurename~\ref{fig:incremc}
represents the case when at label $\ell$ we
have an increment of $c$ whereas \figurename~\ref{fig:decremc}
represents the case when $\ell$ corresponds to decrementing $c$.  In
that last case the value of $c$ is compared not only to the one of
$d$, but also to $0$, in order to know which branch of the \emph{if}
instruction is taken.
Note that only one of the branches can be taken until the end\footnote{State policies are used to treat the special cases, \textit{e.g.} $y_c = y_d = 0$.}.
Instructions involving $d$ are handled in a symmetrical way.

\begin{figure}
\centering
\begin{tikzpicture}[node distance=4cm,initial text=,auto,scale=0.5]
\tikzstyle{every node}=[font=\scriptsize]
\tikzstyle{every state}=[draw=black,fill=white,inner xsep=-5pt,font=\small,shape=rounded rectangle]


\node[state] (q0) at (-6,0)  {\etat{\ell}{1,U}};
\node[state] (q7) at (-0.5,0)  {\etat{(\ell,k_0)}{1,L}};
\node[state] (q1) at (16,2.75)  {\etat{(\ell,r_1,>)}{1,L}};
\node[state] (q2) at (16,-2.75) {\etat{(\ell,r_1,<)}{1,L}};
\node[state] (q3) at (4,2.75) {\etat{(\ell,k_1,>)}{1,L}};
\node[state] (q4) at (4,-2.75) {\etat{(\ell,k_1,<)}{1,D}};
\node[state] (q5) at (10,2.75) {\etat{(\ell,k_2,>)}{1,L}};
\node[state] (q6) at (10,-2.75) {\etat{(\ell,k_2,<)}{1,L}};

\path[->] (q0) edge node {\timedtransnoreset{a_{\ell}^0}{x_1:=0}} (q7);
\path[->] (q7) edge node [pos=0.7] {\timedtrans{y_c=1}{a_{\ell,>}^1}{y_c:=0}} (q3);
\path[->] (q7) edge node [pos=0.7,swap] {\timedtrans{y_d=1}{a_{\ell,<}^1}{y_d:=0}} (q4);
\path[->] (q3) edge node {\timedtrans{y_d=1}{a_{\ell,>}^2}{y_d:=0}} (q5);
\path[->] (q4) edge node {\timedtrans{y_c=1}{a_{\ell,<}^2}{y_c:=0}} (q6);
\path[->] (q5) edge node {\timedtrans{x_1=1}{a_{\ell,>}^3}{x_1:=0}} (q1);
\path[->] (q6) edge node {\timedtrans{x_1=1}{a_{\ell,<}^3}{x_1:=0}} (q2);
\end{tikzpicture}
\caption{Module taking into account the order between the values of $c$ and $d$ when incrementing $c$.}
\label{fig:incremc}
\end{figure}
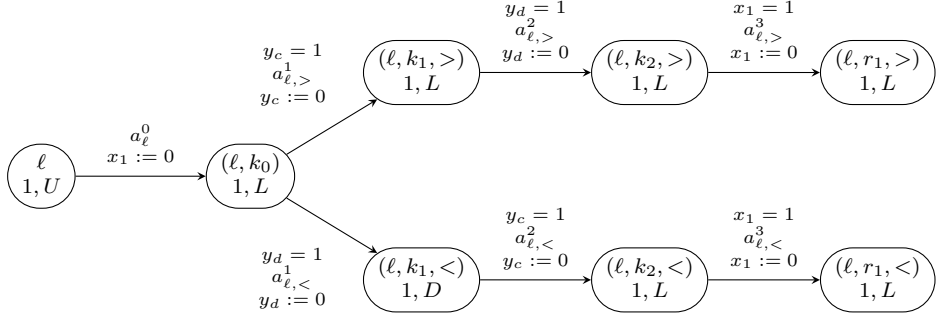
\begin{figure}
\centering
\begin{tikzpicture}[node distance=4cm,initial text=,auto,scale=0.5]
\tikzstyle{every node}=[font=\scriptsize]
\tikzstyle{every state}=[draw=black,fill=white,inner xsep=-5pt,font=\small,shape=rounded rectangle]


\node[state] (q0) at (-6,0)  {\etat{\ell}{1,U}};
\node[state] (q7) at (-0.5,0)  {\etat{(\ell,k_0)}{1,L}};
\node[state] (q1) at (16,2.75)  {\etat{(\ell,r_1,>)}{1,L}};
\node[state] (q2) at (16,-2.75) {\etat{(\ell,r_1,<)}{1,L}};
\node[state] (q3) at (4,2.75) {\etat{(\ell,k_1,>)}{1,D}};
\node[state] (q4) at (4,-2.75) {\etat{(\ell,k_1,<)}{1,L}};
\node[state] (q5) at (10,2.75) {\etat{(\ell,k_2,>)}{1,L}};
\node[state] (q6) at (10,-2.75) {\etat{(\ell,k_2,<)}{1,D}};
\node[state] (q8) at (-6,-3.5)     {\etat{\ell''}{1,U}};

\path[->] (q0) edge node {\timedtrans{y_c>0}{a_{\ell}^0}{x_1:=0}} (q7);
\path[->] (q7) edge node [pos=0.7] {\timedtrans{y_c=1}{a_{\ell,>}^1}{y_c:=0}} (q3);
\path[->] (q7) edge node [pos=0.7,swap] {\timedtrans{y_d=1}{a_{\ell,<}^1}{y_d:=0}} (q4);
\path[->] (q3) edge node {\timedtrans{y_d=1}{a_{\ell,>}^2}{y_d:=0}} (q5);
\path[->] (q4) edge node {\timedtrans{y_c=1}{a_{\ell,<}^2}{y_c:=0}} (q6);
\path[->] (q5) edge node {\timedtrans{x_1=1}{a_{\ell,>}^3}{x_1:=0}} (q1);
\path[->] (q6) edge node {\timedtrans{x_1=1}{a_{\ell,<}^3}{x_1:=0}} (q2);
\path[->] (q0) edge node {\timedtransnoreset{y_c = 0}{a_{\ell,0}}} (q8);
\end{tikzpicture}
\caption{Module taking into account the order between the values of $c$ and $d$ when decrementing $c$.}
\label{fig:decremc}
\end{figure}
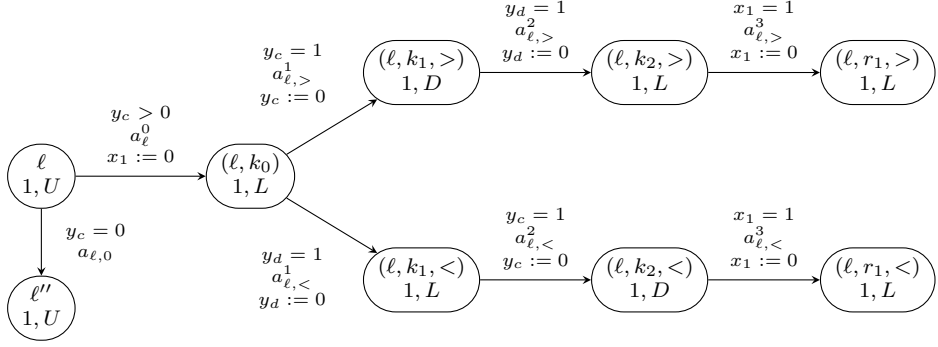

Automaton $\mathcal{A}_\mathcal{M}$ is obtained by joining the modules
described above through the states of $L_\M$.  Let us prove that
automaton \(\mathcal{A}_\mathcal{M}\) simulates the two counter
machine $\mathcal{M}$, so that $\mathcal{M}$ halts iff
\(\mathcal{A}_\mathcal{M}\) reaches the \emph{Halt} state.

Let \(\langle\ell_0,0,0\rangle
\langle\ell_1,n_1,p_1\rangle \dots \langle\ell_i,n_i,p_i\rangle\dots\)
be a run of $\mathcal M$.  We show that this run is simulated in
\(\mathcal{A}_\mathcal{M}\) by the run \(\langle
l_0,\mathbf{0}\rangle\rho_0\langle l_1,v_1\rangle\rho_1\dots\) where
$\rho_i$ is either empty or a subrun through states in
\(\{(\ell_i,r_j,\bowtie) \,|\, j \in \{1,\dots,5\}, \bowtie \in
\{>,<\}\}\) (\emph{i.e.} subruns in modules like
\(\mathcal{A}^{c++}_{c \geq d}\) of \figurename~\ref{fig:modincremc}).
Moreover, it will be the case that \[\forall i,\quad v_i(y_c) = 1 -
\frac{1}{2^{n_i}} \quad \textrm{and} \quad v_i(y_d) = 1 -
\frac{1}{2^{p_i}}\] This holds at the beginning of the execution of
\(\mathcal{A}_\mathcal{M}\).  Suppose that we have simulated the
subrun up to \(\langle\ell_i,n_i,p_i\rangle\).  Then we are in state
$\ell_i$, with clock $y_c$ being \(1 - \frac{1}{2^{n_i}}\) and $y_d$
being \(1 - \frac{1}{2^{p_i}}\).  The next configuration of
$\mathcal{M}$, \(\langle\ell_{i+1},n_{i+1},p_{i+1}\rangle\), depends on
the content of instruction $\ell_i$, and so does the outgoing
transitions of state $\ell_i$ in \(\mathcal{A}_\mathcal{M}\).  We
consider the case where $\ell_i$ decrements $c$ and goes to $\ell'$ if
$c$ is greater than 0 and goes to $\ell''$ otherwise, the other ones
being similar.  We are therefore in the case of
\figurename~\ref{fig:decremc}.  If $n_i = 0$, the next configuration
of $\mathcal{M}$ will be \(\langle\ell'',n_i,p_i\rangle\).
Conversely, in \(\mathcal{A}_\mathcal{M}\), if $n_i = 0$ then $y_c =
0$, and there is no choice but to enter $\ell''$, leaving all clock
values unchanged (because $\ell_i$ is an \emph{Urgent} state).  The
configuration of \(\mathcal{A}_\mathcal{M}\) thus satisfies the
property.  If $n_i > 0$, the next configuration of $\mathcal{M}$
will be \(\langle\ell',n_i -1,p_i\rangle\).  In
\(\mathcal{A}_\mathcal{M}\), the transition chosen is the one that
corresponds to the ordering between $n_i$ and $p_i$.  In both cases,
similarly to the example of $\mathcal{A}^{c++}_{c \geq
  d}(\ell,\ell')$, the run reaches state $\ell'$ with \(y_c =1 -
\frac{1}{2^{n_i-1}}\) and $y_d$ as before, thus preserving the
property.  Hence $\mathcal{M}$ halts iff \(\mathcal{A}_\mathcal{M}\) reaches the
\emph{Halt} state.

The automaton \(\mathcal{A}_\mathcal{M}\) is indeed
the product of an ITA $\mathcal{I}$ and a TA $\mathcal{T}$,
synchronized on actions. Observe that in all the modules described
above, guards never mix a standard clock with an interrupt one.  Since
each transition has a unique label, keeping only guards and resets on
either the clocks of $X$ or on those of $Y$ yields an ITA and a TA
whose product is \(\mathcal{A}_\mathcal{M}\).\qed

\end{document}